\documentclass{article}
\usepackage[utf8]{inputenc}

\usepackage{amsfonts}
\usepackage[margin=1in]{geometry}
\usepackage{hyperref,amsthm,enumerate,
            imakeidx, xparse, mathtools,
            xcolor, amssymb, tensor, 
            soul, graphicx ,titlesec,  appendix, tikz,
            amsmath,scalerel, comment, float, multirow} 
\usepackage[most]{tcolorbox}
\usepackage{stackengine}
\usetikzlibrary{arrows,decorations.markings}
\usetikzlibrary{shapes.misc}
\usetikzlibrary{arrows.meta}
\usetikzlibrary{angles,quotes}
\usepackage{tkz-euclide}
\usetikzlibrary{intersections}
\usetikzlibrary{calc}

\hypersetup{
pdfstartview = {FitH},
}
\hypersetup{
	colorlinks=true,         
	linkcolor=purple,          
	citecolor=red,        
	urlcolor=blue            
}

\newcommand{\be}{\begin{equation}}
\newcommand{\ee}{\end{equation}}

\newtheorem{definition}{Definition}
\newtheorem{theorem}{Theorem}

\newtheorem{proposition}{Proposition}

\newcommand{\cC}{{\cal C}}
\newcommand{\cD}{{\cal D}}

\newcommand{\cT}{{\cal T}}
\newcommand{\cG}{{\cal G}}
\newcommand{\G}{{\cal G}}

\newcommand{\cB}{{\cal B}}
\newcommand{\cJ}{{\cal J}}

\newcommand{\cS}{{\cal S}}
\newcommand{\cP}{{\cal P}}

\def\d{\mathfrak d}
\def\g{\mathfrak g}
\def\b{\mathfrak b}
\def\h{\mathfrak h}

\def\so{\mathfrak{so}}
\def\an{\mathfrak{an}}
\def\su{\mathfrak{su}}

\def\bbR{\mathbb{R}}
\def\R{\mathbb{R}}
\def\bbC{\mathbb{C}}

\def\SU{\mathrm{SU}}
\def\SL{\mathrm{SL}}
\def\SO{\mathrm{SO}}
\def\AN{\mathrm{AN}}
\def\ISO{\mathrm{ISO}}

\def\tell{\tilde \ell}

\def\th{\tilde h}
\def\ty{\tilde y}
\def\tbe{\tilde \beta}

\def\tu{\tilde u}
\def\tla{\tilde \lambda}

\def\bu{\bar u}

\def\oy{\bar y}
\def\obe{\bar \beta}
\def\blam{\bar \lambda}

\def\tblam{\tilde \blam}

\def\tbu{\tilde \bu}
\def\toy{\tilde \oy}
\def\tobe{\tilde \obe}

\def\Dr{\Delta}
\def\dr{\rightarrow}
\def\Dl{\underline \Delta}
\def\rhd{\triangleright}
\def\lhd{\triangleleft}
\def\la{\langle}
\def\ra{\rangle}
\def\ot{\otimes}
\def\op{\oplus}
\def\ip{\lrcorner\,}
\newcommand\blackbowtie{\mathrel{\scalerel*{\blacktriangleright\joinrel\blacktriangleleft}{x}}}
\newcommand\bicrossr{\mathrel{\scalerel*{\blacktriangleright\joinrel\mathrel{\triangleleft}}{x}}}
\newcommand\bicrossl{\mathrel{\scalerel*{\mathrel{\triangleright}\joinrel\blacktriangleleft}{x}}}

\def\poi#1{\{ #1 \}}

\def\demi{\frac{1}{2}}
\def\mone{^{-1}}
\def\kmone{\kappa\mone}

\usepackage{authblk}

\makeatletter
\renewenvironment{abstract}{%
    \if@twocolumn
      \section*{\abstractname}%
    \else 
      \begin{center}%
        {\bfseries\sffamily\abstractname\vspace{\z@}}
      \end{center}%
      \quotation
    \fi}
    {\if@twocolumn\else\endquotation\fi}
\makeatother

\title{Polyhedron phase space using 2-groups: \\   $\kappa$-Poincar\'e as a  Poisson 2-group}

\begin{document}

\author[2]{\sffamily Florian Girelli\thanks{florian.girelli@uwaterloo.ca}}
\author[1]{\sffamily Matteo Laudonio\thanks{matteo.laudonio@u-bordeaux.fr}}
\author[2]{\sffamily Panagiotis Tsimiklis\thanks{ptsimiklis@uwaterloo.ca}}

\affil[1]{\small LABRI, Universit\'e de Bordeaux, 351 Cours de la Libération, 33400 Talence, France}
\affil[2]{\small Department of Applied Mathematics, University of Waterloo, 200 University Avenue West, Waterloo, Ontario, Canada, N2L 3G1}

\maketitle
\begin{abstract}
      We construct a phase space for a three dimensional cellular complex with decorations on edges and faces using crossed modules (strict 2-groups) equipped with a (non-trivial) Poisson structure. We do not use the most general crossed module, but only the ones where the target map ($t$-map) is trivial. As a particular case, we recover that  deformations of the Poincar\'e group can be exported to deformations of the Poincar\'e 2-group. The $\kappa$-deformation case provides a natural candidate for describing discrete  geometries with curved edge decorations (but still flat face decorations).  Our construction generalizes the classical phase space defined in \cite{Asante:2019lki} for  a 3d triangulation in terms of an un-deformed Poincar\'e 2-group.  
\end{abstract}

\maketitle
\tableofcontents

\section{Introduction}
Topological quantum field theories (TQFT) \cite{Atiyah:1989vu} have gathered a lot of interest for many different reasons. From a mathematical point of view they can be used to construct topological invariants \cite{Witten:1989mi, Crane:1994ji}. From a physical perspective, they can be used to construct quantum gravity models. 3d gravity is itself a topological theory \cite{Witten:1988hc}.  4d gravity can be seen as a constrained topological theory \cite{Plebanski:1977zz} so one can constrain the relevant TQFT to build a quantum gravity model \cite{Crane:1993vs, Barrett:1995mg} --- this is the spinfoam approach \cite{rovelli2004}. From the condensed matter perspective, topological order have been an active and fruitful research direction \cite{Wen:1989iv, Wen:2012hm}, with application in quantum information theory \cite{Kitaev1997}. 

TQFT's in 3d are naturally described in terms of monoidal (spherical) categories \cite{Barrett:1993ab}, which appear when dealing with representations of quantum groups or Hopf algebras \cite{Majid:1996kd}. The TQFT is encoded in a state sum, a triangulation/cellular complex decorated by representations of the (quantum) group such as the Turaev-Viro model \cite{TV}. The state sum can be seen as a transition amplitude between initial and final states which are given in terms of triangulation/cellular decomposition of 2d surfaces which bound the complex. These states have sometimes been called \textit{spin networks} in the gravity context. They were first introduced by Penrose \cite{Penrose71angularmomentum:} who intended to reconstruct space(-time) through colliding spinning tops with conserved angular momentum. The spin network states naturally encode  quantum states for the 2d triangulation/cellular complex. In particular an intertwiner encodes the quantum state of a (curved) polygon \cite{Dupuis:2013haa}. We can construct a classical picture behind these quantum states, ie a phase space for the 2d triangulation \cite{Bonzom:2014wva, Bonzom:2014bua}. This phase space can be deduced from the discretization of a BF action, possibly with a non-zero cosmological constant (which specifies then the quantum group deformation)\cite{Dupuis:2020ndx}.

We emphasize that the phase spaces  we consider are more general than the typical cotangent bundle phase space. Indeed, the classical version of a quantum group is a Lie group equipped with a non-trivial Poisson bracket. Quantum group symmetries are classically given by a group of symmetries equipped with non-trivial Poisson bracket. Poisson Lie groups encode  symmetries of phase spaces which have curvature in momentum space. The relevant phase space is then called a \textit{Heisenberg double} \cite{Alekseev_1994, SemenovTianShansky:1993ws, SemenovTianShansky:1993ws}.

Spin networks also naturally appear when quantizing 4d BF theory. In fact, by discretizing 4d BF theory (with no cosmological constant) we recover the same type of phase space we have obtained when dealing with a 2d triangulation. The reason  is simple a posteriori and is encoded in Minkowski's theorem \cite{alexandrov}.  This theorem states that given a (non-planar) polygon  one can construct a polyhedron where the edges of the polygon are interpreted as the area normals of the faces. Hence if we  describe a polyhedron in terms of its faces and have no information regarding the edges, it is essentially the same as a polygon. Bluntly, given a spin network, we actually cannot know whether we are  dealing with a 2d or 3d space due to this equivalence between polygon/polyhedron.    

Describing the quantum states of 3d geometry (for a 4d spacetime/manifold) in terms of spin networks can be viewed as problematic for several reasons. 
From a topological perspective, one expects that the proper way to capture topological features are  through 2-categories \cite{Baez:1995xq, MACKAAY1999288, MACKAAY2000353, Lurie:2009keu}, where both the edges \textit{and} the faces of the triangulation are decorated. These 2-categories should be encoding the representations of a quantum 2-group or Hopf 2-algebra.   While  a definition of such quantum strict 2-group has been proposed in \cite{Majid:2012gy},  there are not so many explicit examples, especially involving Lie groups, to the best of our knowledge. Furthermore, classifications of representations are lacking for most  Lie 2-groups, to the notable exception of the Poincar\'e 2-group \cite{Baez:2008hz}. There is also no generalization of the Peter-Weyl theorem -- which is instrumental in defining quantum states in terms of spin networks --  for general 2-groups.    Quantum 2-groups should have a classical version, given by Poisson 2-groups. Again, while they have been formally defined \cite{chen2013}, to the best of our knowledge there are not many explicit examples. A state sum model using 2-complexes has been built by Yetter \cite{Yetter:1992rz, Yetter:1993dh} using finite groups. Yetter's state sum model was later showed to be a discretization of the transition amplitude characterized by an action with 2-group symmetries called the BFCG action  when dealing with Lie groups \cite{Girelli:2007tt, Martins:2010ry}. The initial and final states associated with this transition amplitude are now given in terms of a categorification of the notion of spin networks, which we would coin \textit{ 2-spin networks}\footnote{They would possibly be  related to the notion of G-networks introduced in \cite{Asante:2019lki}, when dealing with the Poincar\'e 2-group.}. One expects that the classical picture behind these states to be given in terms of a phase space defined in terms of 2-groups. While it is tempting to say that it could be some type of 2-Heisenberg double, such object has not been defined to the best of our knowledge.

If spin networks cannot sufficiently capture topological features of 3d geometries, they may also be lacking as the main tool of a quantum gravity model.
In fact we can argue further that only having access to the face information and no decorations on edges is problematic. After all, edge decorations should be related to some discretization of the  metric degrees of freedom, which we would like to have. Again, putting decorations on both edges and faces naturally points to the concept of 2-group, which in turn naturally points to using a BFCG type action with  constraints to recover gravity and build the gravity model \cite{Mikovic:2011si}. With this mind, some phase space structure associated to a 3d triangulation was introduced in \cite{Asante:2019lki} to relate the discretization of a Poincar\'e 2-group \textit{BFCG} discretization to a state sum model based on 2-representations of the Poincar\'e 2-group. This phase space  is the classical version of 2-spin networks defined for the Poincar\'e 2-group and describes ``flat" triangulation, in the sense that the variables decorating the edges\footnote{In the following, we will use edges and faces or triangles to denote resp. 1d and 2d objects in the triangulation. Links and wedges (sometimes dual faces) will denote resp. 1d and 2d objects in the dual 2-complex.} are element in an abelian group. 

\medskip 

The key question we want to address here is whether we can generalize this construction to curved polyhedra, namely with edge decorations with value in a non-abelian group\footnote{The Minkowski theorem was generalized to the curved case \cite{Haggard:2015ima}, where decorations of the faces are specified by a non-abelian group. However, there is an hidden dependence on some edge variables with value in a non-abelian group. It would be interesting to see how this is consistent with a 2-group picture. We thank A. Riello for the clarification.}. This would imply in particular to have non-trivial Poisson brackets for the decorations associated to the faces  dual to the edges.  Hence, upon quantization this would amount to dealing with a quantum 2-group.

\medskip

We do not solve this problem in full generality, but we provide some new examples based on a specific class of 2-groups, namely \textit{strict} 2-groups\footnote{We will essentially work with crossed modules as strict 2-groups.} with trivial $t$-map, which we will call \textit{skeletal} 2-groups following \cite{Baez:2008hz}. skeletal 2-groups always have abelian group decorations on the faces.  The Poincar\'e 2-group or the (co-)tangent 2-group are natural examples of skeletal 2-groups \cite{Baez:2010ya}. Because the $t$-map is trivial, a skeletal 2-group behaves very much as a group, the main difference being in the geometrical interpretation of the different products of the (sub)groups. This is the technical insights that allows us to circumvent the general definition of a ``2-Heisenberg double" and deal with usual Heisenberg doubles, built from skeletal 2-groups. We can then use symplectic reduction to fuse different phase spaces to construct the general phase space of a 3d triangulation with both decoration on edges and faces. Note that  symplectic reduction is  a natural tool to glue regions containing gauge degrees of freedom, see for example \cite{Riello:2017iti, Riello:2021lfl}.  

An initially unexpected outcome of the construction is that one natural deformation of the Poincar\'e 2-group is actually specified by the $\kappa$-Poincar\'e deformation \cite{Lukierski:1992dt, Majid:1994cy}. Such deformation is a well known deformation of the Poincar\'e group, which generates a non-commutative Minkowski space-time and is used to study quantum gravity phenomenology\cite{AmelinoCamelia:1999pm}. Here it naturally arises when switching curvature on the edge decorations. Once again, due to the similarity between groups and skeletal 2-groups, it is only a matter of interpretation to see the (classical) $\kappa$-Poincar\'e group  as a Poincar\'e 2-group equipped with a non-trivial Poisson bracket. The Heisenberg double we will consider as our main example 
and  will use as our basic building block to build the triangulation phase space,
will be based on the $\kappa$-Poincar\'e (2-)group and its dual (called the $\kappa$-Poincar\'e algebra). 

\medskip

In Section \ref{Sec_MathematicalToolBox}, we provide all the mathematical ingredients we will need to define the phase space for a 3d triangulation. We recall the definitions of Poisson Lie groups, Heisenberg double and the concept of symplectic reduction when the momentum map is with value in a (non-)abelian group. We explain how this symplectic reduction allows to fuse/glue phase spaces. We move then to recall the notion of strict 2-groups and Poisson strict 2-groups. We provide some examples when the strict 2-group is simple, ie the $t$-map is trivial. We show then how the strict skeletal 2-groups can be embedded in a Heisenberg double and focus on the construction of the $\kappa$-Poincar\'e 2-group and its dual. This provides the main example of a non-trivial   Poisson 2-group which can be used to generate non-trivial decorations on the triangulation edges. 

In Section \ref{sec:polygon}, we recall how the notion of a Heisenberg double together with symplectic reduction can be used to build the phase space of a 2d triangulation. We then argue using Minkowski's theorem how this construction can be extended to the 3d case. 

{Section \ref{Sec_2-groupPhaseSpace}  is in some sense the main section of the paper. It details the construction of the phase space for a 3d triangulation in terms of Heisenberg doubles defined in terms of skeletal 2-groups. We provide the general gluings we need to do to construct the phase space of an arbitrary triangulation.  We show in particular how the phase space constructed in \cite{Asante:2019lki} is a special example of the construction.} 

In the concluding section, we discuss the many new directions that this work opens up.

\section{Mathematical tool box}
\label{Sec_MathematicalToolBox}

In this section, we review the notions of Poisson groups and 2-groups, as well as Heisenberg/Drinfeld doubles. We highlight that for a certain class of 2-groups, namely the ones for which the $t$-map is trivial, they can be seen as groups and hence embedded in the Heisenberg/Drinfeld doubles framework. We  show that the Poincar\'e 2-group can be equipped with a non-trivial Poisson bracket, which is the same as the ones deforming the standard Poincar\'e group into the $\kappa$-Poincar\'e group. Furthermore we discuss how  Majid's semi-dualization \cite{Majid:1996kd, Majid:2008iz}  can be used to define different 2-groups. We recall the symplectic reduction framework, in the case of group valued momentum maps, which will be the main tool we will use to build phase spaces for a triangulation. 

\subsection{Heisenberg double, symmetries and symplectic reduction}
\label{Sec_Heisenberg-Drinfeld}

\paragraph{Poisson Lie group.} 
\begin{definition}\cite{Chari:1994pz, Majid:1996kd}
Let $G$ be a Lie group with Poisson brackets $\poi{\cdot,\cdot}_G$. If the group multiplication $\mu: G \times G\to G$ is a Poisson map then $G$ is said to be a Poisson Lie group. Explicitly, for $g,h \in G$ and $f_1,f_2 \in C^\infty (G)$, the compatibility is
\begin{align}
    \poi {f_1, f_2}_G (\mu(h,g)) = \poi{ f_1\circ R_g, f_2 \circ R_g}_G(h)+ \poi{ f_1\circ L_h, f_2\circ L_h}_G(g), \label{Poisson Lie group}
\end{align}
where $L_h$ is the left translation by $h$ and $R_g$ is the right translation by $g$. 
\end{definition}

We can equivalently describe the Poisson brackets by the Poisson bivector $\pi$, given by $\poi {f_1,f_2}_G(g) = \pi_g(df_1\wedge df_2)$. By evaluating \eqref{Poisson Lie group} at $g=h=e$ (where $e$ is the identity of $G$) we find that the Poisson 
bivector vanishes at the identity, $\pi_e =0$. It follows that the  Poisson structure of a Poisson Lie group is not symplectic. 

\begin{definition}\cite{Chari:1994pz, Majid:1996kd}
The associated Lie algebra $\g$ of a Poisson Lie group $G$ is a Lie bi-algebra; it is the Lie algebra of $G$, with its Lie bracket and an additional structure, the co-cycle map $\delta_{\g}: \g \to \g \wedge \g$ derived from the Poisson bivector of $G$. The co-cycle map satisfies the co-Jacobi identity: cyclic permutations of $(\delta\otimes id)\circ \delta $ sum to zero. It also satisfies the co-cycle condition,
\begin{align}
    \delta [x,y] = ad_x\delta y -ad_y \delta x,\quad  \forall x, y \in \g. 
\end{align}
\end{definition}

\paragraph{Classical double.}
Consider the pair of Lie bi-algebras $\g$ and $\g^*$, such that there exists   a bilinear map $\la \,,\,\ra : \g \times \g^* \to \bbR$ as
\be
    \la e^i , e^*_j \ra = \delta^i_j ,
    \label{Lie_bi-linear_map}
\ee
where $e^i\in \g$ and $e_j^*\in \g^*$ are the generators of their respective Lie bi-algebras. 
Let the Lie brackets and co-cycles of the Lie bi-algebra $\g$ be
\be
    [e^i , e^j] = {c^{ij}}_k e^k  
    \,\, , \quad
    \delta_{\g} (e^i) = \demi {d^i}_{jk} e^j \wedge e^k .
\ee
Let the Lie brackets and co-cycles of $\g^*$ be defined by the conditions
\be
    \la e^i , [e^*_j , e^*_k] \ra = \la \delta_{\g} (e^i) , e^*_j \ot e^*_k \ra
    \,\, , \quad
    \la [e^i , e^j] , e^*_k \ra = \la e^i \ot e^j , \delta_{\g^*} (e^*_k) \ra .
    \label{Liedual:brackets_co-cycle}
\ee
The Lie bi-algebra $\g^*$ is thus the Lie bi-algebra with Lie brackets and co-cycles respectively dual to the co-cycles and Lie brackets of $\g$:
\be
    [e^*_i , e^*_j] = {d^k}_{ij} e^*_k
    \,\, , \quad
    \delta_{\g^*} (e^*_i) = \demi {c^{jk}}_i e^*_j \wedge e^*_k .
\ee
The co-adjoint actions between $\g$ and $\g^*$ are 
\be
    ad^*_{e^j} (e^*_i) = e^*_i \lhd e^j := \la e^*_i , [e^j , e^k] \ra e^*_k = {c^{jk}}_i e^*_k
    \,\, , \quad
    ad^*_{e^*_i} (e^j) = e^*_i \rhd e^j := - \la [e^*_i , e^*_k] , e^j \ra e^k = {d^j}_{ik} e^k .
    \label{Co-Adjoint_Action}
\ee
The adjoint action of any Lie algebra on itself is given by its Lie brackets. The bilinear map \eqref{Lie_bi-linear_map} encodes the fact that the adjoint action of a Lie algebra on itself, is dual to the co-adjoint action of that algebra on its dual:
\be
    \la [e^i, e^j] , e^*_k \ra = \la e^j , ad^*_{e^i} (e^*_k) \ra
    \,\, , \quad
    \la e^i , [e^*_j , e^*_k] \ra = \la ad^*_{e^*_j} (e^i) , e^*_k \ra .
\ee

\begin{definition}[Classical double.]
\label{Def_ClassicalDouble}
    \cite{Chari:1994pz, Majid:1996kd}
    The classical double of $\g$ is the double cross sum Lie bi-algebra $\d = \g^{cop} \bowtie \g^*$, equipped with with the element $r = e^*_i \ot e^i \, \in \, \g^* \ot \g$, called classical $r$-matrix. Here $\g^{cop}$ is the Lie bi-algebra with the same Lie brackets of $\g$ and opposite co-cycle: $\delta_{\g^{cop}} (e^i) = - \delta_{\g} (e^i)$. As a Lie algebra, $\d$ is defined by the brackets
    \be
        [e^i , e^j] = {c^{ij}}_k e^k
        \,\, , \quad
        [e^*_i , e^*_j] = {d^k}_{ij} e^*_k
        \,\, , \quad
        [e^*_i , e^j] = e^*_i \rhd e^j + e^*_i \lhd e^j = {d^j}_{ki} e^k + {c^{jk}}_i e^*_k ,
        \label{ClassicalDouble_LieBrackets}
    \ee
    where the cross brackets are induced by the co-adjoint actions \eqref{Co-Adjoint_Action}. As a Lie co-algebra, $\d$ is defined by the co-cycles
    \be
        \delta_{\d} (e^i) = - \delta_{\g} (e^i)
        \,\, , \quad
        \delta_{\d} (e^*_i) = \delta_{\g^*} (e^*_i) .
        \label{ClassicalDouble_Cocycles}
    \ee
    The classical $r$-matrix can be split into its symmetric and antisymmetric parts
    \be
        r_- = \demi (r - r^t) = \demi \sum_i e^*_i \wedge e^i
        \,\, , \quad
        r_+ = \demi (r + r^t) .
    \ee
    The co-cycles \eqref{ClassicalDouble_Cocycles} can equivalently be expressed through the co-boundary condition,
    \be
        \delta_{\d} (x) = [x \ot 1 + 1 \ot x , r_-] 
        , \quad \forall x \in \d , 
        \label{Co-boundaryEq}
    \ee
    making $\d$ a co-boundary Lie bi-algebra. 
\end{definition}

\subsubsection{Heisenberg and Drinfeld doubles}

In the following we consider the cases where the cocycle of $\d$ is given by a $r$-matrix.

\paragraph{Heisenberg double.} 

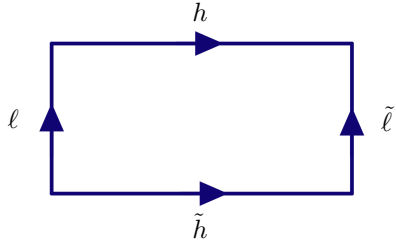
\begin{figure}
    \centering
       \begin{tikzpicture}
    
    \node[anchor=south west,inner sep=0] (image) at (0,0)
    {\includegraphics[width=0.3\textwidth]{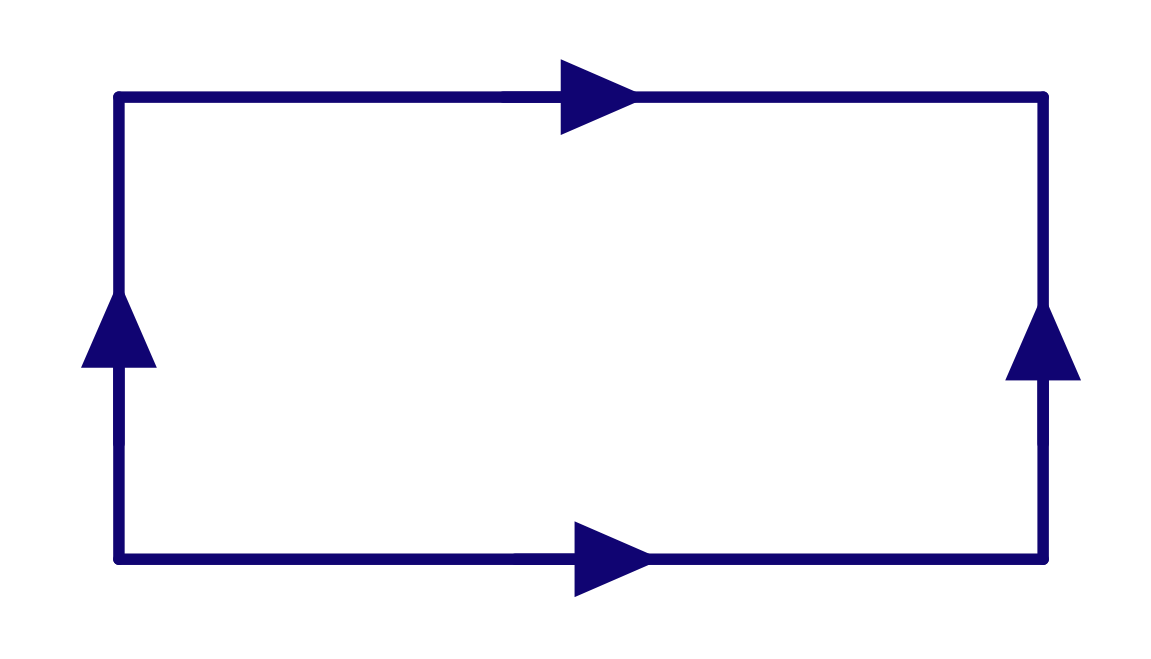}};
    \begin{scope}[x={(image.south east)},y={(image.north west)}]
    \node at (0,0.5) {$\ell$};
    \node at (0.5,0) {$\th$};
    \node at (0.5,1) {$h$};
    \node at (1,0.5) {$\tell$};
    \end{scope}
    \end{tikzpicture}
 
    \caption{The two decompositions of an element in the Heisenberg double are depicted by a ribbon diagram}
    \label{fig:rib}
\end{figure}
Upon exponentiation, the components, $\g$ and $\g^*$ of the classical double $\d$ generate respectively the group $G$ and $G^*$. Similarly, we define the (simply connected) Lie group generated by $\d$ to be  $\cD$. The group $\cD$ equipped with the Poisson bi-vector
\be
    \pi_+ (d) = - [d \ot d , r_-]_+, \quad d\in\cD,
\ee
is called the Heisenberg double of $\cD$ \cite{Chari:1994pz, Majid:1996kd, SemenovTianShansky:1993ws, Alekseev_1994}  and is denoted $\cD_+$ (or simply $\cD$ if there is no confusion  with the Drinfeld double --- to be defined shortly). $[\cdot,\cdot]_+$ stands for the anticommutator. The Heisenberg double can be factorized\footnote{This splitting will not be global in general. For our purpose it is sufficient to consider elements $d\in\cD$ close to the identity so that this splitting is true. We refer to \cite{Alekseev_1994} for the general case. } as $\cD_+ \cong G\bowtie G^* \cong G^* \bowtie G$. With $h,\th \in G$ and $\ell,\tell \in G^*$, the so-called \textit{ribbon relation} follows from these two splittings:
\be
    d = \ell h = \th \tell . \label{ribbon}
\ee
The ribbon relation can be written diagramatically as in Fig.\ \ref{fig:rib}. The ribbon decomposition allows us to define the mutual actions between $G$ and $G^*$:
\be
    \th = \ell \rhd h, \quad 
    \tell = \ell\lhd h, \quad 
    \ell = \th\rhd \tell, \quad 
    h = \th\lhd \tell. \label{actions from ribbons}
\ee
From the Poisson bivector $\pi_+$, we can derive the Poisson brackets of the Heisenberg double; these Poisson brackets are symplectic, thus even if $G$ and $G^*$ are Poisson Lie groups, the Heisenberg double is not. Being symplectic, $\cD_+$ is instead considered as a phase space with symplectic form \cite{Alekseev_1994, SemenovTianShansky:1993ws}
\be
    \Omega = \demi \big(\la \Dr \ell \wedge \Dr \th \ra + \la \Dl \tell \wedge \Dl h \ra \big)
    \label{SymplecticForm}
\ee 
where $\Dr u= \delta u u\mone$ and $\Dl u=u\mone \delta u$ are respectively the right and left Maurer Cartan forms.

\paragraph{Drinfeld double.}
The Lie group $\cD$ equipped with the Poisson bivector
\be
    \pi_-(d) = [d \ot d , r_-], \quad d\in \cD,
\ee
is a Poisson Lie group called \textit{Drinfeld double} \cite{Chari:1994pz, Majid:1996kd, SemenovTianShansky:1993ws, Alekseev_1994} and denoted  $\cD_-$. Here,  $[\cdot,\cdot]$ are the usual commutators. The co-cycles of the classical double $\d$ are the infinitesimal limit of the Poisson bivector $\pi_-$ associated to the Poisson brackets of $\cD_-$.

\medskip

The Lie group $\cD$ naturally acts on itself by the left or right multiplication. The action of the Drinfeld double $\cD_-$ on the Heisenberg double $\cD_+$ is  a Poisson map, i.e. the symplectic Poisson bracket is covariant under the action of $\cD_-$ \cite{SemenovTianShansky:1993ws}. We can focus in particular at the action of the subgroups $G$ and $G^*$ on $\cD_+$.

\paragraph{$G$ transformations on the left.} 
Let $\alpha \in \g$. The infinitesimal (left) action $\delta^L_\alpha$ of  $G \ni h'\sim 1 + \alpha$ on $ \cD_+\ni d$, induced by  the left group multiplication is 
\be 
    \label{rotsymleft} 
    h' d \sim (1 + \alpha) d 
    \quad \textrm{with} \quad 
    d = \ell h = \th \tell .
\ee
We deduce then the transformations for the sub-components \cite{Dupuis:2020ndx}
\be \label{rotsymleft1} 
    \delta^L_{\alpha} g = \alpha g
    \quad \to \quad
    \begin{array}{|l l}
        &
        \delta^L_{\alpha} \th = \alpha \th 
        \,\, , \quad 
        \delta^L_{\alpha} \tell = 0 , \\
        &
        \delta^L_{\alpha} \ell = \alpha \ell - \ell (\alpha \lhd \ell) 
        \,\, , \quad
        \delta^L_{\alpha} h = (\alpha \lhd\ell) h .
    \end{array}
\ee

\paragraph{$G^*$ transformations on the left.} 
A similar calculation can be performed for the infinitesimal (left) transformations $\delta^L_{\phi}$ of $G^* \ni \ell'\sim 1 + \phi$, with $\phi \in \g^*$
\be
    \ell' d \sim (1 + \phi) d 
    \quad \textrm{with} \quad
    d = \ell h = \th\tell ,
\ee
from which, the transformations for the sub-components are \cite{Dupuis:2020ndx}
\be \label{transsymleft1} 
    \delta^L_{\phi} d = \phi d
    \quad \to \quad
    \begin{array}{| l l }
        &
        \delta^L_{\phi} \ell = \phi \ell 
        \,\, , \quad 
        \delta^L_{\phi} h = 0 , \\
        &
        \delta^L_{\phi} \th = \phi \th - \th (\th\mone \rhd \phi) = \phi \th - \th (\phi \lhd \th) 
        \,\, , \quad
        \delta^L_{\phi} \tell = (\phi \lhd \th) \, \tell .
    \end{array}
\ee

\subsubsection{Momentum maps and symplectic reduction}
\label{Sec_SymplecticReduction}

In this section we review the concepts of momentum maps and symplectic reduction applied to our specific case of interest. We refer to  \cite{alekseev1998} for the general set up. The symplectic reduction is the essential tool to build up the phase spaces for the triangulation we will consider. To take an analogy, if Heisenberg doubles can be seen as bricks then the symplectic reduction can be interpreted as the cement. 
 
\paragraph{Momentum maps.} For what concerns us, momentum maps are functions on phase space which generate the infinitesimal symmetry transformations of interest using the symplectic Poisson bracket. We will be exclusively focused on the symmetry transformations given by the left or right translations on $\cD$, the Heisenberg double, given by the action of $G$, $G^*$ and possibly some of their subgroups. These actions are free and proper.  
 
In the simple example of a phase space $T^*\mathbb{R}\sim \R\times \R$ with coordinates $(x,p)$, the momentum $p$ generates the infinitesimal translations in $x$, while $x$ generates the infinitesimal translations in $p$. Hence these coordinates can be seen as momentum maps. More explicitly, if each type of infinitesimal translation is 
such that \begin{align}
&\delta_\epsilon x= \epsilon, \delta_\epsilon p=0, \quad 
\delta_\eta x= 0, \delta_\eta p=\eta,
\end{align}
where $\delta_{\epsilon,\eta}$ is understood as the vector field generating the transformation, and $\Omega=\delta x\wedge \delta p$,
then 
\be\label{ex1}
    \delta_{\epsilon} \ip \Omega = \epsilon \delta p= \la \epsilon\, ,\, \delta p\ra, \quad  \delta_{\eta} \ip \Omega =  - \eta \delta x=-\la \eta\, , \, \delta x \ra,
\ee
where we introduce the canonical pairing between $\R$ and $\R^*\sim \R$ and took  $\epsilon\in\R$ and $\eta\in\R^*$. Hence when the manifold of interest is given by a phase space, we see that some momentum maps are just given by the projection either on configuration or momentum coordinates. This concept will extend to the more general Heisenberg double we will consider. 

We denote by $\sigma$ the left Poisson action of  the  Poisson Lie group  $G$ on a phase space $\cP$, and $G^*$ the Poisson Lie group dual to $G$. We denote by $X^L$ the left invariant 1-form on $G^*$ which takes value $\alpha$ at the identity of $e$.  $\delta_X$ is the vector field generating the infinitesimal version of $\sigma$.  
\begin{definition}
The momentum map of the action $\sigma$ is the $C^\infty$ map $\mathcal{J}: \cP\rightarrow G^*$ if
\begin{align}
    \delta_X  \ip \Omega = \mathcal{J} ^* ({X^L}) .
\end{align}
    We say that the momentum map is equivariant if the momentum map commutes with the symmetry action. 
\end{definition}
In the simple example $T^*\R$, since $\R$ is abelian, left or right invariant 1-forms coincide.  In the more general  
Heisenberg double context, since configuration and momentum space are non-abelian groups, we have two types of translation available, on the left and on the right.

\begin{proposition}
We consider the group element $d=\ell h=\th\tell\in\cD$.  The  $G^*$ (resp. $G$) group element $\ell$ (resp. $\th$) generates the infinitesimal left $G$ transformation \eqref{rotsymleft1}  (resp.  the infinitesimal left $G^*$ transformation \eqref{transsymleft1})
\be
    \delta^L_{\alpha} \ip \Omega = \la \alpha \, , \, \Dr \ell \ra, \quad  \delta^L_{\phi} \ip \Omega = - \la \phi ,\, \Dr \th\, \ra, \quad \alpha\in \g, \, \phi\in\g^*. 
\ee
 $\ell:= \cJ^L_{G^*}(d)$ and $\th:=\cJ^L_{G}(d)$ are therefore  {momentum maps}.  In a similar manner, $h$ and $\tell$ are respectively the momentum map for the right transformations of $G^*$ and $G$.
\end{proposition}
Proof of this proposition is given in \cite{Dupuis:2020ndx}. Hence the momentum maps we will consider are given by the projections
\begin{align}
    \cJ_G^R : \cD\sim G^*\bowtie G \dr G, 
    \quad
    \cJ_{G^*}^R : \cD\sim G\bowtie G^* \dr G^*, \\ 
    \cJ_G^L : \cD\sim G\bowtie G^* \dr G, 
    \quad
    \cJ_{G^*}^L : \cD\sim G^*\bowtie G \dr G^*.
\end{align}
These momentum maps are clearly equivariant. For example, consider $\cJ_{G^*}^L$, (we recall that $G$ acts on $\cD$ by left or right multiplication),
\be
    \cJ_{G^*}^L (h' d)= 
    \cJ_{G^*}^L (h' \ell h )= \cJ_{G^*}^L \big((h'\rhd \ell )((h'\lhd\ell) h)   \big) = h'\rhd \ell= h'\rhd  \cJ_{G^*}^L (d), 
    \quad d\in\cD, \, h',h\in G, \, \ell\in G^*. 
\ee

\paragraph{Some properties of momentum maps.}
We can list other relevant properties of momentum maps, all discussed in \cite{alekseev1998}. The following statements, made for the action of $G$, also apply to the action of $G^*$. 
\begin{itemize}
\item If we can decompose the group $G$  (resp. $G^*$) into subgroups $G\sim G_1\bowtie G_2$ (resp.  $G^*\sim G^*_1\bowtie G^*_2$), then the projection of $\cD$ into $G_i$ also defines a momentum map and generates the infinitesimal right $G^*_i$ transformations .
\item We have seen that $\ell\in G^*$ generates the infinitesimal right $G$ translations through the symplectic form $\Omega$. Then  $\ell\mone$ is a momentum map infinitesimal left $G$ translations through the symplectic form $-\Omega$.
\item If we consider the phase space $\cD^{(n)}$ made of $n$ copies of $\cD$,   $\cD^{(n)}=\cD\times...\times \cD$ with symplectic form $\Omega^{(n)}=\sum_i \Omega_i$, then the global right infinitesimal $G$ translation is generated by the momentum map $\ell_1...\ell_n$.   
 \end{itemize}

\paragraph{Symplectic reduction.}
Symplectic reduction is the process of taking a symplectic space with some symmetries and reducing the symmetries (ie by considering functions on the symplectic space which are invariant under the symmetries) such that the resulting space is still symplectic. 
An invariant function can be characterized as a function that commutes with momentum map associated to the symmetries. This is just a different language to discuss Dirac's approach to constrained systems  \cite{dirac2001lectures}. Indeed the constraint corresponds to the requirement that the momentum map is constrained to take a specific value, for example the unit in  $G^*$.  
 
\begin{theorem} \cite{lu, alekseev1998}. \label{thm:main}
Assume $\sigma: G\times \cD\dr \cD$ is a Poisson action of a Poisson Lie group $G$ on the Poisson space $\cD$, with its associated equivariant momentum map $\cJ_{G^*}$. We consider an element $\ell_0\in G^*$ and $G_{\ell_0}$ its isotropy subgroup under the action\footnote{Since $G$ acts on $G^*$, $G_{\ell_0}$ is the set of elements which leave $\ell_0$ invariant.} of $G$ on $G^*$.
If $G$ acts freely and properly on $\cJ\mone(\ell_0)$, then the symplectic form $\Omega$ of $\cD$ is pulled back to the quotient space $\cJ\mone(g)/G_{\ell_0}$, which is denoted\footnote{If $\ell_0$ is the unit, then we write $\cD/\!/G$.} $\cD/\!/G_{\ell_0}$.
\end{theorem}
The general proof of this theorem can be found for example in \cite{lu}. 

The actions we are considering, namely the left of right translations, are free and proper. The associated momentum maps are equivariant. We can therefore use the symplectic reduction. We are going to illustrate the most important cases we will use later to build the phase space for a triangulation.

\subsubsection{Relevant examples of symplectic reduction}\label{sec:main gluings}
The symplectic reduction will be used extensively to glue/fuse phase spaces together and still recover a phase space. There will be three main cases to consider. We emphasize that momentum and configuration variables are treated on a equal footing. Hence while we formulate the constraints in terms of ``momenta" in $G^*$, it would similarly apply  for ``configuration" variables in $G$.  

\paragraph{Gluing rectangular ribbons: identifying momenta pairwise.}

Let us consider a direct product of symplectic phase spaces $\cD\times \cD=\cD^{(2)}$,  with respective symplectic form and ribbon constraint
 \be
     \Omega_a = \demi \big( \la \Dr \ell_a \wedge \Dr \th_a \ra + \la \Dl \tell_a \wedge \Dl h_a \ra \big) 
      , \quad \ell_a h_a= \th_a\tell_a, \quad a = 1,2.
 \ee
 The momentum map generating the global (right) $G$ transformation on $\cD^{(2)}$ is $\tell_1\tell_2$ \cite{alekseev1998}. We consider the associated constraint, or identification,   
 \be
 \cC= \tell_1\tell_2=1 \Leftrightarrow  \tell_1 = \tell_2\mone.
 \ee
We can now construct the symplectic form for the new phase space when we have  performed the symplectic reduction. 
\begin{proposition}\label{prop:glue-ribbon}
The reduced phase space $\cD^{(2)}/\!/ G$ is isomorphic to $\cD$.
\end{proposition}

We are going to proceed to prove this proposition in two ways, a long  and a short one. We still include the long way as it is still instructive to see what happens.
\begin{proof}
When reducing it is convenient to introduce the new variables 
\be
     h_{12} = h_1 \th_{2}\mone,
     \quad
     \th_{12} = \th_1 h_{2}\mone.
     \label{symplectic_constraint:2d_gluing_triangles}
 \ee
From the ribbon constraint  and the definition of $h_{12}$, we have the identities
 \begin{align}
     &
    \Dr \ell = \Dr (\th \tell h\mone) =
     \Dr \th + \th \Dr \tell \th\mone - (\th\tell) \Dl h (\th\tell)\mone ,\\
    &
    \Dl h_{12} = - \Dr \th_{2} + h_{12}\mone \Dr h_1 h_{12}
    \,\, , \quad
    \Dr \th_{12} = \Dr \th_1 - \th_{12} \Dr h_{2} \th_{12}\mone .
 \end{align}
 Let us evaluate $\Omega_{12}$ on the constraint surface,  
 \begin{align}
     \Omega_{12} 
     & = \Omega_1 + \Omega_{2} =
     \demi \big( 
     \la \Dr \ell_1 \wedge \Dr \th_1 \ra + \la \Dl \tell_1 \wedge \Dl h_1 \ra 
     +
     \la \Dr \ell_{2} \wedge \Dr \th_{2} \ra + \la \Dl \tell_{2} \wedge \Dl h_{2} \ra
     \big) 
     \nonumber \\
     & = 
     \demi \big( 
     \la \Dr \ell_1 \wedge \Dr \th_{12} \ra + \la \Dl \ell_{2}\mone \wedge \Dl h_{12} \ra 
     \big)
     + \demi \big(
     \la (\Dl \tell_{2} + \th_{1}\mone \Dr \ell_1 \th_{1}) \wedge \Dl h_{2}\ra + 
     \la (\Dl \tell_1 + \th_{2}\mone \Dr \ell_{2} \th_{2}) \wedge \Dl h_1 \ra 
     \big) \nonumber\\
     & = 
     \demi \big( 
     \la \Dr \ell_1 \wedge \Dr \th_{12} \ra + \la \Dl \ell_{2}\mone \wedge \Dl h_{2} \ra 
     \big)
     \nonumber \\
     & \quad 
     + \demi \big(
     \la (\Dl \tell_{2} + (\Dl \th_1 + \Dr \tell_1 - \tell_1 \Dl h_1 \tell_1\mone)) \wedge \Dl h_{2}\ra 
     +
     \la (\Dl \tell_1 + (\Dl \th_{2} + \Dr \tell_{2} - \tell_{2} \Dl h_{2} \tell_{2}\mone)) \wedge \Dl h_1 \ra 
     \big) 
     \nonumber \\
     & = 
     \demi \big( 
     \la \Dr \ell_1 \wedge \Dr \th_{12} \ra + \la \Dl \ell_{2}\mone \wedge \Dl h_{12} \ra 
     \big)
     \nonumber \\
     & \quad 
     + \demi \big(
     \la \tell_{2}\mone \Dr (\tell_{2} \tell_1) \tell_{2} \wedge \Dl h_{2} \ra +
     \la \tell_1\mone \Dr (\tell_1 \tell_{2}) \tell_1 \wedge \Dl h_1 \ra
     \big) -
     \demi \big(
     \la (\tell_{2} \Dl h_{2} \tell_{2}\mone - \tell_1\mone \Dl h_{2} \tell_1) \wedge \Dl h_1 \ra
     \big)
     \nonumber \\
     & \quad 
     + \demi \big(
     \la \Dl \th_1 \wedge \Dl h_{2} \ra + \la \Dl \th_{2} \wedge \Dl h_1 \ra 
     \big) 
     \nonumber \\
    & \approx       
    \demi \big( 
    \la \Dr \ell_1 \wedge \Dr \th_{12} \ra + \la \Dl \ell_{2}\mone \wedge \Dl h_{12} \ra 
     \big)\label{resulting form}.
 \end{align}
where we evaluated the constraint, $\tell_1\tell_2=1$ and used that the bilinear form between terms in $\g$ are zero.  
In the constrained mechanics language, the constraints  \eqref {symplectic_constraint:2d_gluing_triangles} can be viewed as the second class constraint. We note that the constraint $\th_{12} = \th_1 h_{2}\mone$ is not an independent constraint so that we do have indeed the right number of constraint. 
\begin{eqnarray}
    \left.\begin{array}{r} 
    \th_1 = \ell_1\rhd h_1 \\ 
    h_2\mone = \tell_2\mone \rhd \th_2\mone
    \end{array}\right\} 
    \,{\longrightarrow}\, 
    \th_{12} & = & \th_1 h_{2}\mone= (\ell_1\rhd h_1)(\tell_2\mone \rhd \th_2\mone) \approx (\ell_1\rhd h_1)(\tell_1 \rhd \th_2\mone) 
    \nonumber \\
    & = &
    (\ell_1\rhd h_1)((\ell_1\lhd h_1) \rhd \th_2\mone) = \ell_1 \rhd (h_1\th_2\mone) \approx \ell_1 \rhd h_{12}.
\end{eqnarray}
It is interesting to note that the constraints 
$\tell_1\tell_2=1$ or $h_1\th_2\mone=h_{12}$ can be treated symmetrically. Indeed, we could have considered instead the constraint  $h_1\th_2\mone=h_{12}$ and dealt with $\tell_1\tell_2=1$ as the second class constraint. This would have led to the same reduced phase space. Indeed the symplectic reduction also works when the momentum map is not constrained to be the unit (see Theorem \ref{thm:main}). 
\be
\cD^{(2)}/\!/ G \cong \cD^{(2)}/\!/ G^*_{h_{12}}\cong\cD.
\ee

\medskip

There is a quicker way to derive the reduced symplectic structure. Indeed, we know that for each phase space we have a ribbon constraint.  
\be
\ell_1h_1=\th_1\tell_1, \quad \ell_2h_2=\th_2\tell_2.
\ee
If we impose the constraint $\tell_1\tell_2=1$, then we can combine these ribbon constraints into a single one.
\be
    \left.\begin{array}{c}
    \ell_1h_1=\th_1\tell_1\dr \th_1\mone \ell_1h_1=\tell_1 \\  
    \ell_2h_2=\th_2\tell_2\dr \th_2\mone \ell_2h_2=\tell_2
    \end{array}\right\} \stackrel{\tell_1\tell_2=1}{\longrightarrow} 
    \ell_1 (h_1h_2\mone)= (\th_1\th_2\mone)\ell_2\mone \,\,\, \Leftrightarrow \,\,\,
    \ell_1 \, h_{12}= \th_{12} \ell_2\mone .
\ee
Since we have a ribbon constraint, we can deduce then the associated symplectic form if this ribbon constraint defines a Heisenberg double. This is exactly the symplectic form we derived in \eqref{resulting form}. Hence combining ribbon constraints using a given constraint allows to identify the reduced Heisenberg double/phase space.
\end{proof}

The reduction has therefore in this specific case a simple geometrical interpretation. We consider two ribbons and the constraint  $\tell_1\tell_2=1$ allows to glue them into a new ribbon, with a long ribbon edge given by the other constraint $h_1\th_2\mone=h_{12}$. This is illustrated in Fig. \ref{pair of ribs}. 
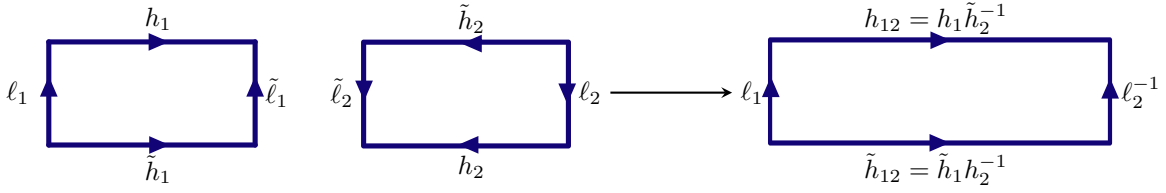
\begin{figure}
    \begin{center}
    \begin{tikzpicture}
    
    \node[anchor=south west,inner sep=0] (image) at (0,0)
    {\includegraphics{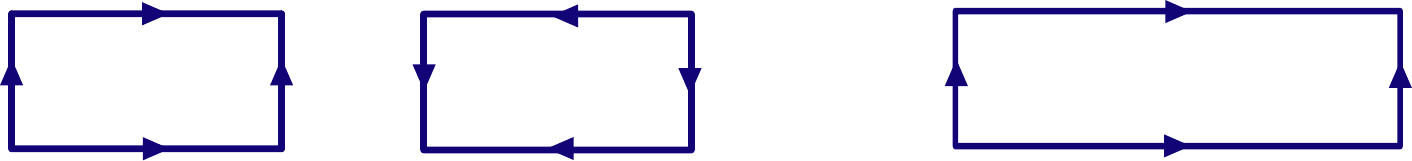}};
    \begin{scope}[x={(image.south east)},y={(image.north west)}]
    \node[] at (-0.02, 0.5) {$\ell_1$};
    \node[] at (0.22,0.5) {$\tell_1$};
    \node[] at (0.11,1.1) {$h_1$};
    \node[] at (0.11,-0.12) {$\th_1$};
    \node[] at (0.28,0.5) {$\tell_2$};
    \node[] at (0.4,1.1) {$\th_2$};
    \node[] at (0.4, -0.1) {$h_2$};
    \node[] at (0.51, 0.5) {$\ell_2$};
    \node[] at (0.52, 0.5) (glue1) {};
    \node[] at (0.65, 0.5) (glue2) {};
    \tikzstyle{arrow} = [thick,->,>=stealth]
    \draw[arrow] (glue1) edge (glue2);
    \node[] at (0.66, 0.5) {$\ell_1$};
    \node[] at (0.83,1.1) {$h_{12} = h_1\th_2\mone$};
    \node[] at (0.83,-0.12) {$\th_{12}=\th_1h_2\mone$};
    \node[] at (1.02,0.5) {$\ell_2\mone$};
    \end{scope}
    \end{tikzpicture}
    \caption{Each phase space is represented by a ribbon diagram. Two ribbons are glued along an edge to create a single ribbon. The geometric gluing is the same as demanding the constrain $\tell_1\tell_2=1$.}
    \label{pair of ribs}
    \end{center}
\end{figure}
\medskip

We note that this gluing is the classical counterpart of the fusion product \cite{alekseev1998}. Such gluing is associative, and if we were considering the other order for the constraint, we would still get an isomorphic phase space \cite{alekseev1998}. 

In the following, we will essentially combine the ribbon constraints to determine the reduced phase space since then the symplectic structure follows when we have a ribbon constraint.

\paragraph{Gluing octagonal ribbons: identifying some sub-components of momentum pairwise.}
Instead of identifying the full momentum variables, we can just identify a subset of them. Consider the case where $G\cong G_1\bowtie G_2$ and dually  $G^*\cong G^*_1\bowtie G^*_2$. The momentum map given by the projection of $\cD$ into $G_1$ generates the infinitesimal right $G^*_1$ transformations. Let us introduce some notations: 
\be
    h = u\lambda, \,\, \tilde h = \tilde u \tilde \lambda \quad \ell= \beta y, \,\, \tilde \ell= \tilde \beta\tilde y, \quad u,\tilde u\in G_1, \, \lambda, \tilde \lambda \in G_2, \, \beta, \tilde \beta\in G_1^*, \, y,  \tilde y\in G_2^*. 
\ee
In terms of geometric interpretation, we can deform the ribbon into an octagon, which keeps track of all the different group elements, see Fig. \ref{octagon}.
\begin{figure}
    \centering
    \includegraphics{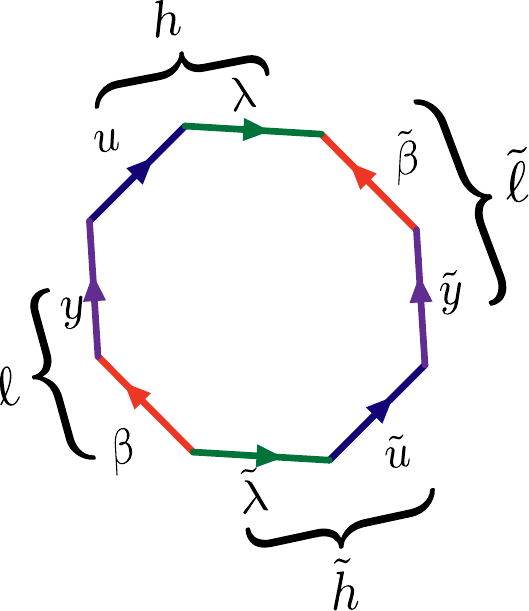}
    \caption{The new ribbon diagram representing the division of $G$ and $G^*$ into subgroups. }
    \label{octagon}
\end{figure}

\medskip 

Let us consider again $\cD^{(2)}= \cD\times\cD$. The global left infinitesimal transformation of $G_1^*$ on $\cD^{(2)}$ is given by the momentum map $u_1 u_2$.

\begin{proposition}\label{prop:glue-octo}
The symplectic reduction $\cD^{(2)}/\!/ G^*_1$ is isomorphic to the Heisenberg double $(G_1 \bowtie (G_2\bowtie G_2))\, \bowtie  \, (G^*_1 \bowtie (G^*_2\bowtie G^*_2))$. 
\end{proposition}
Since we are identifying only a subspace of momentum space, the reduced phase space is larger than $\cD$. The symplectic reduction is nicely interpreted again as the gluing of (octagonal) ribbons, as illustrated in Fig. \ref{gluing octo}.  The new ribbon after fusion has twelve sides. 

\begin{proof}
To prove the proposition, we take this time the short way only, that is we only focus on the ribbon constraints. Fig. \ref{gluing octo} provides a diagrammatic proof of the new ribbon constraint.

As we alluded earlier, once one has identified the ribbon constraint, the symplectic form follows.  
\end{proof}

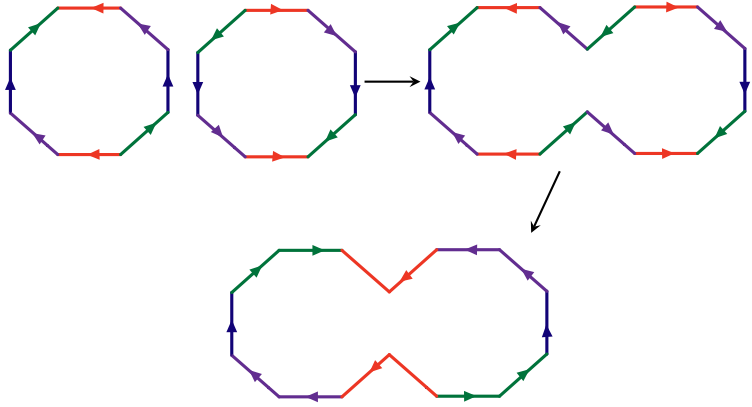
\begin{figure}
    \begin{center}
    \begin{tikzpicture}
     \node[anchor=south west,inner sep=0] (image) at (0,0)
    {\includegraphics[width=0.6\textwidth]{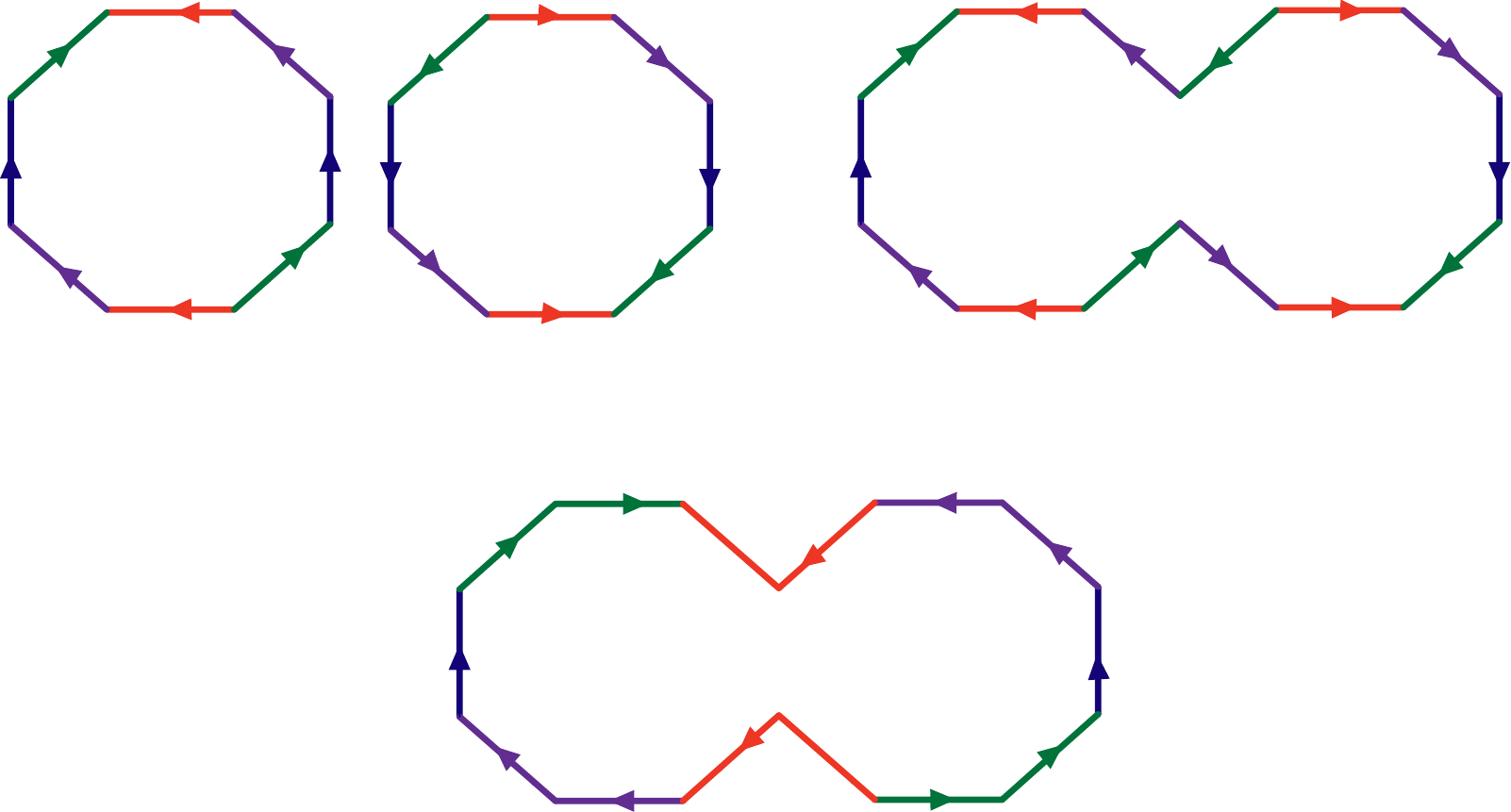}};
    \begin{scope}[x={(image.south east)},y={(image.north west)}]
    \tikzstyle{arrow} = [thick,->,>=stealth]
    \node at (0.47,0.8) (g1) {};
    \node at (0.57,0.8) (g2) {};
    \draw[arrow] (g1) edge (g2);
    \node at (0.75, 0.6) (g3){};
    \node at (0.7, 0.4) (g4){};
    \draw[arrow] (g3) edge (g4);
    \end{scope}
    \end{tikzpicture}
    \caption{We glue two octagons. We then re-arrange the sides belonging to the same groups. This rearranging, involving either actions, back actions and conjugations can induce non-trivial contributions, especially from the conjugations. The choice of ordering is a priori arbitrary and for our concerns, will depend on the choice of frame  we intend to express our variables in. }
    \label{gluing octo}
    \end{center}
\end{figure}

\paragraph{Closure constraints.}
The last example does not involve explicit parameterization of the reduced phase space. Given $n$ copies of $\cD$, we can look at the global $G$ transformation induced by the momentum map $\ell_1...\ell_n$ on $\cD^{(n)}=\cD\times...\times\cD$. The symplectic reduction $\cD^{(n)}/\!/G$ induced by the constraint $\ell_1...\ell_n=1$ is then diagrammatically represented in Fig. \ref{clos}.

\begin{figure}[H]
    \begin{center}
    \begin{tikzpicture}
     \node[anchor=south west,inner sep=0] (image) at (0,0)
    {\includegraphics[width=0.4\textwidth]{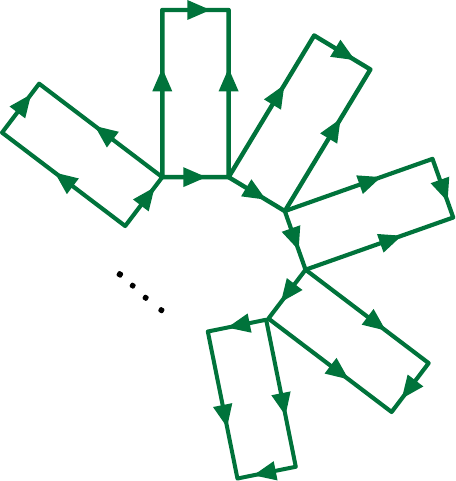}};
    \begin{scope}[x={(image.south east)},y={(image.north west)}]
    \tikzstyle{arrow} = [thick,->,>=stealth]
    \node at (0.45,0.6) {$\ell_1$};
    \node at (0.54,0.57) {$\ell_2$};
    \node at (0.6,0.5) {$\ell_3$};
    \node at (0.6,0.43) {$\ell_4$};
    \node at (0.5,0.35) {$\ell_5$};
    \node at (0.35, 0.56) {$\ell_n$};
    \end{scope}
    \end{tikzpicture}
        \caption{The closure constraint is represented in terms of the ribbon variables.}
    \label{clos}
    \end{center}
\end{figure}
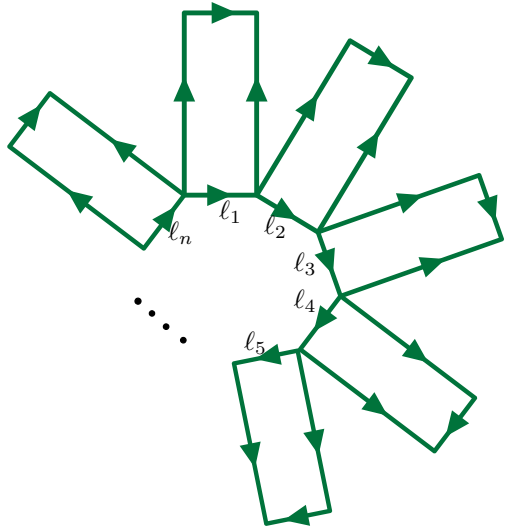

\subsection{2-groups}
\label{Sec_2-groups}

\subsubsection{Definition and examples}
In this section we  review the definition and some examples of strict 2-groups which will be relevant to decorate our polyhedra.

\begin{definition} \cite{baez:2003, Baez:2010ya}
A strict 2-group,  also called a crossed module, is a pair of groups, $G$ and $H$ with a group action $\triangleright$ of $G$ on $H$ such that the map $h\mapsto g\triangleright h$ is a homomorphism for any $g \in G$. In addition a 2-group also includes a group homomorphism $t: H \to G$. The map $t$ is compatible with the action in the following sense:
\be\label{pfeifer}
    t(g\triangleright h)  = gt(h)g^{-1}, \quad
    t(h)\triangleright h'  = hh'h^{-1}, \quad g \in G, \quad  h,h'\in H
\ee
\end{definition}
Since the above definition is the only type of 2-group we consider, we will often omit the word ``strict". A crossed module is equipped with a pair of multiplication rules. We can represent an element of the 2-group as a surface with boundary. The two multiplications allow us to compose surfaces in a consistent way.

A diagrammatic representation of an element $(h,g)\in H\times G$ is given by a surface decorated by an element  $h\in H$. On the boundary of the surface, we specify a point $s$, called the source, and a target point. The boundary is divided into two paths originating from the source. One of the two paths is labelled by $g_1 \in G$ and is seen as the source of the surface decoration. The other path label is defined by the map $t$: $g_2=t(h)g_1$. The setup is shown in Fig. \ref{2groupelement}. 
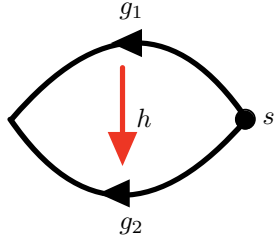
\begin{figure}[H]
    \begin{center}
    \begin{tikzpicture}
    \node[anchor=south west,inner sep=0] (image) at (0,0)
    {\includegraphics[width=0.2\textwidth]{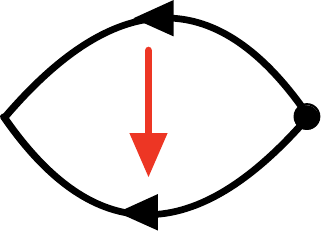}};
    \begin{scope}[x={(image.south east)},y={(image.north west)}]
      \node[] at (0.5, -0.1) {$g_2$};
      \node[] at (0.5, 1.1) {$g_1$};
      \node[] at (0.55, 0.5) {$h$};
      \node[] at (1.05, 0.5) {$s$};
    \end{scope}
    \end{tikzpicture}
    \caption{Graphical representation of $(h,g_1)$ in the 2-group given by $(G,H,t,\triangleright)$, with $t(h)=g_2 g_1\mone$.}
    \label{2groupelement}
\end{center}
\end{figure}
 
Now we consider the first of the two products on the crossed module, the horizontal multiplication. Let $(h_1,g_1), (h_2,g_2) \in H\times G$, then horizontal multiplication is given by
\begin{align}
    (h_1,g_1)\bullet (h_2,g_2) = (h_1(g_1\triangleright h_2),g_1g_2).
\end{align}
This multiplication has an unit given by $(1_H,1_G)$, where $1_H$ and $1_G$ are the identities in $H$ and $G$ respectively. For each $(h,g)\in H\times G$, there is a left and a right horizontal inverse. The left inverse is $(h,g)^{-1_{\bullet ,L}} = ((g\mone\triangleright h)\mone,g\mone)$ and the right inverse is $(h,g)^{-1_{\bullet,R}}= (g\mone \triangleright h\mone,g\mone)$. The geometrical picture for this product is shown in figure \ref{fig:horizontal}
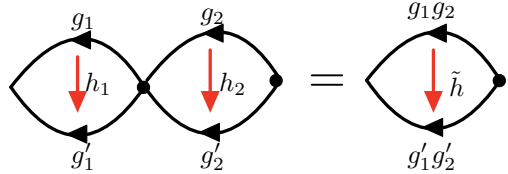
\begin{figure}[H]
    \begin{center}
    \begin{tikzpicture}
    \node[anchor=south west,inner sep=0] (image) at (0,0)
    {\includegraphics[width=0.4\textwidth]{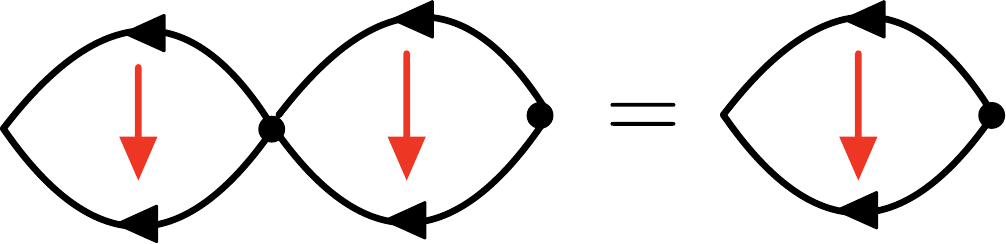}};
    \begin{scope}[x={(image.south east)},y={(image.north west)}]
     \node[] at (0.15, -0.1) {$g_1'$};
     \node[] at (0.18,0.5) {$h_1$};
     \node[] at (0.15, 1.01) {$g_1$};
     \node[] at (0.41,-0.1) {$g_2'$};
     \node[] at (0.41,1.05){$g_2$};
     \node[] at (0.45,0.5) {$h_2$};
     \node[] at (0.85, -0.1) {$g_1'g_2'$};
     \node[] at (0.85,1.1) {$g_1g_2$};
     \node[] at (0.9, 0.5) {$\th$};
    \end{scope}
    \end{tikzpicture}
    \caption{Horizontal product $(h_1,g_1)\bullet (h_2,g_2) = (\th, g_1g_2)$, where $\th = h_1(g_1\triangleright h_2$).}
    \label{fig:horizontal}
    
    \end{center}
\end{figure}

The second multiplication is the vertical multiplication. For $(h_1,g_1), (h_2,g_2) \in H\times G$ satisfying $t(h_1)=g_2g_1^{-1}$, we define
\begin{align}
    (h_2,g_2)\diamond (h_1,g_1) = (h_2h_1,g_1).
\end{align}
The diagrammatic depiction of the vertical multiplication is the composition of two surfaces sharing an edge. This is shown in figure \ref{fig:vertical}.
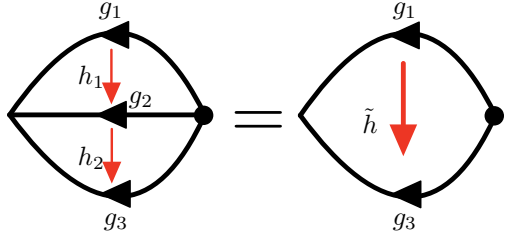
\begin{figure}[H]
    \begin{center}
    \begin{tikzpicture}
    \node[anchor=south west,inner sep=0] (image) at (0,0)
    {\includegraphics[width=0.4\textwidth]{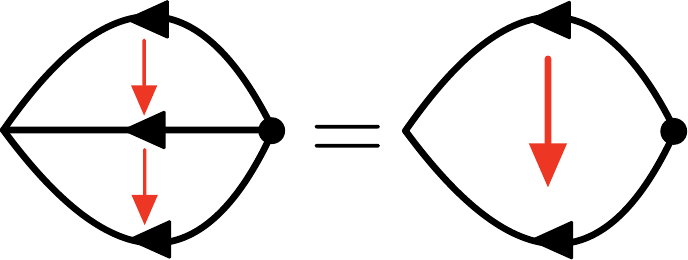}};
    \begin{scope}[x={(image.south east)},y={(image.north west)}]
    \node[] at (0.22,-0.07) {$g_3$};
    \node[] at (0.21,1.05) {$g_1$};
    \node[] at (0.27, 0.58) {$g_2$};
    \node[] at (0.17, 0.7) {$h_1$};
    \node[] at (0.17, 0.27) {$h_2$};
    \node[] at (0.8, -0.07) {$g_3$};
    \node[] at (0.8, 1.05) {$g_1$};
    \node[] at (0.73, 0.47) {$\th$};
     \end{scope}
    \end{tikzpicture}
    \caption{The diagrammatic depiction of the vertical composition $(h_2,g_2)\diamond (h_1,g_1) = (h_2h_1,g_1)$. The condition $t(h_1)=g_2g_1^{-1}$ is expressed graphically by requiring that the two surfaces being composed share an edge. }
    \label{fig:vertical}
    \end{center}
\end{figure}

When vertically composing surfaces, we must ensure that the source points of the surfaces match. It is therefore convenient to detail how one goes about changing the source point, a process called whiskering. Let the original source point be at $s$ on the boundary, and say that we would like to move the source to a new point $s'$ also on the boundary. If the path along the boundary from $s$ to $s'$ is decorated by $g'$, then we have the relation $(h,g_1)_s = (h',g_1g'^{-1})_{s'}$, where the subscripts indicate the base point. The new surface variable, $h'$ satisfies $t(h')=t(h)$, as we can see from figure \ref{fig:changesource}.

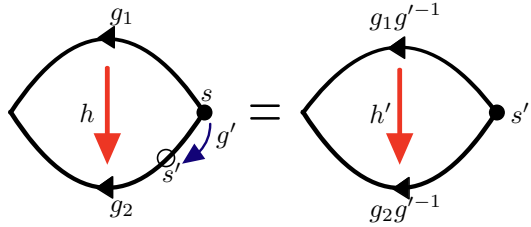
\begin{figure}[H]
    \begin{center}
    \begin{tikzpicture}
    \node[anchor=south west,inner sep=0] (image) at (0,0)
    {\includegraphics[width=0.4\textwidth]{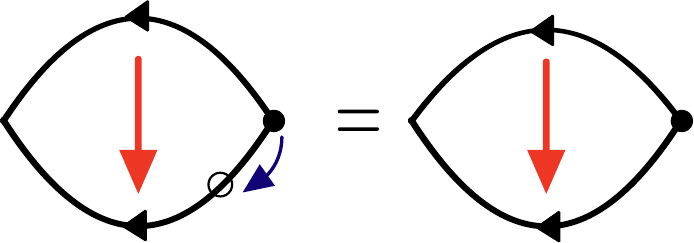}};
    \begin{scope}[x={(image.south east)},y={(image.north west)}]
      \node[] at (0.23,-0.05) {$g_2$};
      \node[] at (0.23,1.05) {$g_1$};
      \node[] at (0.16, 0.5) {$h$};
      \node[] at (0.33, 0.16) {$s'$};
      \node[] at (0.4, 0.6) {$s$};
      \node[] at (0.44, 0.37) {$g'$};
      \node[] at (0.8, -0.05) {$g_2g'^{-1}$};
      \node[] at (0.8, 1.05) {$g_1g'^{-1}$};
      \node[] at (0.75, 0.5) {$h'$};
      \node[] at (1.03,0.5) {$s'$};
    \end{scope}
    \end{tikzpicture}
    \caption{Diagrammatic representation of changing the source point of a surface from $s$ to $s'$, where the two sources are separated by $g'$.}
    \label{fig:changesource}
    \end{center}
\end{figure}

In a similar vein, we can change the target point, the point on the boundary where the two paths $g_1$ and $g_2$ meet. Let the original target point be denoted by $\tau$ and the new target point by $\tau'$, and say the path connecting $\tau$ to $\tau'$ is given by $g'$. Then the new crossed module element is given by $_\tau(h,g_1)\to _{\tau'}(g'\triangleright h, g'g_1)$. This is shown in figure \ref{fig:changetarget}. Since the map $t$ is calculated by taking the path around the loop starting at the target point, it is no surprise that the new surface variable is given by the action of $g'$. 
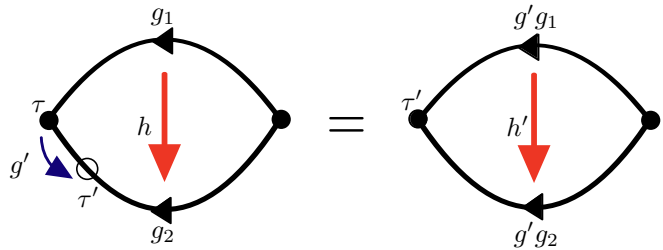
\begin{figure}[H]
    \begin{center}
    \begin{tikzpicture}
    \node[anchor=south west,inner sep=0] (image) at (0,0)
    {\includegraphics[width=0.5\textwidth]{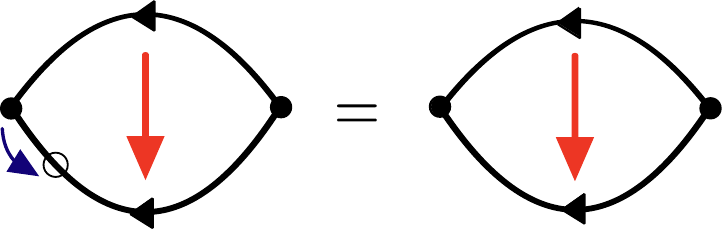}};
     \begin{scope}[x={(image.south east)},y={(image.north west)}]
     \node[] at (0,0.6) {$\tau$};
     \node[] at (-0.03, 0.3) {$g'$};
     \node[] at (0.08, 0.14) {$\tau'$};
     \node[] at (0.17,0.5) {$h$};
     \node[] at (0.2, -0.05) {$g_2$};
     \node[] at (0.2, 1.05) {$g_1$};
     \node[] at (0.6, 0.63) {$\tau'$};
     \node[] at (0.77, 0.5) {$h'$};
     \node[] at (0.8, -0.05) {$g'g_2$};
     \node[] at (0.8, 1.05) {$g'g_1$};
     \end{scope}
    \end{tikzpicture}
    \caption{Changing the target point by a path $g'$. We have that $h' = g'\triangleright h$. }
    \label{fig:changetarget}
    \end{center}
\end{figure}

In fact, the new target or source point need not be on the boundary of the surface. We can choose a new target or source which lies off the boundary altogether without changing the calculation. We call this kind of source/target changing \textit{whiskering} and is shown in Fig. \ref{fig:whiskering}.
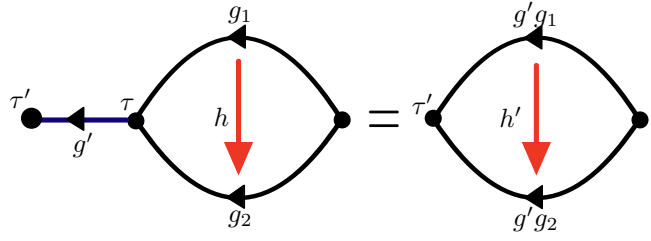
\begin{figure}[H]
    \begin{center}
    \begin{tikzpicture}
    \node[anchor=south west,inner sep=0] (image) at (0,0)
    {\includegraphics[]{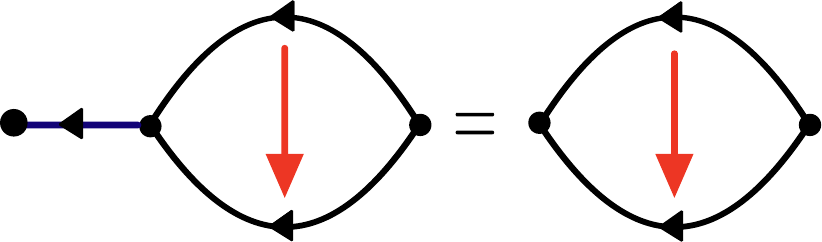}};
    \begin{scope}[x={(image.south east)},y={(image.north west)}]
    \node[] at (0,0.6) {$\tau'$};
    \node[] at (0.1,0.35) {$g'$};
    \node[] at (0.17, 0.57) {$\tau$};
    \node[] at (0.32,0.5) {$h$};
    \node[] at (0.35,-0.05) {$g_2$};
    \node[] at (0.35, 1.05) {$g_1$};
    \node[] at (0.64,0.57) {$\tau'$};
    \node[] at (0.82,-0.05) {$g'g_2$};
    \node[] at (0.82, 1.05) {$g'g_1$};
    \node[] at (0.78, 0.5) {$h'$};
    \end{scope}
    \end{tikzpicture}
    \caption{Whiskering on the target of the surface. Once again we have that $h' = g'\triangleright h$.}
    \label{fig:whiskering}
    \end{center}
\end{figure}

Let us focus on a special class of strict 2-groups, which will be the main objects we will work with. 
\begin{definition}
If a crossed module has $t(h)=
1$ for all $h\in H$, we say that it is a \textit{skeletal crossed module}.  
    \end{definition}
This restriction on the $t$-map has several implications.  

First the compatibility relations \eqref{pfeifer} imply that the group $H$ has to be abelian. Second  this constraint amounts to have a flat holonomy along the closed loop. 
\be
    g_2=t(h)g_1
    \,\, \leftrightarrow \,\,
    g_2g_1\mone=1.
\ee
Finally, in this case, we can relate the vertical multiplication to the horizontal one.  Indeed, since $g_1=g_2$ for each $g\in G$, the vertical composition $(h_1,g)\diamond (h_2,g) = (h_1h_2,g)$ simply corresponds to the multiplication of the group $H$. We can also write
\be
  (h_2,g)\diamond (h_1,g) = (h_2,1)\bullet (h_1,g).\label{horizontalisvertical}
\ee
In this sense, the horizontal product which can be seen as the product of the group $G\ltimes H$ can be used to specify the vertical product. Roughly said, skeletal crossed modules are not groups but are not so far away from being a group. We will use this extensively to associate a phase space structure to a pair of 2-groups.

\paragraph{Examples} Let us consider  some examples of skeletal crossed modules. \cite{Baez:2010ya}

\textit{Trivial 2-groups.} Consider the crossed module where $H={1}$ is the trivial group and $G$ any group. In this case, both the action and map $t$ are trivial. This 2-group is essentially a group. This is the usual framework for lattice gauge theory with $G$ decorating paths and no information on the faces (or trivial group elements on faces).

Alternatively, we may consider $G$ to be the trivial group. The compatibility conditions mean that $H$ must then be abelian. Once again, this 2-group is simply a group. Geometrically, in analogy with the lattice picture before, we have decorations on plaquettes, but not on paths.

\textit{Semi-direct product.} Let $\cG$ be a semi direct product group $\cG = H\rtimes G$ with $H$ abelian, so that for any $k \in \cG$ there is the unique factorization $k = hg$ with $h\in H$ and $g \in G$. Such a group can be identified as a skeletal crossed module with the action 
 given by conjugation. 
The identification is $hg \mapsto (h,g)$. Under this representation, multiplication in $\cG$ corresponds to horizontal multiplication:
\begin{align}
    (hg)(h'g') = h (gh'g\mone)gg' \mapsto (h (gh'g\mone),gg')= (h,g)\bullet (h',g'),
\end{align}
with the action being conjugation. It is clear that the left and right horizontal inverses are the same and is consistent with the mapping: $(hg)\mone =(g\mone h\mone g) g\mone \mapsto (g\mone h\mone g, g\mone)$.

The vertical multiplication can be seen as simply the multiplication in the subgroup $H$.

\textit{Poincar\'e and Euclidean 2-group.} The Poincar\'e (or Euclidean) group can be seen as examples of semi-direct product 2-groups, with $H$ the group of translations and $G$ the Lorentz (or rotation) group. 

\textit{Co-tangent 2-group.} Let $G$ be a Lie group and $H=\g^*$ be its dual Lie algebra, which we see as an abelian group. Then, $G$ acts on $\g^*$ via the co-adjoint action.  
Since the $t$-map is trivial, $\g^*\ltimes G$ is  a skeletal crossed module.

\paragraph{Lie 2-algebras} For our concerns, the pair of groups which make up a 2-group are most of the time Lie groups. As such, it is useful to define the infinitesimal version of crossed modules coming from Lie groups. 

\begin{definition}
\cite{Baez:2003fs, Baez:2010ya} A strict Lie 2-algebra is a pair of Lie algebras with a Lie algebra homomorphism $dt: \h \to \g$ and an action $\triangleright: \g \times \h \to \h$ satisfying 
\begin{align}
    [x_1,x_2]\triangleright y = x_1\triangleright (x_2\triangleright y) -x_2\triangleright(x_1\triangleright y)&&
    x\triangleright[y_1,y_2] = [x\triangleright y_1,y_2]+[y_1,x\triangleright y_2]
\end{align}
for all $x_i \in \g$ and $y_i \in \h$. The homomorphism and the action must be compatible in the sense that
\begin{align}
    dt(x\triangleright y) = [x,dt (y)]&&
    dt(y_1)\triangleright y_2 = [y_1,y_2].
\end{align}
A skeletal strict Lie 2-algebra is an infinitesimal crossed module $\g\ltimes \h$ with a trivial $dt$-map, so that $\h$ is abelian.    
\end{definition}

Given a Lie 2-group, one can construct a corresponding Lie 2-algebra \cite{Baez:2003fs}. The Lie algebras $\g$ and $\h$ are taken to be Lie algebras of the groups $G$ and $H$ respectively. The action and the Lie algebra homomorphism are the infinitesimal versions of the action and homomorphism of the 2-group.

\subsubsection{Poisson Lie 2-groups and 2-bi-algebras}
Just as a Poisson Lie group was a Lie group equipped with a Poisson bracket respecting group multiplication, we can define the same for crossed modules.
\begin{definition}\cite{chen2013}
A Poisson Lie 2-group is a Lie 2-group $\cG$ equipped with a Poisson bracket that is
compatible with both product laws.
\end{definition}
Explicitly, for any functions $f_1$ and $f_2$ on $\cG$ and $x, y \in \cG$,
\begin{align}
    \{f_1,f_2\}(x\bullet y) =  \{f_1\circ L_x^\bullet, f_2\circ L_x^\bullet \} (y) +  \{f_1\circ R_y^\bullet, f_2\circ R_y^\bullet \} (x)\label{horizcompat}\\
    \{f_1,f_2\}(x\diamond y) =  \{f_1\circ L_x^\diamond, f_2\circ L_x^\diamond \} (y) +  \{f_1\circ R_y^\diamond, f_2\circ R_y^\diamond \} (x)\label{verticompat}
\end{align}
where $L^\bullet_x$ ($R_y^\bullet$) is the left (right) horizontal multiplication by $x$ ($y$), similarly for $L^\diamond$.

\paragraph{Examples.}
\begin{itemize}
    \item If $\cG$ is any 2-group, equipping it with the trivial Poisson bracket $\{\cdot,\cdot\} = 0$ makes $\cG$ a Poisson Lie 2-group. 
    \item If $\cG$ is a trivial 2-group in the sense that either $G$ or $H$ is the trivial group, then if the other group is a Poisson Lie group, $\cG$ is a Poisson Lie 2-group. 
    \item skeletal crossed modules provide some easy examples of such Poisson 2-groups. Indeed, consider $H$ a Poisson Lie group and $\cG=H\rtimes G $ a Poisson Lie group such that the induced Poisson bracket on  $H$ makes it a Poisson Lie group. Then $\cG$ is can be viewed as a Poisson 2-group. Indeed, $H\rtimes G$ being a Poisson group means that the horizontal product is compatible with the Poisson structure. If $H$ is a Poisson Lie group, then the vertical product is also compatible with the Poisson structure.
    \item The co-tangent 2-group as a non-trivial Poisson 2-group. The semi-direct product  $G \ltimes \g^*$ can be seen as a the Drinfeld double associated to the Heisenberg double $G \ltimes \g^*$ which is nothing else than $T^*G$. Explicitly,  $p^i$ are the coordinate functions of $\g^*$ and $g^i{}_j$ are the matrix elements of $G$  then the following Poisson structure (the Drinfeld double Poisson brackets) gives a Poisson 2-group
\begin{align}
    &
    \{g^{i}{}_j,g^k{}_l\} = \{ g^i{}_j, p^k\} =0 
    \,\,,\quad \{p^i,p^j\} = f^{ij}{}_kp^k, 
\end{align}
where $f^{ij}{}_k$ is the structure constant of $\g$.
\end{itemize}

\medskip

The infinitesimal version of a Poisson Lie group is a Lie bi-algebra. Similarly, the infinitesimal version of a Poisson 2-group is a Lie 2-bi-algebra. We defined a bi-algebra as Lie algebra equipped with a a cocyle. We can equivalently define it as a pair of Lie algebra, $(\g,\g^*)$ in duality, where the Lie bracket of $\g^*$ is defined by the cocyle of $\g$. This is this type of definition used in \cite{CHEN201359}.   

\begin{definition}\cite{CHEN201359} A Lie 2-bi-algebra (aka a Lie bi-algebra crossed module) is a pair of infinitesimal crossed modules in duality, $(\g\ltimes \h,dt)$,  $(\h^*\rtimes\g^*,-dt^*)$, such that $(\g\ltimes \h, \h^*\ltimes \g^*)$ form a Lie bi-algebra.  
\end{definition}
In particular the pairs $(\g,\g^*)$ and $(\h,\h^*)$ are Lie bi-algebras \cite{CHEN201359}. 
 The following theorem will be quite useful when constructing explicit Lie 2-bi-algebras in Sec. \ref{sec:semi-cross}. 

\begin{theorem}\cite{CHEN201359}\label{thm:match-crossed} Assume that $(\g\ltimes\h, dt)$ and $(\h^*\ltimes \g^*, -dt^*)$ are infinitesimal crossed modules, then they form a Lie 2-bi-algebra if and only if $(\g,\h^*)$ form a matched pair of Lie algebras $\g\bowtie\h^*$, where the  action of $\g$ on $\h^*$ is the dual of the action of $\g$ on $\h$, and the action of $\h^*$ on $\g$ is dual to the action of $\h^*$ on $\g^*$. 
\end{theorem}

\subsection{Heisenberg double and skeletal crossed modules}
\label{Sec_Semidualization}
In this section, we discuss how we can construct Heisenberg doubles out of a pair of  skeletal crossed modules. The key idea is that if a skeletal crossed module still behaves as a group then it can be embedded into an Heisenberg double. The simplest example is  when the 2-group is trivial, namely just a group. For example, $T^*G\cong \g^*\rtimes G$ can be seen\footnote{Interestingly, $T^*G$ can also be seen as a 2-group. It would be interesting to determine whether  general  notion of 2-Heisenberg double, eg a 2-group built out of a pair of 2-groups equipped with a symplectic structure, can be defined in some way. } as built from a pair of trivial 2-groups, $G$ and $\g^*$. This is the typical structure that appears when  discretizing a BF theory \cite{peter}.  skeletal 2-groups provide an another example.

In particular, we will recover the Poisson version of the $\kappa$-Poincar\'e group and its dual, the $\kappa$-Poincar\'e algebra. Said differently, we will recover the Poincar\'e group equipped with a non-trivial Poisson structure which can then be interpreted as a non trivial Poisson 2-group.

\subsubsection{Semi-dualized Lie bi-algebras and skeletal Lie 2-bialgebras}\label{sec:semi-cross}

Consider the Lie bialgebras $\g_1$ and $\g_2$ with generators $e^i$ and $f_i$ respectively, with dual spaces $\g^*_1 $ and $\g^*_2$ generated by $e^*_i$ and $f^*_i$, equipped with a bi-linear map, defined by
\be
    \la e^i , e^*_j \ra = \la f^i , f^*_j \ra = \delta^i_j ,
    \label{Lie_bi-linear_map_pair}
\ee
and all the other pairings between generators vanish. 

We demand that $\g_1$ and $\g_2$ have trivial co-cycles\footnote{This condition might be released, we leave it for further investigations.}, so that the $\g_1$ and $\g_2$ bialgebra structures are respectively
\begin{align}
&    [e^i , e^j] = {c^{ij}}_k e^k
    \,\, , \quad
    \delta_{\g_1} (e^i) = 0,\\
 &   [f^i , f^j] = {d^{ij}}_k f^k
    \,\, , \quad
    \delta_{\g_2} (f^i) = 0 .
\end{align}

\paragraph{Semi-dualized Lie bi-algebras. } 
First we introduce the concept of matched pair of Lie algebras, we will focus on the cocycle sector after that.
\begin{definition}[Matched pair]     \cite{Majid:1996kd}
\label{Def_MatchedPair}
    Consider a pair of Lie algebras $\g_1$ and $\g_2$, they form a matched pair, noted $\g_1\bowtie \g_2$, if there exist actions
    \begin{align}
        \triangleleft :
        & \quad \mathfrak{g}_2 \ot \mathfrak{g}_1 \to \mathfrak{g}_2
        &
        \triangleright :
        & \quad \mathfrak{g}_2 \ot \mathfrak{g}_1 \to \mathfrak{g}_1 \nonumber \\
        & \quad f^i \triangleleft e^j = {\alpha^{ij}}_k f^k
        &
        & \quad f^i \triangleright e^j = {\beta^{ij}}_k e^k 
        \label{Lie_action-back_action}
    \end{align}
    that obey to the conditions
    \begin{align}
        [f^i , f^j] \rhd e^k & = 
        f^i \rhd (f^j \rhd e^k) - f^j \rhd (f^i \rhd e^k) 
        \,\, , \qquad
        f^i \lhd [e^j , e^k] = 
        (f^i \lhd e^j) \lhd e^k - (f^i \lhd e^k) \lhd e^j , 
        \label{MatchedPair_Condisions_1} \\
        f^i \rhd [e^j , e^k] & = 
        [f^i \rhd e^j , e^k] + [e^j , f^i \rhd e^k] +
        (f^i \lhd e^j) \rhd e^k - (f^i \lhd e^k) \rhd e^j , \\
        [f^i , f^j] \lhd e^k & = 
        [f^i \lhd e^k , f^j] + [f^i , f^j \lhd e^k] +
        f^i \lhd (f^j \rhd e^k) - f^j \lhd (f^i \rhd e^k) .
        \label{MatchedPair_Condisions_2}
    \end{align}
\end{definition}
Given a matched pair, we can consider the dual picture.

\begin{definition}[Dual of a matched pair]
\label{Def_DualMatchedPair}     \cite{Majid:1996kd}
    Consider a matched pair $(\g_1,\g_2)$ equipped with trivial co-cycles. The dual of such matched pair is noted $\g_1\blackbowtie\g_2$ and is a pair of Lie bi-algebras $\g_1^*$ and $\g_2^*$, with Lie brackets and co-cycles inferred by \eqref{Liedual:brackets_co-cycle} and \eqref{Lie_bi-linear_map_pair},
    \begin{align}
        &
        [e^*_i , e^*_j] = 0
        \,\, , \quad
        \delta_{\g_1^*} (e^*_i) = \demi {c^{jk}}_i e^*_j \wedge e^*_k , \\
        &
        [f^*_i , f^*_j] = 0
        \,\, , \quad
        \delta_{\g_2^*} (f^*_i) = \demi {d^{jk}}_i f_j^* \wedge f^*_k ,
    \end{align}
    equipped with the pair of co-actions
    \begin{align}
        \alpha :
        & \quad \mathfrak{g}_2^* \to \mathfrak{g}_2^* \ot \mathfrak{g}_1^*
        &
        \beta :
        & \quad \mathfrak{g}_1^* \to \mathfrak{g}_2^* \ot \mathfrak{g}_1^* \nonumber \\
        & \quad \alpha(f_i^*) = {\alpha^{jk}}_i f_j^* \ot e_k^*
        &
        & \quad \beta(e_i^*) = {\beta^{jk}}_i f_j^* \ot e_k^* 
        \label{Lie_action-back_co-action}
    \end{align}
    derived, by dualization of the actions \eqref{Lie_action-back_action}, as
    \be
        \la f^i \lhd e^j , f^*_k \ra = \la f^i \ot e^j , \alpha(f^*_k) \ra
        \,\,,\quad
        \la f^i \rhd e^j , e^*_k \ra = \la f^i \ot e^j , \beta(e^*_k) \ra .
    \ee    
\end{definition}
We emphasize that if actions are denoted by triangles of the type $\lhd$, or $\rhd$, co-actions are  denoted by black triangles of the type $\blacktriangleleft$, or $ \blacktriangleright$.

From a  matched pair $(\g_1,\g_2)$, we just  built its ``full" dual $(\g_1^*,\g_2^*)$.  Instead we can only dualize a given component of the matched pair, say $\g_2$ for example. We get in this case a bicross sum.

\begin{definition}[Bicross sum]
\label{Def_BicrossSum}   
    \cite{Majid:1996kd}
    Consider the matched pair $(\g_1,\g_2)$, there exist two bicross sum Lie bi-algebras dual to each other\footnote{The co-action is encoded by the black triangle.}: the Lie bi-algebra $\g_2^* \bicrossr \g_1$, equipped with an action of $\g_1$ on $\g_2^*$ and a co-action $\beta$ of $\g_2^*$ on $\g_1$,
    \begin{align}
        \rhd^* :
        & \quad \mathfrak{g}_1 \ot \mathfrak{g}_2^* \to \mathfrak{g}_2^*
        &
        \beta :
        & \quad \mathfrak{g}_1 \to \mathfrak{g}_1 \ot \mathfrak{g}_2^*
        \nonumber \\
        & \quad e^i \rhd^* f^*_j = {\alpha^{ki}}_j f^*_k
        &
        & \quad \beta(e^i) = {\beta^{ki}}_j e^j \ot f^*_k
    \end{align}
    and the Lie bi-algebra $\g_2 \bicrossl \g_1^*$, equipped with an action of $\g_2$ on $\g_1^*$ and a co-action of $\g_1^*$ on $\g_2$,
    \begin{align}
        \lhd^* :
        & \quad \mathfrak{g}_1^* \ot \mathfrak{g}_2 \to \mathfrak{g}_1^*
        &
        \beta :
        & \quad \mathfrak{g}_2 \to \mathfrak{g}_1^* \ot \mathfrak{g}_2
        \nonumber \\
        & \quad e^*_i \lhd^* f^j = {\beta^{jk}}_i e^*_k
        &
        & \quad \alpha(f^i) = {\alpha^{ij}}_k e^*_j \ot f^k
    \end{align}
    The actions of $\g_1$ on $\g_2^*$ and of $\g_2$ on $\g_1^*$ are derived by dualization of the actions in the matched pair \eqref{Lie_action-back_action}:
    \be
        \la e^i \rhd^* f^*_j , f^k \ra = \la f^*_j , f^k \lhd e^i \ra
        \,\, , \quad
        \la e^*_i \lhd^* f^j , e^k \ra = \la e^*_i , f^j \rhd e^k \ra .
        \label{Action_g1-g2*_g2g1*}
    \ee
    The co-actions of $\g_2^*$ on $\g_1$ and of $\g_1^*$ on $\g_2$ are derived by dualization of the actions $\lhd^*,\rhd^*$ in \eqref{Action_g1-g2*_g2g1*}:
    \be
        \la \beta(e^i) , e^*_j \ot f^k \ra = \la e^i , e^*_j \lhd^* f^k \ra
        \,\, , \quad
        \la \alpha(f^i) , e^j \ot f^*_k \ra = \la f^i , e^j \rhd^* f^*_k \ra .
    \ee
    The cross Lie brackets of $\g_2^* \bicrossr \g_1$ and $\g_2 \bicrossl \g_1^*$ are induced by the actions
    \be
        [e^i , f^*_j] = e^i \rhd^* f^*_j = {\alpha^{ki}}_j f^*_k 
        \,\, , \quad
        [e^*_i , f^j] = e^*_i \lhd^* f^j = {\beta^{jk}}_i e^*_k ,
    \ee
    and the co-cycle of $\g_1$ and $\g_2$ are corrected by the co-actions
    \be
        \begin{aligned}
            \delta_{\g_2^* \bicrossr \g_1} (e^i) & = 
            \delta_{\g_1}(e^i) + (1 - \tau) \circ \beta(e^i) = 
            {\beta^{ki}}_j e^j \wedge f^*_k , \\
            \delta_{\g_2 \bicrossl \g_1^*} (f^i) & = 
            \delta_{\g_2}(f^i) + (1 - \tau) \circ \alpha(f^i) = 
            {\alpha^{ij}}_k e^*_j \wedge f^k .
            \end{aligned} 
    \ee
    
\end{definition}
We can now construct two types of matched pairs. The first one consists in putting together the original (matched pair) $\g_1\bowtie \g_2$ and its dual. Doing this we recover the Lie algebra $\d$ which was used to generate the Heisenberg double associated to $\cD\cong T^*(G_1\bowtie G_2)$.
As such this is a pretty standard construction.

\begin{proposition}[Classical double behind $T^*(G_1\bowtie G_2)$]
\label{Prop_CotangentBundle}
    Given a matched pair $(\g_1,\g_2)$ endowed with trivial co-cycles, the bi-algebra 
    \be
        \d = (\g_1 \bowtie \g_2)^{cop} \bowtie (\g_1^* \blackbowtie \g_2^*) = (\g_1 \bowtie \g_2) \ltimes (\g_1^* \blackbowtie \g_2^*)
    \ee
    with the Lie brackets and co-cycles of $\g_1 \bowtie \g_2$ and $\g_1^* \blackbowtie \g_2^*$ being
    \begin{align}
        [e^i , e^j] & = {c^{ij}}_k e^k
            \,\, &&
            [f^i , f^j] = {d^{ij}}_k f^k \\
            [f^i , e^j] &= {\alpha^{ij}}_k f^k + {\beta^{ij}}_k e^k && [f^*_i , f^*_j] = [e^*_i , e^*_j]=[f^*_i , e^*_j]=  0\\
            \delta_{\g} (e^i) & = 0 
            &&
            \delta_{\g} (f^i) = 0 \\
            \delta_{\g^*} (e^*_i) & =
            \demi {c^{jk}}_i e^*_j \wedge e^*_k + {\beta^{jk}}_i f^*_j \wedge e^*_k &&
            \delta_{\g^*} (f^*_i)  =  
            \demi {d^{jk}}_i f^*_j \wedge f^*_k + {\alpha^{jk}}_i f^*_j \wedge e^*_k .
            \label{Co-cycles_ClassicalDouble_Cotangent}
    \end{align}
    is a classical double. 
    The action of $\g_1 \bowtie \g_2$ on $\g_1^* \blackbowtie \g_2^*$ is the co-adjoint action $ad^*: (\g_1^* \op \g_2^*) \ot (\g_1 \bowtie \g_2) \to (\g_1^* \op \g_2^*)$ defined in \eqref{Co-Adjoint_Action}, and induces the cross Lie brackets
    \be
        \begin{aligned}
            [f^*_i , f^j] & = {d^{jk}}_i f^*_k + {\alpha^{jk}}_i e^*_k
            \,\, , \quad
            [e^*_i , f^j] = {\beta^{jk}}_i e^*_k , \\
            [e^*_i , e^j] & = {c^{jk}}_i e^*_k - {\beta^{kj}}_i f^*_k
            \,\, , \quad
            [f^*_i , e^j] = - {\alpha^{kj}}_i f^*_k .
        \end{aligned}
    \ee
    The co-cycle of the Lie bi-algebra $\d$ can be expressed in terms of the $r$-matrix
    \be
        r = e^*_i \ot e^i + f^*_i \ot f^i .
        \label{Classical_r-matrix_Cotangent}
    \ee
\end{proposition}

\begin{proof}
    The matched pair $(\g_1,\g_2)$ described in Def. \ref{Def_MatchedPair}, defined on the vector space $\g_1 \op \g_2$ and equipped with trivial co-cycles, as a Lie bi-algebra is equivalent to a classical double (introduced in Def. \ref{Def_ClassicalDouble}) with cross Lie brackets induced by the actions of the matched pair \eqref{Lie_action-back_action}. \\
    Similarly, the dual of the matched pair is a Lie bi-algebra with trivial Lie brackets and co-cycles given by the co-cycles of $\g_1^*$ and $\g^*_2$, corrected by the co-actions \eqref{Lie_action-back_co-action}:
    \be
        \delta_{\g_1^* \blackbowtie \g_2^*} (e^*_i) =
        \delta_{\g_1^*} (e^*_i) + (1 - \tau) \circ \beta (e^*_i)
        \,\, , \quad
        \delta_{\g_1^* \blackbowtie \g_2^*} (f^*_i) =
        \delta_{\g_2^*} (f^*_i) + (1 - \tau) \circ \alpha (f^*_i) .
    \ee
    The Lie bi-algebra $\d$ is thus constructed as the cross sum of $\g_1 \bowtie \g_2$ and $\g_1^* \blackbowtie \g_2^*$, with action of the first on the second given by the co-adjoint actions defined in \eqref{Co-Adjoint_Action}. One can check that the co-adjoint action of $\g_1 \bowtie \g_2$ on $\g_1^* \blackbowtie \g_2^*$ satisfies the conditions \eqref{MatchedPair_Condisions_1}-\eqref{MatchedPair_Condisions_2} for $\d$ to be a matched pair; to carry out this computation it is enough to use the Jacobi identity for $\g_1$ and $\g_2$, plus the axioms \eqref{MatchedPair_Condisions_1}-\eqref{MatchedPair_Condisions_2} for the matched pair $(\g_1,\g_2)$. \\
    The co-cycles \eqref{Co-cycles_ClassicalDouble_Cotangent} can be derived through the co-boundary equation \eqref{Co-boundaryEq} using the classical $r$-matrix \eqref{Classical_r-matrix_Cotangent}. 
\end{proof}

The second match pair we can construct can be viewed as a semi-dualization of the previous option: we interchange the $\g_2$ and $\g_2^*$ sectors. We then construct a classical double out of a pair of bicross algebras. This is the structure we will use to deal with non-trivial Poisson 2-groups.  This  can generate curvature in momentum space. As we will see below, it is  the key structure relevant to study some deformation of special relativity \cite{Lukierski:1992dt, AmelinoCamelia:1999pm}.

\begin{proposition}[Semi-dualized case]
\label{Prop_Semidualization}
    Given a matched pair $(\g_1,\g_2)$ endowed with trivial co-cycles, the bi-algebra
    \be
        \b = (\g_2 \bicrossl \g_1^*)^{cop} \bowtie (\g_2^* \bicrossr \g_1)
    \ee
    with the Lie brackets and co-cycles of $(\g_2 \bicrossl \g_1^*)^{cop}$ and $\g_2^* \bicrossr \g_1$ being
    \begin{align}
        & 
        \begin{aligned}
            [f^i , f^j] & = {d^{ij}}_k f^k
            \,\, , \quad
            [e^*_i , e^*_j] = 0 , \\
            [e^*_i , f^j] & = {\beta^{jk}}_i e^*_k , \\
        \end{aligned}
        &
        \begin{aligned}
            [e^i , e^j] & = {c^{ij}}_k e^k 
            \,\, , \quad
            [f^*_i , f^*_j] = 0 , \\
            [e^i , f^*_j] & = {\alpha^{ki}}_j f^*_k ,
        \end{aligned} \\
        &
        \begin{aligned}
            \delta_{\g_2 \bicrossl \g_1^*} (f^i) & = -
            {\alpha^{ij}}_k e^*_j \wedge f^k , \\
            \delta_{\g_2 \bicrossl \g_1^*} (e^*_i) & = - 
            \demi {c^{jk}}_i e^*_j \wedge e^*_k ,
        \end{aligned}
        &
        \begin{aligned}
            \delta_{\g_2^* \bicrossr \g_1} (e^i) & = 
            {\beta^{ki}}_j e^j \wedge f^*_k , \\
            \delta_{\g_2^* \bicrossr \g_1} (f^*_i) & = 
            \demi {d^{jk}}_i f^*_j \wedge f^*_k .
        \end{aligned}
        \label{Co-cycles_ClassicalDouble_SemiDual}
    \end{align}
    is a classical double. 
    
    Since the Lie bi-algebras $\g_2 \bicrossl \g_1^*$ and $\g_2^* \bicrossr \g_1$ are dual each other, their mutual actions are encoded in \eqref{Lie_action-back_action} and in the respective co-adjoint actions, that induce the cross Lie brackets
    \be
        \begin{aligned}
            [f^*_i , f^j] & = {d^{jk}}_i f^*_k + {\alpha^{jk}}_i e^*_k
            \,\, , \quad
            [e^i , e^*_j] = {c^{ki}}_j e^*_k + {\beta^{ki}}_j f^*_k , \\
            [e^i , f^j] & = -{\alpha^{ji}}_k f^k - {\beta^{ji}}_k e^k
            \,\, , \quad
            [e^*_i , f^*_j] = 0
        \end{aligned}
        \label{CrossActions_Semidual}
    \ee
    The co-cycle of the Lie bi-algebra $\b$ can be expressed in terms of the $r$-matrix
    \be
        r = e^i \ot e^*_i + f^*_i \ot f^i .
        \label{Classical_r-matrix_SemiDual}
    \ee
\end{proposition}

\begin{proof}
    Given the matched pair $(\g_1,\g_2)$ equipped with trivial co-cycles, we first build the bicross sum Lie bi-algebras $\g_2 \bicrossl \g_1^*$ and $\g_2^* \bicrossr \g_1$ as explained in Def. \ref{Def_BicrossSum}. \\
    Since these two Lie bi-algebras are dual each other, we can use them to construct a classical double as in Def. \ref{Def_ClassicalDouble}, or equivalently a matched pair as in Def. \ref{Def_MatchedPair}, with mutual actions given by the actions of the matched pair $(\g_1,\g_2)$, eq. \eqref{Lie_action-back_action}, plus the co-adjoint actions of $\g_1$ on $\g_1^*$ and of $\g_2$ on $\g_2^*$. One can check that the mutual action between $\g_2 \bicrossl \g_1^*$ and $\g_2^* \bicrossr \g_1$ satisfy the properties \eqref{MatchedPair_Condisions_1}-\eqref{MatchedPair_Condisions_2}; to carry out this computation it is enough to use the Jacobi identity for $\g_1$ and $\g_2$, plus the axioms \eqref{MatchedPair_Condisions_1}-\eqref{MatchedPair_Condisions_2} for the matched pair $(\g_1,\g_2)$. \\
    The co-cycles \eqref{Co-cycles_ClassicalDouble_SemiDual} can be derived through the co-boundary equation \eqref{Co-boundaryEq} using the classical $r$-matrix \eqref{Classical_r-matrix_SemiDual}.
\end{proof} 

We emphasize for the following that the Lie algebras $\g_1^*$ and $\g_2^*$ are always abelian Lie algebras. We note that $\d$ and $\b$ are identical as Lie algebras,  but we have been identifying two different types of sub Lie algebras. Indeed, the list of Lie brackets is always the same for $\d$ or $\b$ and we have sliced  the total algebra in two different ways.    $\d$ and $\b$ are however different as Lie bi-algebras.

Let us show now how the semi-dualized case can be interpreted in terms of the 2-group framework. 

\medskip 

\paragraph{Lie 2-bialgebra from semi-dualization. }
According to Thm. \ref{thm:match-crossed}, $\g\ltimes\h$ and $\h^*\ltimes \g^*$ are both Lie 2-bi-algebras if $\g$ and $\h^*$ form a matched pair. If we set 
\be
\g\equiv \g_1, \, \h^*\equiv \g_2, \, \g^*\equiv \g_1^*, \, \h\equiv \g_2^*,
\ee
then the theorem states that $\g_1\ltimes \g_2^*$ and $\g_2\ltimes \g_1^*$ are both Lie 2-bi-algebras if $\g_1$ and $\g_2$ form a matched pair, $\g_1\bowtie\g_2$, which is exactly the choice we made. More exactly, we also need to have that  the  action of $\g\equiv\g_1$ on $\h^*\equiv \g_2$ is the dual of the action of $\g\equiv\g_1$ on $\h\equiv \g_2^*$, and the action of $\h^*\equiv \g_2$ on $\g\equiv\g_1 $ is dual to the action of $\h^*\equiv \g_2$ on $\g^*\equiv\g_1^* $. This is exactly what the semi-dualization implements.

Having this in mind, we can reconsider the semi-dualized bi-algebra of Prop.\ \ref{Prop_Semidualization}. The components, $(\g_2 \bicrossl \g_1^*)$ and $(\g_2^* \bicrossr \g_1)$ can be viewed individually as infinitesimal crossed modules, or $2$-Lie bi-algebras, with trivial $t$-map. Hence we can interpreted $\b$ as a classical double, defined in terms of \textit{simple} infinitesimal crossed modules. 

\medskip 
We can exponentiate the Lie ($2$-)bi-algebras to recover groups or skeletal crossed modules. Upon exponentiation, we define
\begin{align}
      \g_1 \ltimes \g_2^* \rightarrow \cG= G_1\ltimes G_2^*, \quad    \g_2 \ltimes \g_1^* \rightarrow \cG^*= G_2\ltimes G_1^*.
\end{align}
Hence $\b$ becomes the group\footnote{Once again, we will neglect the possible subtleties that might arise if such splitting cannot be done globally. } $\cB=(G_1\ltimes G_2^*)\bowtie (G_2\ltimes G_1^*)$, which can be equipped with a symplectic Poisson bracket or a Poisson bracket which makes $\cB$ a Poisson Lie group (Drinfeld double). Since both $G_2^*$ and $G_1^*$ are abelian groups, both $(G_1\ltimes G_2^*)$ and $(G_2\ltimes G_1^*)$ can be viewed as skeletal Poisson crossed modules \cite{chen2013}. 

\medskip

As we already alluded, the Poincar\'e case is a key example of bicross bi-algebra. We are going to illustrate the construction in this case, and as such recover the Poincar\'e 2-group equipped with a non-trivial Poisson bracket.

\subsubsection{Example:  Poincar\'e case}
We use the metric $\eta^{\mu\nu}$ as the Minkowski metric. 
We take the Lie bi-algebra $\g_2$ to be the $n$ dimensional Lie algebra $\an$ with  generators $P^{\mu}$ having the same dimension as the scale $\kappa\mone$, satisfying\footnote{We are using the time-like deformation but other deformations, such as space or light-like are also possible \cite{Ballesteros:1995mi, Lukierski:1996sb}. The Euclidean signature can also be considered in a similar manner. }
\be
    [P^{\mu} , P^{\nu}] = \kmone (\eta^{\mu 0} P^{\nu} - \eta^{\nu 0} P^{\mu})
    \,\, , \quad
    \delta_{\an} (P^{\mu}) = 0 ,
\ee
and the Lorentz Lie bi-algebra $\g_1$ to be $\so$ of dimension $\frac{n(n-1)}{2}$, generated by $J^{\mu\nu}$, obeying
\be
    [J^{\mu\nu} , J^{\rho\sigma}] =
    \eta^{\mu\rho} J^{\nu\sigma} + 
    \eta^{\nu\sigma} J^{\mu\rho} - 
    \eta^{\mu\sigma} J^{\nu\rho} - 
    \eta^{\nu\rho} J^{\mu\sigma} 
    \,\, , \quad
    \delta_{\so} (J^{\mu\nu}) = 0,
\ee
The Lie bi-algebras $(\an,\so)$, equipped with the mutual actions
\begin{align}
    \triangleleft :
    & \quad \mathfrak{so} \ot \mathfrak{an} \to \mathfrak{so}
    &
    \triangleright :
    & \quad \mathfrak{so} \ot \mathfrak{an} \to \mathfrak{an} \nonumber \\
    & \quad J^{\mu\nu} \triangleleft P^{\rho} = \kmone (\eta^{\nu 0} J^{\mu\rho} - \eta^{\mu 0} J^{\nu\rho})
    &
    & \quad J^{\mu\nu} \triangleright P^{\rho} = \eta^{\mu \rho} P^{\nu} - \eta^{\nu \rho} P^{\mu}
    \label{actions_so-an}
\end{align}
 form a matched pair, as in Def. \ref{Def_MatchedPair}. 

Consider the dual Lie bi-algebras $\g_1^* \equiv \an^*$  and $\g_2^* \equiv \so^*$ generated by $P^*_{\mu}$ and $J^*_{\mu\nu}$ respectively. The bi-linear pairing is defined as
\be
    \la P^{\mu} , P^*_{\nu} \ra =
    {\eta^{\mu}}_{\nu}
    \,\, , \quad
    \la J^{\mu\nu} , J^*_{\rho\sigma} \ra = 
    {\eta^{\mu}}_{\rho} {\eta^{\nu}}_{\sigma} - {\eta^{\nu}}_{\rho} {\eta^{\mu}}_{\sigma} .
\ee
Using \eqref{Liedual:brackets_co-cycle}, we derive the $\an^*$ and $\so^*$ Lie bi-algebra structures:
\begin{align}
    & [P^*_{\mu} , P^*_{\nu}] = 0 
    \,\, , \quad
    \delta_{\an^*} (P^*_{\mu}) = \kmone P^*_{\mu} \wedge P^*_0 
\\
&    [J^*_{\mu\nu} , J^*_{\rho\sigma}] = 0
    \,\, , \quad
    \delta_{\so^*} (J^*_{\mu\nu}) = J^*_{\mu\rho} \wedge {J^*_{\nu}}^{\rho} ,
\end{align}
Finally, by dualization of the actions \eqref{actions_so-an}, we derive the co-actions
\begin{align}
    \alpha :
    & \quad \so^* \to \so^* \ot \an^*
    &
    \beta :
    & \quad \an^* \to \so^* \ot \an^* \nonumber \\
    & \quad \alpha(J^*_{\mu\nu}) = \kmone(J^*_{0 \mu} \ot P^*_{\nu} - J^*_{0 \nu} \ot P^*_{\mu})
    &
    & \quad \beta(P^*_{\mu}) = - {J^*_{\mu}}^{\rho} \ot P^*_{\rho}
    \label{co-actions_so-an}
\end{align}

\paragraph{Classical double of $\an \bowtie \so$.}
Following Prop. \ref{Prop_CotangentBundle}, from the matched pair $(\an,\so)$ with trivial co-cycles, we can construct the cotangent bundle of $\an \bowtie \so$ as a classical double
\be
    \d = (\an \bowtie \so) \ltimes (\an^* \blackbowtie \so^*) ,
\ee
equipped with the classical $r$-matrix
\be
    r = P^*_{\mu} \ot P^{\mu} + \demi J^*_{\mu\nu} \ot J^{\mu\nu} .
\ee

\paragraph{Semi-dualized case.}
According to Prop. \ref{Prop_Semidualization}, given the matched pair $(\an,\so)$ with trivial co-cycles, we can construct the classical double
\be
    \b = (\so \bicrossl \an^*)^{cop} \bowtie (\so^* \bicrossr \an) ,
\ee
equipped with the classical $r$-matrix
\be
    r = P^{\mu} \ot P^*_{\mu} + \demi J^*_{\mu\nu} \ot J^{\mu\nu} .
\ee
In particular, the Lie bi-algebras $\so \bicrossl \an^*$ and $\so^* \bicrossl \an$ are given by
\begin{align}
    &
    \begin{aligned}
        [P^*_{\mu} , P^*_{\nu}] & = 
        0 , \\
        [J^{\mu\nu} , J^{\rho\sigma}] & =
        \eta^{\mu\rho} J^{\nu\sigma} +
        \eta^{\nu\sigma} J^{\mu\rho} -
        \eta^{\mu\sigma} J^{\nu\rho} -
        \eta^{\nu\rho} J^{\mu\sigma} , \\
        [P^*_{\rho} , J^{\mu\nu}] & =
        {\eta^{\nu}}_{\rho} {P^*}^{\mu} - {\eta^{\mu}}_{\rho} {P^*}^{\nu} ,
    \end{aligned}
    &
    \begin{aligned}
        [P^{\mu} , P^{\nu}] & = 
        \kmone (\eta^{\mu 0} P^{\nu} - \eta^{\nu 0} P^{\mu}) , \\
        [J^*_{\mu\nu} , J^*_{\rho\sigma}] & =
        0 , \\
        [P^{\rho} , J^*_{\mu\nu}] & =
        \kmone ({\eta^{\rho}}_{\nu} J^*_{0 \mu} - {\eta^{\rho}}_{\mu} J^*_{0 \nu}) ,
    \end{aligned} \\
    &
    \begin{aligned}
        &
        \delta_{\so \bicrossl \an^*} (P^*_{\mu}) = 
        \kmone P^*_{\mu} \wedge P^*_0 , \\
        &
        \delta_{\so \bicrossl \an^*} (J^{\mu\nu}) = 
        \kmone P^*_{\rho} \wedge (\eta^{\nu 0} J^{\mu \rho} - \eta^{\mu 0} J^{\nu \rho}) ,
    \end{aligned}
    &
    \begin{aligned}
        &
        \delta_{\so^* \bicrossr \an} (P^{\mu}) =
        P^{\rho} \wedge {{J^*}^{\mu}}_{\rho} , \\
        &
        \delta_{\so^* \bicrossr \an} (J^*_{\mu\nu}) = 
        J^*_{\mu\rho} \wedge {J^*_{\nu}}^{\rho}.
    \end{aligned}
\end{align}
The cross Lie brackets, induced by the actions of the matched pair and by the co-adjoint actions, are
\be
    \begin{aligned}
        [P^{\rho} , J^{\mu\nu}] & = 
        \kmone (\eta^{\mu 0} J^{\nu\rho} - \eta^{\nu 0} J^{\mu\rho}) +
        (\eta^{\nu \rho} P^{\mu} - \eta^{\mu \rho} P^{\nu}) , \\
        [J^*_{\mu\nu} , P^*_{\rho}] & = 
        0 , \\
        [J^*_{\rho\sigma} , J^{\mu\nu}] & =
        \big(
        {\eta^{\nu}}_{\rho} {{J^*}^{\mu}}_{\sigma} +
        {\eta^{\mu}}_{\sigma} {{J^*}^{\nu}}_{\rho} -
        {\eta^{\nu}}_{\sigma} {{J^*}^{\mu}}_{\rho} -
        {\eta^{\mu}}_{\rho} {{J^*}^{\nu}}_{\sigma}
        \big)
        \nonumber \\
        & \quad +
        \kmone \big(
        \eta^{\nu 0} ({\eta^{\mu}}_{\rho} P^*_{\sigma} - {\eta^{\mu}}_{\sigma} P^*_{\rho}) -
        \eta^{\mu 0} ({\eta^{\nu}}_{\rho} P^*_{\sigma} - {\eta^{\nu}}_{\sigma} P^*_{\rho})
        \big) 
        , \\
        [P^{\mu} , P^*_{\nu}] & =
        - \kmone \big(
        \eta^{\mu 0} P^*_{\nu} + {\eta^{\mu}}_{\nu} P^*_0
        \big) 
        +
        {{J^*}^{\mu}}_{\nu} 
        .
    \end{aligned}
\ee
As a Lie algebra, $\so \bicrossl \an^*\cong \so\ltimes \mathbb{R}^n$ is the Poincar\'e/Euclidian Lie algebra. So upon exponentiation it would give the Poincar\'e/Euclidian Lie group. Since its co-cycle is non trivial, we would get the Poincar\'e group equipped with a non-trivial Poisson bracket. This is in fact the classical version of the so-called $\kappa$-Poincar\'e group \cite{Zakrzewski_1994}. 

Since it is only the geometrical interpretation that distinguishes the Poincar\'e group from the Poincar\'e 2-group, this construction allows to consider the Poincar\'e 2-group equipped with a non-trivial Poisson structure, giving rise upon quantization to a quantum strict 2-groups \cite{Majid:2012gy}.

\medskip

Similarly, the Lie bi-algebra $\so^* \bicrossl \an$ upon exponentiation gives rise to a group $\AN\ltimes\SO^*\cong \AN\ltimes \mathbb{R}^{{(n^2-n)}/{2}}$ equipped with a non-trivial Poisson structure. Upon quantization this  gives rise to the $\kappa$-Poincar\'e algebra \cite{Zakrzewski_1994, Kosinski}. Once again, this is a semi-direct product of groups, with $\SO^*$ being an abelian group. Hence this can also be interpreted as a skeletal crossed module, equipped with a non-trivial Poisson bracket. 

\medskip

Upon exponentiation, $\b$ becomes the Lie group noted $\cB$, which is isomorphic to $\cD$ as a Lie group, since as Lie algebras, $\b$ and $\d$ are the same.

The Poisson Lie group associated with $\b$ is  the Drinfeld double of the Poincar\'e group; the quantization of such Drinfeld double gives rise to the quantum double built as the double cross product of the $\kappa$-Poincar\'e group and the $\kappa$-Poincar\'e algebra. 
The exponentiation of $\b$, equipped with a symplectic Poisson structure, is the Heisenberg double of the Poincar\'e group $\ISO$, with a non-trivial Poisson brackets, with a non-abelian dual group $\ISO^*\cong \AN\ltimes\SO^*$.  

\medskip

Finally, as we emphasized many times now, the difference between the Poincar\'e group and Poincar\'e 2-group is mainly in the geometrical interpretation of the products. The construction we have presented can therefore be re-interpreted in terms of skeletal crossed modules and we have a Drinfeld double or Heisenberg double defined in terms of a pair of skeletal crossed module, with non-trivial Poisson structures. This will be the main object to construct the phase space for a three dimensional triangulation with decorations both on faces and edges.

\section{Polygon and polyhedron phase space}
\label{Sec_PolyPhaseSpace}\label{sec:polygon}

In this section we describe the phase space of a two dimensional discrete geometry. For simplicity, we focus on a triangulation, but our considerations extend directly to an arbitrary polygonal decomposition. We will focus on the flat case to set up the framework and then move to the case of a triangulation with homogeneously curved triangles. 

We will then recall how a polygon can be used to specify a polyhedron, thanks to the Minkowski theorem. This allows to obtain directly a phase space for a polyhedron.

\subsection{Polygon phase space}

First, we review the seminal work by Kapovich and Millson, restricted to the case of a triangle \cite{kapovich1996}. In the following we will deal the Euclidian picture. Discussion of a Lorentzian picture can be found in \cite{Livine:2018vxi}. 

Consider the $\bbR^3$ vector space in the $\su(2)$ basis: $\ell \equiv \vec \ell \cdot \vec \sigma$, where $\sigma_I$, $I=1,2,3$ are the Pauli matrices. The Killing form on $\su(2)$ allows to encode the scalar product:
\be
    \la \ell , \mathbf{y} \ra = \vec \ell \cdot \vec y 
    \quad \textrm{with} \quad
    \la \sigma_I, \sigma_J\ra =\delta_{IJ} .
\ee
Consider a set of three vectors $\ell_i$ in the Euclidian plane $\bbR^3$, which close, that is 
\be
     \cC=\ell_1 +\ell_2 +\ell_3 = 0. 
    \label{constraint0}
\ee  
This constraint defines a triangle. We interpret the vectors as edges of the triangle, with edge length $l_i\equiv|\ell_i|$. The vectors $\ell_i$ are equivalent to defining geodesics in $\bbR^3$. We note that the triangle can be embedded in $\bbR^3$ \textit{up to  global rotations and translations} in $\bbR^3$. 

We now introduce a phase space structure for such triangle. In the Hamiltonian formalism, a constraint is usually associated to a symmetry invariance \cite{dirac2001lectures}. We would like to pick the Poisson structure such that the geometric constraint $\cC$ becomes a momentum  map generating the symmetry freedom of the triangle embedding  in $\bbR^3$. In particular, we want that this constraint implements the rotational symmetry. 

As argued by Kapovich and Millson \cite{kapovich1996}, the following Poisson structure 
\be
    \label{poisson0}
    \poi{\ell^I_i , \ell^J_j} = \delta_{ij} {\epsilon^{IJ}}_K \ell^K_i
    \,\, , \quad 
    \poi{l_i , \ell^J_j} = 0,
\ee
is precisely such that $\cC$ generates the infinitesimal rotations in space. Hence $\cC$ can be viewed as the relevant momentum map. Note that since we have three edges, each with three components, the dimensions of this space is odd, hence it cannot have a symplectic Poisson bracket. Instead,  Kapovich and Millson \cite{kapovich1996} restricted the vectors to their direction only, so that their norm, the edge length, is fixed. As such the vector $\ell_i$ is an element of the sphere $\cS^2$ and it is well known that the Poisson bracket \eqref{poisson0} is symplectic on such sphere. Hence the final phase space $\cP^{KM}$ for a triangle with fixed edge lengths is given by  the symplectic reduction
\be
    \cP^{KM} = (S^2 \times S^2 \times S^2) /\!/ \SO(3).
\ee
From the gravity perspective, it would be natural to introduce the \textit{length} as a degree of freedom. However by doing so we do not  get a phase space in general, since the reduced space will have dimension $3N-6$, with $N$ being the total number of edges. 

To avoid this issue, the simplest way  is to add extra variables. The interest of doing so is twofold. 
\begin{itemize}
    \item We can introduce more geometric information which would  allow us to extend the phase space to a full 2d cellular decomposition for example. 
    \item We can have a realization of the rotational and \textit{translational} symmetries in terms of our phase space variables.
\end{itemize}
We can consider two different extensions which are actually related. The first consists of adding some holonomy information, and we extend $\bbR^3$ for each edge to $T^*\SO(3)\cong \SO(3)\ltimes \so^*(3)\cong \SO(3)\ltimes \mathbb{R}^3$. The chosen holonomy provides a way to transport vectors on a path dual to the edge. In this sense the holonomy contains some non-local information, which  allow us to discuss about adjacent triangles/polygons sharing the edge dual to the holonomy. 

The second extension consists of extending $\bbR^3$ to $\bbR^4\sim \bbC^2$, where the extra degree of liberty is a phase. The new phase space is therefore parametrized in terms of a pair of complex numbers $(z^-,z^+) \in \bbC^2$, usually called spinor variables (since they will transform as spinor under the infinitesimal rotations). This is a more local picture, working for each triangle/polygon and it is very convenient to describe the algebra of observables for a given polygon \cite{Girelli:2005ii, Freidel:2009ck, Dupuis:2013lka}. This is the root of the \textit{spinor approach} to loop quantum gravity \cite{Livine:2011zz, Livine:2011gp}. 

Let us focus on the first option.

\paragraph{Flat triangle phase space.}
Label the edges of the triangle by $1,2,3$ and associate to each of them a vector $\ell_i \in \bbR^3$. Denote the center of the triangle by $c$ and consider holonomies $h_{ci}\in \SO(3)$ joining the center to a point in the edge $i$. Such path is called \textit{(half) link}.

We take for phase space $T^*\SO(3)\sim \SO(3)\ltimes \R^3$, with symplectic Poisson bracket generated by the coordinate functions $\ell^I$ of $\so^*(3)\cong\bbR^3$ and $h \in\SO(3)$:
\be
    \poi{\ell^I , \ell^J} = \epsilon^{IJ}{}_K \ell^K
    \,\, , \quad 
    \poi{h^I{}_J , h^K{}_L} = 0
    \,\, , \quad 
    \poi{\ell^I , h^J{}_K} = (h \sigma^I)^J{}_K .
\ee 
We note that in the $\ell$ sector (``momentum sector") we have used the same Poisson bracket as  used by Kapovich and Millson. 

To construct the phase space for the full triangle, we consider three of these phase spaces and impose the constraint that the edges close. 
\be
 \ell_1 + \ell_2 + \ell_3 = 0.
\ee
Physically, the triangle phase space is nothing other than three spinning tops constrained to have a total angular momentum to be zero. The fact that geometry can be recovered from spinning tops with angular momentum conservation laws, was key to inception of the so-called spin networks by Penrose \cite{Penrose71angularmomentum:} to describe a quantum space.  

To generalize the construction to the curved case in particular, it is useful to remember that $T^*\SO(3)=\cD_0$ can be seen as a Heisenberg double.
The group element $\ell_1 + \ell_2 + \ell_3\in\bbR^3$ can be interpreted as the momentum map generating the global (left) rigid rotations denoted  $\SO^t(3)$. Hence we can define a new  phase space\footnote{We will note in the following $\cD^{(n)}\equiv \cD\times ... \times \cD= \cD^{\times n}$. } for the triangle $t$ as 
\be
    \cP_t^0 = ( T^*\SO(3) )^{\times 3} /\!/ \SO^t(3)
    \cong 
    \cD_0^{(3)}/\!/ \SO^t(3).
\ee

\paragraph{Curved triangle phase space.}
We move on to the curved case, focusing on the triangle with negative curvature.   We refer to \cite{treloar} for the case of positive curvature. The triangle with negative curvature has been treated in  \cite{Bonzom:2014wva, Charles:2016xzi}. Let us review the construction which is relevant to us \cite{Charles:2016xzi}.

The triangle edges are given in terms of geodesics in the hyperboloid. Since the hyperboloid can be seen as the quotient $\SL(2,\bbC)/\SU(2)\cong \AN(2)$.  Such geodesics can be characterized by elements/holonomies in $\AN(2)$ which we will again denote by $\ell$. 

Our strategy to construct a phase space is based on using Heisenberg doubles. Hence we deform $\bbR^3$ into the non-abelian group $\AN(2)$. The fact that the triangle is closed is then recovered in the closure constraint
\be
    \label{cons-ell}
    \ell_1 \ell_2 \ell_3 = 1 .
\ee
The new phase space associated to each edge is given by a deformation of the standard cotangent bundle into a Heisenberg double with curvature both in momentum and configuration space.  
\be
    \cD_0=T^*\SO(3)\cong  \SO(3)\ltimes \bbR^3 
    \,\,\,\dr \,\,\,
    \cD=\SL(2,\bbC) \cong \SU(2)\bowtie \AN_2.
\ee

By definition, $ \ell_1 \ell_2 \ell_3$ is still the momentum map generating the global (left) rotations $\SU^t(2)$. Hence we can again define the symplectic reduction to define the phase space for a curved triangle.  
\be
    \cP = \SL(2,\bbC)^{\times 3} /\!/ \SU(2)\cong \cD^{(3)} /\!/ \SU^t(2).
\ee

 \medskip 

We can then extend the construction even more. We can consider for every pair of link/edge a Heisenberg double
 $\cD = G \bowtie G^* = G^* \bowtie G$, endowed with the ribbon factorization \eqref{ribbon}, with actions given by \eqref{actions from ribbons} and a symplectic form \eqref{SymplecticForm}.
By convention, $G$ decorates the link, while $G^*$ decorate the edge and provide the information to reconstruct the geometry of the polygon. If $G^*$ is abelian, we will say that the polygon is flat.

The phase space of the polygon with $n$ edges is given by the symplectic reduction \be
    \cP = \cD^{(n)} /\!/ G^t ,
\ee
where the constraint $G^*\ni\ell_1..\ell_n=1$ is the momentum map generating the global left $G$ transformations, denoted $G^t$.

We note that once we generalize the construction to have possibly non-abelian groups both in configuration and momentum spaces, the construction becomes symmetric. In particular, we can implement a global translational symmetry instead of a global rotational symmetry by symplectic reduction over $G^*$ which could be either of $\R^3$ or $\AN_2$ for example. The momentum map implementing the symplectic reduction is the flatness constraint $g_1...g_n=1$, $g_i\in \SU(2)$ \cite{Delcamp:2018sef}.

\paragraph{Reference frames for 2d discrete geometry decorated by Heisenberg doubles.}
To glue triangles or polygons together to construct the phase space for a 2d cellular complex, it is useful to define the notion of a reference frame, to make sure we identify quantities associated to different triangles consistently. 

Let us consider an edge $e = (v_1,v_2)$ connecting vertices $v_1$ and $v_2$ and its dual, called a \textit{link}, $l = (v_1,v_2)^*=(c_1,c_2)$, which connects 
nodes $c_1$ and $c_2$.

As mentioned earlier, we decorate the links with $G$ elements and the edges with $G^*$ elements  so that a Heisenberg double $\cD$ decorates the (edge, link) pair.

\medskip 

\textit{Representations. }
Both the edge and the link carry an orientation with one vertex being the source or ``initial vertex" of the edge and the other being the target or ``final vertex", and similarly for the nodes of a link. We can use  representations of $\cD$ to induce representations $(\varphi, \rho)$ of  $G$ and $G^*$ respectively. 

\begin{itemize}
    \item On the initial node (source) of the link, $c_1$, we consider a vector space $V_{c_1}$ which carries a representation $\rho^{c_1}$ of the group $G^*$ coming from the right decomposition $\cD= G^*\bowtie G$.
    $$\ell \dr\rho^{c_1}(\ell) \equiv \ell^{c_1}
    $$

    \item On the final node (target) of the link, $c_2$,  we consider a vector space $V_{c_2}$ which carries a representation $\rho^{c_2}$ of the group $G^*$ coming from the left decomposition $\cD= G\bowtie G^*$.   
    $$\tell \dr \rho^{c_2}(\tell) \equiv \tell^{c_2}
    $$
    
    \item On the initial vertex (source) of the edge, $v_1$, we consider a vector space $W_{v_1}$ which carries a representation $\varphi^{v_1}$ of the group $G$ coming from the right decomposition $\cD= G \bowtie G^*$.
    $$\th \dr \varphi^{v_1}(\th) \equiv \th^{v_1}
    $$
    
    \item On the final vertex (target) of the edge, $v_2$, we consider a vector space $W_{v_2}$ which carries a representation $\varphi^{v_2}$ of the group $G$ coming from the left decomposition $\cD= G^* \bowtie G$.
    $$h \dr  \varphi^{v_2}(h) \equiv h^{v_2}
    $$
\end{itemize}

\medskip

\textit{Action and parallel transport. }
The left or right mutual actions between $G$ and $G^*$ can be interpreted as parallel transport.  For example, the relation 
\be 
    \ell = \th \rhd \tell= \th \tell h\mone 
    \,\,\, \Leftrightarrow \,\,\,
    \rho^{c_1}(\ell)= \varphi^{v_1}(\th)\rhd \rho^{c_2}(\tell)
\ee 
encodes how the edge variable  is transported from the representation at $c_2$ to the representation at $c_1$, through the representation of $h$  at $v_1$. 

The left action of $G$ on $G^*$ (resp. $G^*$ on $G$) transports the representation of the edge (resp. link) from the final node $c_2$ (resp. final vertex $v_2$) to the initial node $c_1$ (resp. initial vertex $v_1$) of the link (resp. edge); the right action of $G$ on $G^*$ (resp. $G^*$ on $G$) transports the representation of the edge (resp. link) from $c_1$ to $c_2$ (resp. from $v_1$ to $v_2$). We summarize the role of left and right actions in Table \ref{Table_2dActions}.
\begin{table}
\centering
    \begin{tabular}{ | r | | c | c || }
    \cline{2-3}
    \multicolumn{1}{c|}{}
    & 
    \textit{\textbf{$G$ on $G^*$}}
    &
    \textit{\textbf{$G^*$ on $G$}} \\
    \hline
    \textit{\textbf{Left action $\rhd$}}
    &
    $c_2 \to c_1$ & $v_2 \to v_1$ \\
    \hline
    \textit{\textbf{Right action $\lhd$}}
    &
    $c_1 \to c_2$ & $v_1 \to v_2$ \\    
    \hline
    \end{tabular}
    \caption{Geometric interpretation of the left and right mutual actions between $G$ and $G^*$. Geometrically, the left action of a link (resp. edge) on the edge (resp. link) transports the representation of the edge (resp. link) from the node $c_2$ to $c_1$ (resp. from the vertex $v_2$ to $v_1$); the right action of a link (resp. edge) on the edge (resp. link) transports the representation from $c_1$ to $c_2$ (resp. from $v_1$ to $v_2$).}
    \label{Table_2dActions}
\end{table}
Moreover, note that the left action of $\th$ (resp. $\ell$) is the inverse of the right action of $h$ (resp. $\tell$):
\be
    \begin{aligned}
        &
        \ell = \th \rhd \tell = \th \rhd (\ell \lhd h) , \\
        &
        \tell = \ell \lhd h = (\th \rhd \tell) \lhd h , \\
    \end{aligned}
    \qquad \qquad
    \begin{aligned}
        &
        h = \th \lhd \tell = (\ell \rhd h) \lhd \tell , \\
        &
        \th = \ell \rhd h = \ell \rhd (\th \lhd \tell) .
    \end{aligned}
\ee
If we change the orientation of the link, $c_2\dr c_1$, which implies a change of orientation for the edge $v_2\dr v_1$, we use the inverse map for the full Heisenberg double.
\be
    d = \ell h = \th\tell 
    \,\,\,\dr\,\,\,
    d\mone = \tell\mone \th\mone = h\mone \ell\mone.
\ee
This agrees with the interpretation of left and right actions given in Table \ref{Table_2dActions}. For example, the relation 
\be
    \tell\mone = h\mone \rhd \ell\mone
    \,\,\, \Leftrightarrow \,\,\,
    \rho^{c_2}(\tell\mone) = \varphi^{v_2}(h\mone) \rhd \rho^{c_1}(\ell\mone) 
\ee
encodes how the the inverse edge variable is transported from the representation at $c_1$ (final node of $h\mone$ to the representation at $c_2$ (initial node of $h\mone$), through the representation of the inverse link at $v_2$.

\medskip

Specifying these different representations is also at the root of the spinor approach \cite{Livine:2011zz, Livine:2011gp}. Indeed in this case, one deals explicitly with the spinor representations and can reconstruct the full geometric data from the spinor data. \

This identification of representation is important to understand how to glue/fuse terms. Indeed we need to identify the gluing constraint for terms living in the same representation.

\paragraph{Phase space for a $2$ dimensional triangulation.}
We intend now to construct the phase space for a triangulation. The main idea is to glue/fuse triangles along common edges and build from there the total phase space. Since an edge is shared by a pair of triangles, the gluing is done pairwise.

Consider two triangles dual to the nodes $c_1$ and $c_2$. The triangle $c_1^*$ has an edge denoted $e$ and the triangle $c_2^*$  has an edge denoted $e'$.
We denote by $c_e$ (resp.\ $e_{e'}$)the intersection of edge $e$ (resp.\ $e'$) and its dual half link, so that the half link dual to the edge $e$ (resp.\ $e'$) is denoted $(c_1c_e)$ (resp. $(c_2c_{e'})$). We denote by $\cD_{c_1,e}\cong G\bowtie G^*$ (resp. $\cD_{c_2,e'} \cong G\bowtie G^*$) the Heisenberg double associated to  the edge $e$ and half link  $(c_1c_{e})$ (resp. to  the edge $e'$ and half link  $(c_2c_{e'})$). 

We therefore have the pair of ribbon constraints 
\be
    \cD_{c_1,e}\ni 
    \ell^{c_1} h^{v_2} = \th^{v_1} \tell^{c_e}
    \,\,, \quad
    \cD_{c_2,e'}\ni
    \ell^{c_2} h^{v'_2} = \th^{v'_1} \tell^{c_{e'}} .
\ee

To glue $c_1^*$ and $c_2^*$ we must identify $e=(v_1v_2)$ and $e'=(v'_1v'_2)$. To make the identification, we must set the corresponding edge decorations equal in some common representation. This is done by identifying the nodes $c_e$ and $c_{e'}$. This means that the orientation of the edges being identified are opposite,
\be\label{edge-identification}
    e=-e'
    \,\,\,\to\,\,\, v_1=v'_2 \,,\,\, v_2=v'_1 .
\ee
As such, at the node $c_{e'}$, we have a representation $\rho^{c_{e'}}(\tell_{{v_2'v_1}'}\mone)$ and at $c_e$ we have the representation $\rho^{c_{e}}(\tell_{{v_1v_2}})$.
The identification condition is therefore that 
\begin{align}
    \tell_{v_1'v_2'}^{c_{e'}} = \tell_{v_2v_1}^{c_{e'}} =\tell_{v_1v_2}^{c_e}\\
    \tell^{c_e} = (\tell^{c_{e'}})\mone.
\end{align}

We recover the constraint that we used to glue ribbons in Proposition \ref{prop:glue-ribbon}. This is the momentum generating the left global $G$ transformations and we can perform the symplectic reduction.   
\be
    (\cD_{c_1,e}\times \cD_{c_2,e'}) /\!/ G\cong \cD_{c_1c_2,e}=\cD_{l,e}, \quad l=(c_1c_2). 
\ee
In particular, putting the two ribbon constraints together with the gluing constraint, we identify the reduced variables,
\begin{align}
    & 
    \left.\begin{array}{c} 
    (\th^{v_1})\mone  \ell^{c_1} h^{v_2} = \tell^{c_e} , 
    \\ 
    (h^{v'_2})\mone (\ell^{c_2})\mone \th^{v'_1} = (\tell^{c_{e'}})\mone
    \end{array}\right\}
    \,\, \stackrel{v_1=v'_2, \, v_2=v'_1}{\longrightarrow} \,\, 
    \ell^{c_1} h^{v_2} (\th^{v_2})\mone
    =
    \th^{v_1} (h^{v_1})\mone (\ell^{c_2})\mone
    \\
    & \qquad
    \Leftrightarrow \,\,\, 
    \ell^{c_1} h^{v_2}_l = \th^{v_1}_l (\ell^{c_2})\mone
    \,\,,\quad 
    h^{v_2}_l := h^{v_2} (\th^{v_2})\mone
    \,, \,\,\,
    \th^{v_1}_l := \th^{v_1} (h^{v_1})\mone .
\end{align}
The phase space obtained after gluing is simply the Heisenberg double with $G$ decorations on the (full) link $l = (c_1c_2)$ and $G^*$ decorations on the edge $e = (v_1v_2)$. See Fig. \ref{fig:rib+triangle}. Note that a similar construction was done at the \textit{quantum} level using fusion product in \cite{Delcamp:2018sef}.  Note also that,  unsurprisingly, we have recovered the variables obtained after discretization of 3d gravity with a non-zero cosmological constant \cite{Dupuis:2020ndx}. These variables are also called triangle operators in Kitaev's model \cite{Kitaev1997}. 
\medskip\\
By repeating this gluing we can recover the full phase space $\cP^{tot}$ for the 2d triangulation $\cT$ with dual $\cT^*$.
\be
    \cP^{tot} = \big(\bigtimes_{l\in \cT^*} \cD_{l,e}\big) /\!/ \big(\bigtimes_{t\in \cT} G^{t}\big),
\ee
where $\cD_{l,e}$ are the Heisenberg doubles for each full link of the dual complex of the triangulation and $G^t$ is the global $G$ transformation induced by the closure constraint of each triangle $t$ of the triangulation. 
\begin{figure}
    \begin{center}
    \begin{tikzpicture}
    
    \node[anchor=south west,inner sep=0] (image) at (0,0)
    {\includegraphics{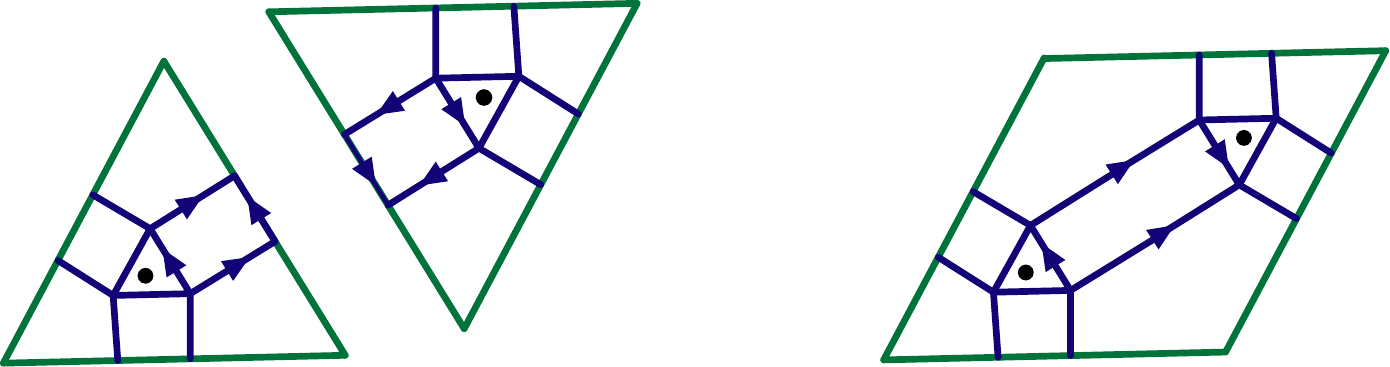}};
    \begin{scope}[x={(image.south east)},y={(image.north west)}]
    \node at (0.14, 0.85) {$v_2$}; 
    \node[] at (0.14,0.52) {$h_1$};
    \node[] at (0.18,0.21) {$\tilde h_1$};
    \node[] at (0.14,0.3) {$\ell_1$};
    \node[] at (0.21,0.45) {$\tell_1$};
    \node[] at (0.27, 0.8) {$\th_2$};
    \node[] at (0.25,0.52) {$\tell_2$};
    \node[] at (0.33,0.47) {$h_2$};
    \node[] at (0.31,0.67) {$\ell_2$};  
    \node[] at (0.45,0.5) (glue1) {};
    \node[] at (0.65,0.5) (glue2) {};
    \tikzstyle{arrow} = [thick,->,>=stealth]
    \draw[arrow] (glue1) edge (glue2);
    \node[] at (0.78,0.63) {$h_{12}$};
    \node[] at (0.85,0.28) {$\th_{12}$};
    \node[] at (0.78, 0.33) {$\ell_1$};
    \node[] at (0.86,0.57) {$\ell_2$};
    \node at (0.25, -0.05) {$v_1$};
    \node at (0.32,0.1) {$v_2'$};
    \node at (0.18,1) {$v_1'$};
    \end{scope}

    \end{tikzpicture}
    \caption{The phase space of a triangle is given by a set of three ribbons. The phase space of two triangles is obtained by gluing ribbons associated to the shared edge.}
    \label{fig:rib+triangle}
    \end{center}
\end{figure}
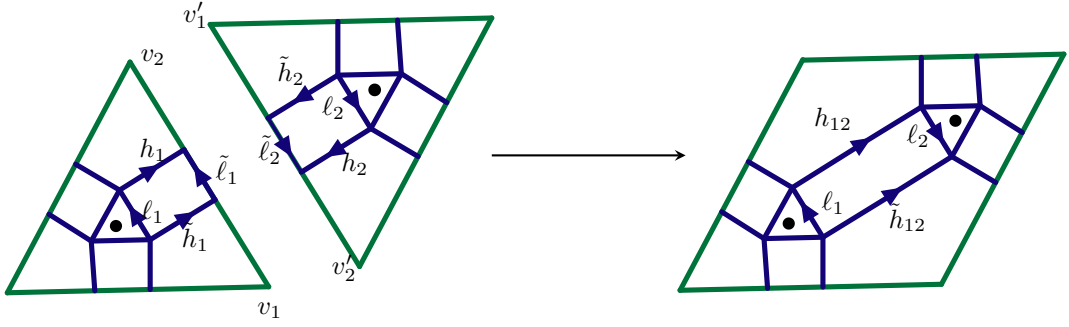

Dynamics can then be implemented by imposing some constraint on the closed $G$ holonomies. From the usual gravity picture, this would be  typically the flatness constraint. If we swap what we call $G$  and $G^*$, dynamics would be implemented by implementing that closed $G^*$ holonomies should be flat, this would be implementing the flatness constraint.

\subsection{Polyhedron phase space using groups}
\label{Sec_PolyhedronPhaseSpace}

Minkowski proved in 1897 that given a set of vectors $\ell_i\in \bbR^3$ which close to form a (non-planar) polygon, 
\be\label{polygon}
  \sum_i \ell_i=0
\ee
then we can  uniquely (up to translations and  rotations) reconstruct a polyhedron, provided we interpret these vectors $\ell_i$ as the normals of the (flat) faces of the polyhedron, and their magnitude to be the area of the relevant face \cite{alexandrov}. 
Given such set of vectors, we can use the Lasserre reconstruction program --- revisited in \cite{Bianchi_2011}, \cite{Sellaroli:2017wwc} --- to reconstruct the polyhedron.  

For our concern, it means that the polygon data is sufficient to reconstruct a polyhedron. Hence, all the phase space structures we discussed in the previous section can be also used to describe a polyhedron or a 3d triangulation (or a 3d cellular decomposition). 

This equivalence between polyhedron and polygon reappears at the the quantum level in the ambiguity of interpreting an intertwiner. If one does not specify the dimension at first, we can interpret it either a quantum state of a polygon or of a polyhedron. This extends for the full spin network, it can encode the quantum state of a 2d triangulation or of a 3d triangulation.   

Let us sketch how the construction of the phase space works for a  3d triangulation $\cT$ and its dual $\cT^*$.

We focus first on a link and its dual, given by the triangle. We consider the Heisenberg double $\cD\cong G\bowtie G^*$, with both $G$ and $G^*$ three dimensional Lie groups. $G$ elements still decorate the links in $\cT^*$, which are dual to the triangles/faces in the   triangulation; hence, $G^*$ elements decorate the faces. There are no decorations on the edges of the triangle meaning that $G^*$ must be abelian in order to have a consistent way of composing face decorations, due to the Eckmann-Hilton argument, see Fig \ref{eckmann}. The Heisenberg double must therefore be of the type $\cD\cong T^*G$ (if we deal with a \textit{Lie group} $G$). 
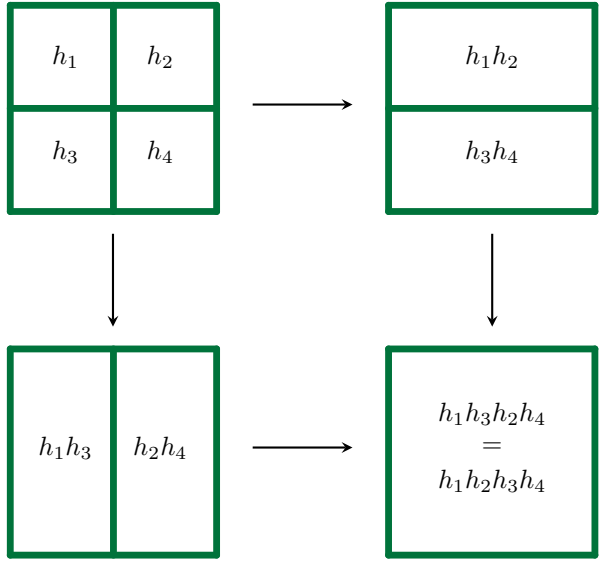
\begin{figure}
    \begin{center}
    \begin{tikzpicture}
    
    \node[anchor=south west,inner sep=0] (image) at (0,0)
    {\includegraphics{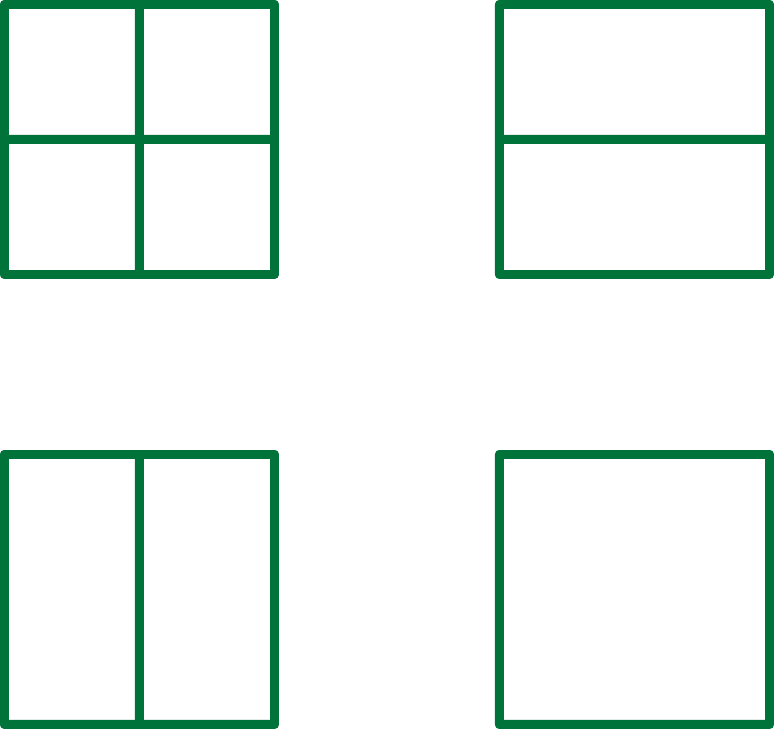}};
    \begin{scope}[x={(image.south east)},y={(image.north west)}]
    \node at (0.1,0.9) {$h_1$};
    \node at (0.26,0.9) {$h_2$};
    \node at (0.1,0.73) {$h_3$};
    \node at (0.26,0.73) {$h_4$};
    
    \node at (0.82,0.9) {$h_1 h_2$};
    \node at (0.82,0.73) {$h_3 h_4$};
    
    \node at (0.1, 0.2) {$h_1 h_3$};
    \node at (0.26,0.2) {$h_2 h_4$};
    
    \node at (0.82,0.26) {$h_1h_3h_2h_4$};
    \node at (0.82,0.2) {$=$};
    \node at (0.82, 0.14) {$h_1h_2h_3h_4 $};

    \node at (0.18,0.6) (glue1) {};
    \node at (0.18,0.4) (glue2) {};
    \node at (0.82,0.6) (glue3) {};
    \node at (0.82,0.4) (glue4) {};
    
    \node at (0.4, 0.815) (glue5) {};
    \node at (0.6, 0.815) (glue6) {};
    \node at (0.4, 0.2) (glue7) {};
    \node at (0.6, 0.2) (glue8) {};
    
    \tikzstyle{arrow} = [thick,->,>=stealth]
    \draw[arrow] (glue1) edge (glue2);
    \draw[arrow] (glue3) edge (glue4);
    \draw[arrow] (glue5) edge (glue6);
    \draw[arrow] (glue7) edge (glue8);
    \end{scope}

    \end{tikzpicture}
    \caption{The Eckmann-Hilton argument shows that in order to compose faces so that the above diagram commutes, the face decorations must commute. This can be avoided by decorating both edges and faces by 2-group elements. }
    \label{eckmann}
    \end{center}
\end{figure}
In the 3d case, the magnitude of the vector in $\ell\in G^*$ encodes the area instead of  a length in 2d, and the direction of the vector itself, instead of being the  edge direction, will indicate the (outgoing) direction normal to the face.

Given four links rooted at the common root $c$, we interpret the constraint \eqref{polygon} as the closure of the tetrahedron. 
From the phase space perspective, the closure of the polyhedron encoded in the constraint \eqref{polygon} is still a constraint on the momentum map generating $G$ left global transformations, which we denote $G^\tau$. The phase space of a tetrahedron $\tau$ can then be recovered from the symplectic reduction 
\be
    \cP = (T^*G)^{\times 4} /\!/ G^\tau.
\ee
Physically, the tetrahedron phase space is nothing else than four spinning tops constrained to have a total angular momentum to be zero. The difference between the phase space of a quadrilateral (ie a polygon) is simply in the way we geometrically interpret the angular momentum $\ell$. In the polygon case, it is associated to the edge, while in the tetrahedron case, it is associated to the normal of the face.  

To get the full phase space of a 3d triangulation in this context, we glue tetrahedra by identifying the face information. The mathematical description is exactly the same as for gluing triangles. This is again the ambiguity between 2d and 3d showing. 

\medskip  

One might argue that this ambiguity between 2d and 3d shows that such description of the phase space of a 3d triangulation is not the best one. We are likely missing some information for the 3d description. Furthermore, for a gravity point of view it is a bit of a draw-back to not have access to some edge decoration, which should encode some information about the metric. 
Finally, to actually define the notion of curved polyhedron, we need to have some data decorating its edges, just like we defined a curved triangle by decorating its edges with a non-abelian group \footnote{\cite{Haggard:2015ima, Riello:2017iti} generalized the Minkowski theorem with non-abelian group decorations apparently solely on the faces. However, it is likely that some edge dependence was actually hidden \cite{Riello}.} 

All these arguments point to the need to decorate the edges of the triangulation.

In order to decorate both the faces and the edges, we require compatibility rules such as those given by the 2-group structure in Sec.\ \ref{Sec_2-groups}. 
2-groups have been argued to be the natural generalization of the notion of (quantum) groups to discuss topological invariant of a 4d manifold (hence in terms of states describing 3d triangulations). Just as quantum groups are the natural object to deal  with (homogeneously) curved 3d manifold, one might expect that quantum 2-groups will be the right structure to deal with (homogeneously)  curved 3d manifold. 
Finally, it was argued in \cite{Mikovic:2011si, Asante:2019lki} that a maybe more adequate construction of spin foam model for 4d gravity could be done using 2-group structure. 

As a summary, it seems natural to construct the notion of phase space for a 3d triangulation using 2-groups. At this time, we will  still be  focusing on abelian decoration on the faces, that is we focus on \textit{simple} 2-groups. As a consequence, we will not put curvature on the triangles \textit{yet}. Nevertheless, we expect that our set up could be generalized to allow curved triangles.

\section{Polyhedron phase space from skeletal 2-groups}
\label{Sec_2-groupPhaseSpace}

\subsection{General construction}

In this section we explain how to build the phase space of a $3$ dimensional (curved) cellular decomposition as the Heisenberg double $\cB = \cG \bowtie \cG^*\cong \cG^*\bowtie \cG$, given by the double cross product of two skeletal crossed modules, $\cG$ and $\cG^*$, as explained in Sec.\ \ref{Sec_2-groups}. 

According to the $2$-group picture introduced in Sec.\ \ref{Sec_2-groups}, we can interpret each crossed module as decorating 1d and 2d objects, as part of the cellular decomposition $\mathcal{T}$ and its dual $\mathcal{T}^*$.  The skeletal crossed module $\cG^*$ decorates $\mathcal{T}$ and $\cG$ decorates $\mathcal{T}^*$.

Just like for the 2d case, it is convenient to fix the reference frame conventions in order to glue/fuse structures in a consistent manner.

\subsubsection{Reference frames for 3d geometries decorated by skeletal crossed modules.}
\label{Sec_3dRefFrames}

We consider an (oriented) edge $e=(v_1v_2)$, a face $t_e$ 
which contains the edge $v_1$ and $v_2$ in its boundary. Dually  we deal with a node $c_1$, and the node $c_t$ which is a point somewhere on the interior of the face $t_e$. The wedge $w_l$ is a dual face which contains $l = (c_1c_t)$, called a (half) link, in its boundary. This set up is shown in Fig.\ \ref{double-wedge}.   

\begin{figure}[H]
\begin{center}
\begin{tikzpicture}
    \node[anchor=south west,inner sep=0] (image) at (0,0) {\includegraphics[width=0.3\textwidth]{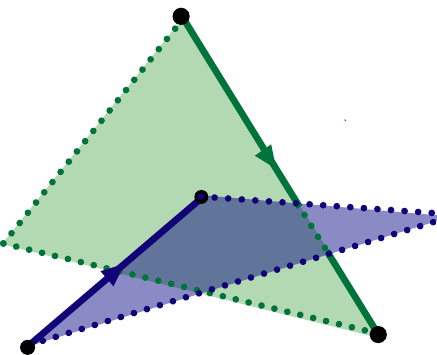}};
    \begin{scope}[x={(image.south east)},y={(image.north west)}]
        \node[] at (0, 0)     (c1)   {$c_1$};
        \node[] at (0.5,0.5)  (ct)   {$c_t$};
        \node[] at (0.95,0.0) (v1)   {$v_1$};
        \node[] at (0.48,1)   (v2)   {$v_2$}; 
        \node[] at (0.65,0.6) (e)    {$e$};
        \node[] at (0.25,0.3) (l)    {$l$};
        \node[] at (0.35,0.65)(te)   {$t_e$};
        \node[] at (0.55,0.35)(wl)   {$w_l$};
    \end{scope}
    
\end{tikzpicture}
\caption{The building block we will glue to reconstruct a full 3d triangulation.  The (half) link and wedge will be decorated by $\cG$ elements, while the edge and face will be decorated by elements in $\cG^*$.
    }
    \label{double-wedge}
\end{center}
\end{figure}
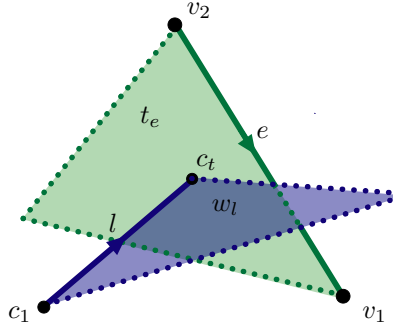

The half link $l = (c_1c_t)$ is decorated with elements in $G_1$, while the wedge $w_l$ is decorated with elements in the abelian group $G_2^*$. 
The edge $e=(v_1v_2)$ is decorated with elements in $G_2$, while the face $t_e$ is decorated with elements in the abelian group $G_1^*$. 
\medskip \\
\textit{2-representations. }
In the 2d case, the  source (resp. target) of a group element/holonomy carries a representation of the dual group.
Since $\cG$ and $\cG^*$ are skeletal crossed modules, we can use a similar picture, using 2-vector spaces and 2-representations \cite{BAEZ1997125, Baez:2008hz}. For our purpose, it is not useful to get the full details on the nature of a 2-representation since we use it more as a book keeping device to keep track of which structures we are fusing. It is enough to say that nodes and vertices (sources and targets of links and edges) carry 2-representations of the dual skeletal crossed module.
The left and right decompositions of $\cB$ into $\cG$ and $\cG^*$ are distinguished using a tilde ($\tilde{\phantom{a}}$).
\begin{itemize}
    \item The source of the link (node $c_1$) carries a 2-representation of the skeletal crossed module $\cG^*$, $\rho^{c_1}(\ell)=\ell^{c_1}$, coming from the right decomposition $\cB = \cG^*\bowtie \cG\ni \ell h$. 
    
    \item The target of the link (node $c_t$) carries a 2-representation of the skeletal crossed module $\cG^*$, $\rho^{c_t}(\tell)=\tell^{c_t}$, coming from the left decomposition $\cB= \cG\bowtie \cG^*\ni \th\tell$.
    
    \item The source of the edge (vertex $v_1$) carries a 2-representation of the skeletal crossed module $\cG$, $\varphi^{v_1}(\th)=\th^{v_1}$, coming from the left decomposition $\cB = \cG \bowtie \cG^* \ni  \th \tell$.   
    
    \item The target of the edge (vertex $v_2$) carries a 2-representation of the skeletal crossed module $\cG$, $\varphi^{v_2}(h)=h^{v_2}$, coming from the right decomposition $\cB = \cG^* \bowtie \cG \ni \ell h$.
\end{itemize}
These 2-representations can be seen as the generalization of the triangular operators appearing in the Kitaev model \cite{Kitaev1997}. 

The skeletal crossed modules $\cG$ and $\cG^*$, in turn, can be split in the left or right decompositions. Different decompositions of $\cG$ (resp. $\cG^*$) are decoration of the same pair edge/face (resp. link/wedge), with the face (resp. wedge) rooted either in the source or in the target of the edge (resp. link).
The different decompositions of $\cG$ and $\cG^*$ are distinguished by use of an overline ($\overline{\phantom{a}}$). In 2-group theory, changing the decomposition of a  crossed module can be viewed as whiskering, see Fig.\  \ref{fig:changesource}. 
\begin{equation*}
    \cB = \cG^* \bowtie \cG = \cG \bowtie \cG^*
    \quad \to \quad
    \ell h = \th \tell
\end{equation*}
\begin{align*}
    &
    \cG \ni 
    h = 
    \begin{cases}
        &
        u y \in G_1 \ltimes G_2^* \\
        &
        \oy u \in G_2^* \rtimes G_1
    \end{cases}
    &
    \cG \ni 
    \th = 
    \begin{cases}
        &
        \tu \ty \in G_1 \ltimes G_2^* \\
        &
        \toy \tu \in G_2^* \rtimes G_1
    \end{cases} \\
    &
    \cG^* \ni
    \ell =
    \begin{cases}
        &
        \beta \lambda \in G_1^* \rtimes G_2 \\
        &
        \lambda \obe \in G_2 \ltimes G_1^*
    \end{cases}
    &
    \cG^* \ni
    \tell =
    \begin{cases}
        &
        \tbe \tla \in G_1^* \rtimes G_2 \\
        &
        \tla \tobe \in G_2 \ltimes G_1^*
    \end{cases}
\end{align*}
Since there is no back action of  $G^*_1$ on $G_2$ (resp. of $G^*_2$ on $G_1$), we note that $u=\bu$, $\tu=\tbu$, $\lambda=\blam$ and $\tla=\tblam$.

\begin{itemize}
    \item Take the skeletal crossed module $\cG$ that decorates a pair link/wedge.
    
    \begin{itemize}
        \item The right decomposition $\cG = G_2^* \rtimes G_1$ implies that $G_2^*$ elements (decoration of wedges) are rooted at the source of the link (node $c_1$).
        \item The left decomposition $\cG = G_1 \ltimes G_2^*$ implies that $G_2^*$ elements (decoration of wedges) are rooted at the target of the link (node $c_t$). 
    \end{itemize}
    
    \item Take the skeletal crossed module $\cG^*$ that decorates a pair edge/face.
    
    \begin{itemize}
        \item The right decomposition $\cG = G_1^* \ltimes G_2$ implies that $G_1^*$ elements (decoration of faces) are rooted at the source of the edge (vertex $v_1$). 
        \item The left decomposition $\cG = G_2 \ltimes G_1^*$ implies that $G_1^*$ elements (decoration of faces) are rooted at the target of the edge (vertex $v_2$). 
    \end{itemize}
\end{itemize}
Therefore, the choice of decoration of an edge (resp. link) is specified by a single node (resp. vertex) as in the 2d case; while decorations of faces and wedges are specified by a pair of node and vertex.  
In Table \ref{Table_3dRefFrames} we list the decorations ofor the pair edge/face, as well as the pair link/wedge, according to the choice of representation and root.
\begin{table}
\centering
    \begin{tabular}{ | l | | c | c | }
    \cline{2-3}
    \multicolumn{1}{c|}{} 
    & \textit{ \textbf{Source of $l$: $c_1$}} & \textit{\textbf{Target of $l$: $c_t$}} \\
    \hline \hline
    \textit{\textbf{Source of $e$: $v_1$}} &
    $\ell^{c_1} = \beta^{c_1}_{v_1}\, \lambda^{c_1} \,\,,\,\, \th^{v_1} = \toy^{v_1}_{c_1}\, \tu^{v_1} $ &
    $\tell^{c_t} = \tbe^{c_t}_{v_1}\, \tla^{c_t} \,\,,\,\, \th^{v_1}= \, \tu^{v_1} \ty_{c_t}^{v_1} $ \\
    \hline
    \textit{\textbf{Target of $e$: $v_2$}} &
    $\ell^{c_1} = \lambda^{c_1}  \,\obe^{c_1}_{v_2}  \,\,,\,\, h^{v_2}  = \oy_{c_1}^{v_2}  \,  u^{v_2} $ &
    $\tell^{c_t} = \tla^{c_t} \,  \tobe^{c_t}_{v_2} \,\,,\,\, h^{v_2} = u^{v_2} \, y_{c_t}^{v_2} $ \\
    \hline \hline
    \end{tabular}
    \caption{List of the group elements used to decorate edges, links, faces and wedges. Decorations of links (resp. edges) are labelled by a single index, related to vertex (resp. node) where the link (resp. edge) is represented. Decorations of wedges (resp. faces) are labelled by an index that specifies the vertex (resp. node) where it is represented, and a subscript that specifies the node (resp. vertex) where it is rooted.}
    \label{Table_3dRefFrames}
\end{table} 
\medskip \\
\textit{Action and parallel transport. }
In analogy with the geometric interpretation given in the $2$ dimensional case, here the left action of $\cG$ on $\cG^*$ (resp. $\cG^*$ on $\cG$) transports the decoration of the edge and face (resp. link and wedge) from the target $c_t$ to the source $c_1$ of the link (resp. from the target $v_2$ to the source $v_1$ of the edge), and similarly the right action of $\cG$ on $\cG^*$ (resp. $\cG^*$ on $\cG$) transports the decoration of edge and face (resp. link and wedge) from $c_1$ to $c_t$ (resp. from $v_1$ to $v_2$).

Since we deal with parallel transport between vertices and nodes (sources and targets of edges and links), the action responsible for the parallel transport between nodes is that of the group $G_1 \subset \cG$ on $\cG^*$ (rather than the action of $\cG$ on $\cG^*$) and the action responsible for the parallel transport between vertices is that of the group $G_2 \subset \cG^*$ on $\cG$ (rather than the action of $\cG^*$ on $\cG$). 

Moreover, the left action of $G_1$ on $G_2^*$ (resp.\ $G_2$ on $G_1^*$) transports the root of the wedge (resp.\ face) from the target of the link $c_t$ to its source $c_1$ (resp.\ from the target of the edge $v_2$ to its source $v_1$) and the right action transports the root of the the wedge (resp.\ face) from the source $c_1$ to the target $c_t$ (resp.\ from the source $v_1$ to the target $v_2$). 

We summarize the geometric role of left and right actions in Table \ref{Table_3dActions}. Note that the geometric interpretation given in Table \ref{Table_3dActions} agrees with the convention in Table \ref{Table_3dRefFrames}. 
\begin{table}
\centering
    \begin{tabular}{ | r | | c | c || }
    \cline{2-3}
    \multicolumn{1}{c|}{}
    & 
    \textit{\textbf{$G_1$ on $\left\{\begin{array}{l}G_2^*\\ \cG^*
    \end{array}\right.$}}
    &
    \textit{\textbf{$G_2$ on $\left\{\begin{array}{l} G_1^*\\ \cG
    \end{array}\right.$}} \\
    \hline
    \textit{\textbf{Left action $\rhd$}}
    &
    $c_2 \to c_1$ & $v_2 \to v_1$ \\
    \hline
    \textit{\textbf{Right action $\lhd$}}
    &
    $c_1 \to c_2$ & $v_1 \to v_2$ \\    
    \hline
    \end{tabular}
    \caption{
    Geometrically, the left action of $G_1$ (resp. $G_2$) on the other groups, transports the decoration of the objects from the node $c_2$ to $c_1$ (resp. from the vertex $v_2$ to $v_1$); the right action of $G_1$ (resp. $G_2$) on the other groups, transports the decoration from $c_1$ to $c_2$ (resp. from $v_1$ to $v_2$).}
    \label{Table_3dActions}
\end{table}
\medskip \\
\textit{Gluing. }
According to the notation in Table \ref{Table_3dRefFrames}, decorations of links and edges carry a single index while decorations of faces and wedges carry two indices. 

To fuse phase spaces (Heisenberg doubles), we impose a geometric constraint, identifying some objects in the two phase spaces. In order to properly impose these conditions, the decorations of such objects need to be represented and rooted in the same reference frame, therefore we also demand that all the indices of such decorations match.

\medskip

When $\cB$ is decorating the pair ((edge, face); ((half-)link, wedge)), we will call it the   \textit{atomic phase space}. We will glue/fuse such atomic phase spaces to construct the phase space associated to the full cellular complex.  

\medskip

The atomic phase space is encoded in the octagonal ribbon constraint, as in Fig.\ \ref{octagon}. We note that there are actually many possible constraints, depending on the choice of representation of each group. 
\be\label{ribbon-octo}
    \ell h = \th \tell
    \,\, \Leftrightarrow \,\,
    \left\{\begin{array}{c}
        (\beta \,\lambda) \, (y \, u) = (\ty \, \tu) \, (\tbe \, \tla) , \\
        (\beta \,\lambda) \, (y \, u) = (\tu \, \toy) \, (\tbe \,\tla) , \\
        \vdots
    \end{array}
    \right.
\ee

As in the $2$ dimensional case the symplectic structure was given in terms of a symplectic form, here the symplectic structure for the atomic phase space is induced by the form
\begin{align}
    \Omega & = 
    \demi \big( \la \Dr \ell \wedge \Dr \th \ra + \la \Dl \tell \wedge \Dl h \ra \big) =
    \demi \big( \la \Dr (\beta \lambda) \wedge \Dr (\toy \tu) \ra + \la \Dl (\tla \tobe) \wedge \Dl (u y) \ra \big)
    \nonumber \\
    & =
    \demi \big( \la (\delta \beta + \Dr \lambda + [\beta , \Dr \lambda
    ]) \wedge (\delta \toy + \Dr \tu + [\toy , \Dr \tu]) \ra +
    \la (\Dl \tla + \delta \tobe - [\tobe , \Dl \tla]) \wedge (\Dr u + \delta y - [y , \Dl u]) \ra \big)
    \nonumber \\
    & =
    \demi \big( 
    \la (\delta \beta + [\beta , \Dr \lambda
    ]) \wedge \Dr \tu \ra +
    \la \Dr \lambda \wedge (\delta \toy + [\toy , \Dr \tu]) \ra +
    \la (\delta \tobe - [\tobe , \Dl \tla]) \wedge \Dr u \ra +
    \la \Dl \tla \wedge (\delta y - [y , \Dl u]) \ra
    \big) .
\end{align}
Let us now define the ways to glue/fuse the different geometric objects.

\subsubsection{Gluing rules}
\label{Sec_GluingRules}

In the previous subsection, we characterized  how  the Heisenberg double of a skeletal $2$-group, what we called the atomic phase space, decorates the single building block of our phase space, composed by edge, face, link and wedge, see Fig. \ref{double-wedge}. 

To build the full phase space for a polyhedron or a 3d cellular complex, we will glue these atomic phase spaces to build the phase space for the structure of interest. 
The gluing can be decomposed according to four types of gluing/fusion. We recall from the previous discussion that the gluing comes from the identification of the variables dual to the ones we want to glue. The gluing operations, achieved through symplectic reduction, are associative and commutative. 

\begin{itemize}

    \item \textit{Link gluing.}  Links are decorated by $G_1$ elements, so this gluing is performed by  identifying elements in $G_1^*$. We perform a symplectic reduction with respect to the $G_1$ symmetry.
    
    \item \textit{Face gluing.} Faces are decorated by $G_1^*$ elements, so this gluing is performed by  identifying elements in $G_1$. Note that this amounts to perform a vertical product from the 2-group picture. We perform a symplectic reduction with respect to the $G_1^*$ symmetry. 
    
    \item \textit{Wedge gluing.} Wedges are decorated by $G_2^*$ elements, so this gluing is performed by  identifying elements in $G_2$. Note that this amounts to perform a vertical product from the 2-group picture. We perform a symplectic reduction with respect to the $G_2^*$ symmetry. 
    
    \item \textit{Edge gluing}. Edges are decorated by $G_2$ elements, so this gluing is performed by  identifying elements in $G_2^*$. We perform a symplectic reduction with respect to the $G_2$ symmetry.
 
\end{itemize}

Upon performing the symplectic reduction, we can construct the new ribbon constraints to determine the symplectic structure in terms of the fused variables. We emphasize however that when we keep repeating the different gluings, the ribbon equations will be mostly useful for expressing all the variables in the right frame and root. This will be useful for example to define the closure constraints where we must ensure that the different variables are expressed in the right frame (and properly rooted).

\medskip

Let us illustrate explicitly these gluings.  
We consider the  atomic phase space $\cB_{c_1,e}$ associated to the edge $e=(v_1v_2)$, face $t_{e}$ and (half) link $l=(c_1c_t)$, wedge $ w_{l}$  (resp. $\cB_{c_2,e'}$ associated to the edge $e'=(v'_1v'_2)$, face $ t'_{e'}$ and half link $l'=(c_2c_{t'})$, wedge $ w_{l'}$), as in Fig. \ref{double-wedge}. 

For the sake of clarity, it is useful to introduce the following notation. Consider the cross Lie bracket \eqref{CrossActions_Semidual}. Following the Def. \ref{Def_ClassicalDouble} of a classical double, the Lie bracket between elements of $\g_1$ and $\g_1^*$ (resp. $\g_2$ and $\g_2^*$) are constructed as the mutual co-adjoint actions between $\g_2 \ltimes \g_1^*$ and $\g_2^* \rtimes \g_1$; however, these brackets are not stable in $\g_1$ and $\g_1^*$ (resp. $\g_2$ and $\g_2^*$). 
At the group level, the exponentiation of such Lie brackets leads to the group conjugations
\be\label{conjug}
    \begin{aligned}
        u \beta u\mone
        = \beta' y'
        & \,\,\, \Leftrightarrow \,\,
        \begin{cases}
            & \beta' =  (u \beta u\mone)|_{G_1^*} , \\
            & y' = (u \beta u\mone)|_{G_2^*} ,
        \end{cases}
        \\
        \lambda y \lambda\mone 
        = y'' \beta'' 
        & \,\,\, \Leftrightarrow \,\,
        \begin{cases}
            & \beta'' = (\lambda y \lambda\mone)|_{G_1^*} , \\
            & y'' = (\lambda y \lambda\mone)|_{G_2^*} .
        \end{cases}
    \end{aligned}
\ee
Therefore, in order to express the conjugations $u \beta u\mone$ and $\lambda y \lambda\mone$, we will use their projections into the groups $G_1^*$ and $G_2^*$.

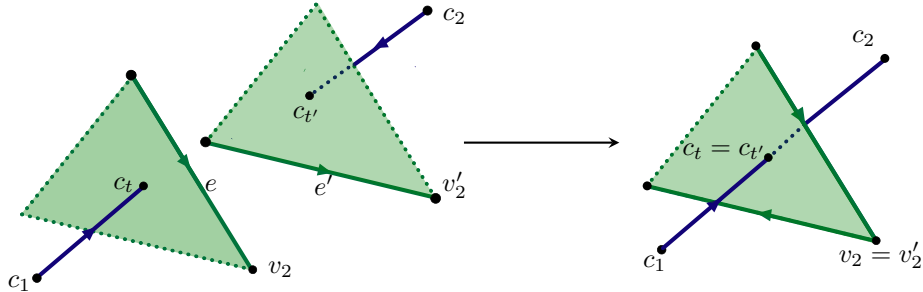
\begin{figure}[H]
\begin{center}
\begin{tikzpicture}
    \node[anchor=south west,inner sep=0] (image) at (0,0) {\includegraphics[width=0.7\textwidth]{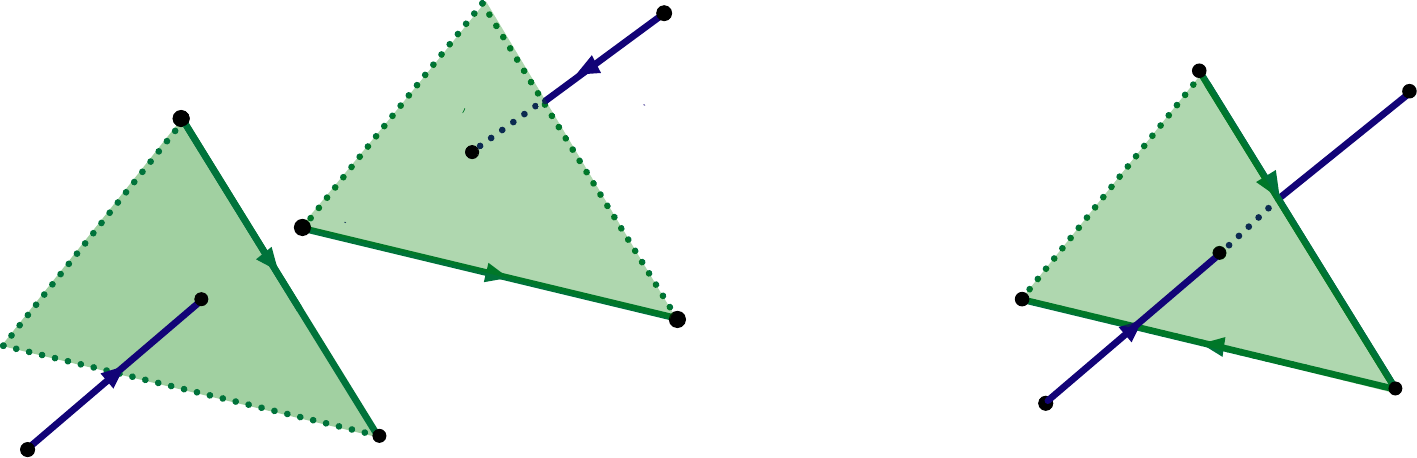}};
    \begin{scope}[x={(image.south east)},y={(image.north west)}]
       \node[] at (0,0.00)     (c1)  {$c_1$};
       \node[] at (0.12, 0.35) (ct)  {$c_t$};
       \node[] at (0.3,0.05)   (v2)  {$v_2$};
       \node[] at (0.33,0.6)  (ct') {$c_{t'}$};
       \node[] at (0.5, 0.95) (c2)  {$c_2$};
       \node[] at (0.5, 0.35)  (v2') {$v_2'$};
       \node[] at (0.73,0.07)  (c1n) {$c_1$};
       \node[] at (0.97,0.87)  (c2n) {$c_2$};
       \tikzstyle{arrow} = [thick,->,>=stealth]
       \node[] at (0.5,0.5) (glue1) {};
       \node[] at (0.7,0.5) (glue2) {};
       \draw[arrow] (glue1) edge (glue2);
       \node[] at (0.99, 0.1) (v)  {$v_2 = v_2'$};
       \node[] at (0.35, 0.35) (e')  {$e'$};
       \node[] at (0.22, 0.35) (e)  {$e$};
       \node[] at (0.81,0.47) (ctct') {$c_t = c_{t'}$};
    \end{scope}
    
\end{tikzpicture}
\caption{We fuse two links by identifying the dual faces $t_e$ and $t_{e'}$. 
We require that the points $c_t$ and $c_{t'}$ in the respective face to match. As we discussed before, face decorations are rooted at a point. Matching the face decorations imposes that they have the same root. Picking $v_2$ and $v_2'$ as respectively root of $t_e$ and $t_{e'}$, we therefore must have $v_2 = v_2' $.   
   }
    \label{link-fused}
\end{center}
\end{figure}

\paragraph{Link gluing.}
The gluing of two links, as illustrated in Fig. \ref{link-fused}, is done by identifying the face data dual to each link.
\medskip\\
The geometric identification consists in setting\footnote{Faces are identified with opposite outgoing directions.}
\be
     t_e = - t_{e'} .
\ee
We enforce this geometric constraint as a condition on $G_1^*$ elements.
Let us represent the decorations of the faces $t_e$ and $t_{e'}$ at the nodes $c_t$ and $c_{t'}$. We root them respectively at the vertices $v_2$ and $v_2'$. 
\be
    \begin{aligned}
     \cB_{l,e}:
     \quad   &
        \cG^* = G_2 \ltimes G_1^* \ni 
        \tell^{c_t} \equiv \tell_1 = \tla_1 \, \tobe_1 \equiv \tla^{c_t} \, \tobe^{c_t}_{v_2} , \\
    \cB_{l',e'}: 
    \quad    & 
        \cG^* = G_2 \ltimes G_1^* \ni 
        \tell^{c_{t'}} \equiv \tell_2 = \tla_2 \, \tobe_2 \equiv \tla^{c_{t'}} \, \tobe^{c_{t'}}_{v_2'} , 
    \end{aligned}
    \quad \,
    \begin{aligned}
        &
        \cG = G_1 \ltimes G_2^* \ni 
        h^{v_2} \equiv h_1 = u_1 \, y_1 \equiv u^{v_2} \, y^{v_2}_{c_t} , \\
        &
        \cG = G_1 \ltimes G_2^* \ni 
        h^{v_2'} \equiv h_2 = u_2 \, y_2 \equiv u^{v_2'}\, y^{v_2'}_{c_{t'}}.
    \end{aligned}
\ee
In order to express the face decorations in the same reference frame, their indices have to match, hence we further demand that
\be
    c_t = c_{t'} \,\,,\quad v_2 = v_2' .
\ee
The condition $c_t = c_{t'}$ provides the gluing the half links.  
The proper geometric condition for the face decorations is encoded in the identification 
\be
    \tobe^{c_t}_{v_2} = \big(\tobe^{c_{t'}}_{v_2'}\big)^{-1} 
    \,\, \Leftrightarrow \,\,
    \tobe_1 =  \tobe_2 \mone.
    \label{Link_GluingConstraint}
\ee
Using the above constraint and the Heisenberg ribbon equation, we obtain the extended ribbon equation
\begin{align}
    &
    i=1,2, \quad
    \ell_i h_i = \th_i \tell_i = \th_i (\tla_i \tobe_i) 
    \quad \Leftrightarrow \quad 
    \tobe_i = \th_i\mone \tla_i\mone \ell_i h_i
    \quad \stackrel{\tobe_1 =  \tobe_2 \mone}{\Rightarrow} \quad
    \th_1\mone\tla_1\mone \ell_1 h_1 =  h_2 \mone \ell_2\mone \tla_2 \th_2
    \nonumber \\
    &
    \nonumber \\
    &
    \begin{aligned}
        \quad \Leftrightarrow \qquad
        \tla_1\mone (\tu_1 \ty_1)\mone (\lambda_1 \obe_1) (\oy_1 u_1)
        & =
        (\oy_2 u_2)\mone (\lambda_2 \obe_2)\mone (\tu_2 \ty_2) \tla_2
        \\
        (\lambda_1 \obe_1) (\oy_1 u_1) \tla_2\mone (\tu_2 \ty_2)\mone
        & =
        (\tu_1 \ty_1) \tla_1 (\oy_2 u_2)\mone (\lambda_2 \obe_2)\mone
        \\
        \lambda_1 \, \obe_1 \, \ty_1 (u_1 \rhd \tla_2\mone) (u_1 \lhd \tla_2\mone) \ty_2\mone \, \tu_2\mone
        & =
        \tu_1 \, \oy_1 (\tla_1 \rhd u_2\mone) (\tla_1 \lhd u_2\mone) \oy_2\mone \, \obe_2\mone \, \lambda_2\mone
        \\
        \lambda_1 \, \obe_1 (u_1 \rhd \tla_2\mone) y' \, \beta' ((u_1 \lhd \tla_2\mone) \rhd \ty_2\mone) (u_1 \lhd \tla_2\mone) \tu_2\mone
        & =
        \tu_1 (\tla_1 \rhd u_2\mone) (\ty_1 \lhd (\tla_1 \rhd u_2\mone)) y'' \, \beta'' (\tla_1 \lhd u_2\mone) \obe_2\mone \, \lambda_2\mone ,
    \end{aligned}
    \nonumber \\
    & \nonumber \\
    &
    \begin{aligned}
        \qquad
        &
        \big(\lambda_1 (u_1 \rhd \tla_2\mone)\big)
        \big(\obe_1 \lhd (u_1 \rhd \tla_2\mone) \beta' \big)
        \big(y' ((u_1 \lhd \tla_2\mone) \rhd \ty_2\mone)\big) 
        \big((u_1 \lhd \tla_2\mone) \tu_2\mone\big)
        \\
        & \qquad \qquad \qquad \qquad  =
        \big(\tu_1 (\tla_1 \rhd u_2\mone)\big)
        \big((\ty_1 \lhd (\tla_1 \rhd u_2\mone)) y''\big)
        \big(\beta'' (\tla_1 \lhd u_2\mone) \rhd \obe_2\mone\big)
        \big((\tla_1 \lhd u_2\mone) \lambda_2\mone\big) ,
    \end{aligned}
    \label{ExtendedRibEq_LinkGluing}
\end{align}
with $(u_1 \rhd \tla_2\mone)\mone \, \oy_1 \, (u_1 \rhd \tla_2\mone) = y' \beta'$ and $(\tla_1 \lhd u_2\mone) \, \oy_2 \, (\tla_1 \lhd u_2\mone)\mone = y'' \beta''$. The link fusion is therefore obtained through the symplectic reduction
\be
    (\cB_{c_1,e} \times \cB_{c_2,e'}) /\!/ G_1 
    \label{PhaseSpace_LinkGluing}
\ee
with momentum map $\tobe^{c_t}_{v_2} \, \tobe^{c_{t'}}_{v_2'} = \tobe^{c_t}_{v_2} + \tobe^{c_{t'}}_{v_2'} \in G_1^*$ and fused link variable $\tu^{v_1}_l = \tu^{v_1} (\tla^{c_t} \rhd (u^{v_2})\mone) \in G_1$.

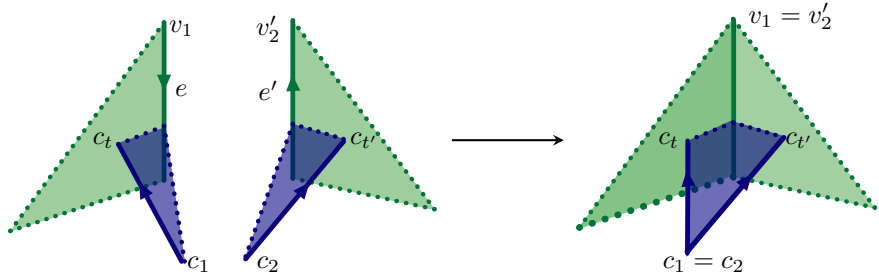
\begin{figure}[H]
    \begin{center}   
    \begin{tikzpicture}
    \node[anchor=south west,inner sep=0] (image) at (0,0)
    {\includegraphics[width=0.7\textwidth]{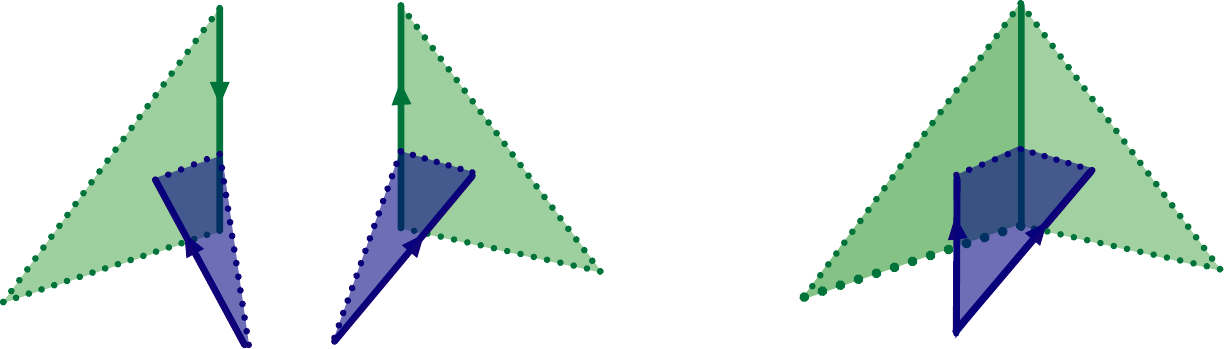}};
    \begin{scope}[x={(image.south east)},y={(image.north west)}]
       \node[] at (0.22,0) (c1) {$c_1$};
       \node[] at (0.3,0) (c2) {$c_2$};
       \node[] at (0.2,0.95) (v1) {$v_1$};
       \node[] at (0.3, 0.95) (V2') {$v_2'$};
       \node[] at (0.2, 0.7) (e) {$e$};
       \node[] at (0.3, 0.7) (e') {$e'$};
       \node[] at (0.11, 0.5)(ct) {$c_t$};
       \node[] at (0.41, 0.5) (ct'){$c_{t'}$};
       \node[] at (0.8,0) (c2c1) {$c_1=c_2$};
       \node[] at (0.5,0.5) (glue1) {};
       \node[] at (0.65,0.5) (glue2) {};
       \tikzstyle{arrow} = [thick,->,>=stealth]
       \draw[arrow] (glue1) edge (glue2);
       \node[] at (0.76, 0.5) {$c_t$};
       \node[] at (0.91,0.5) {$c_{t'}$};
       \node[] at (0.9,1) {$v_1= v_2'$};
    \end{scope}
    \end{tikzpicture}
    \caption{We fuse two wedges by identifying the edges $e$ and ${e'}$, so that they have opposite orientation.  To fuse the wedges, we require that the two wedges are rooted at the same point, so we identify the nodes $c_1$ and $c_2$. 
    }
    \label{wedge-fused}
    \end{center}
\end{figure}

\paragraph{Wedge gluing Part 1.}
The gluing of two wedges, as illustrated in Fig. \ref{wedge-fused}, is done by identifying the edge data, dual to each wedge.
\medskip\\
The geometric identification consists in setting\footnote{In this example the edges are identified with opposite orientations.}
\be 
    e = - e' 
    \,\,\,\leftrightarrow\,\,\,
    v_1 = v_2' \,,\,\, v_2 = v_1'.
\ee
We enforce the geometric constraint as a condition on $G_2$ elements. 
Let us represent the decorations of the edges $e$ and $e'$ at the nodes $c_1$ and $c_2$.
\be
    \begin{aligned}
        \cB_{l , e} :
        \quad &
        \cG^* = G_2 \ltimes G_1^* \ni 
        \ell^{c_1} \equiv \ell_1 = \lambda_1 \, \obe_1 \equiv \lambda^{c_1} \, \obe^{c_1}_{v_2} , \\
        \cB_{l' , e'} :
        \quad &
        \cG^* = G_1^* \rtimes G_2 \ni 
        \ell^{c_2} \equiv \ell_2 = \beta_2 \, \lambda_2 \equiv \beta^{c_2}_{v_1'} \,\lambda^{c_2} ,
    \end{aligned}
    \quad \,
    \begin{aligned}
        &
        \cG = G_2^* \rtimes G_1 \ni 
        h^{v_2} \equiv h_1 = \oy_1 \, u_1 \equiv \oy^{v_2}_{c_1} \, u^{v_2} , \\
        &
        \cG = G_2^* \rtimes G_1 \ni 
        \th^{v_1'} \equiv \th_2 = \toy_2 \, \tu_2 \equiv \toy^{v_1'}_{c_2} \, \tu^{v_1'} .
    \end{aligned}
\ee
In order to express the edge decorations in the same reference frame, their indices have to match, hence we further demand that
\be
    c_1 = c_2 .
\ee
The proper geometric condition is encoded in the identification
\be
    \lambda^{c_1} = (\lambda^{c_2})\mone 
    \,\,\Leftrightarrow\,\,
    \lambda_1 = \lambda_2\mone .    
    \label{Wedge_GluingConstraint_Tetra}
\ee
Using the above constraint and the Heisenberg ribbon equation, we obtain the extended ribbon equation
\begin{align}
    &
    \begin{cases}
        &
        \th_1 \tell_1 = \ell_1 h_1 = (\lambda_1 \obe_1) h_1 
        \quad \Leftrightarrow \quad 
        \lambda_1 = \th_1 \tell_1 h_1\mone \obe_1\mone \\
        &
        \th_2 \tell_2 = \ell_2 h_2 = (\beta_2 \lambda_2) h_2 
        \quad \Leftrightarrow \quad 
        \lambda_2 = \beta_2\mone \th_2 \tell_2 h_2\mone \\
    \end{cases}
    \quad \stackrel{\lambda_1 = \lambda_2\mone}{\Rightarrow} \quad
    \th_1 \tell_1 h_1\mone \obe_1\mone = h_2\mone \tell_2\mone \th_2\mone \beta_2
    \nonumber \\
    &
    \nonumber \\
    & \qquad 
    \begin{aligned}
        \quad \Leftrightarrow \qquad
        (\tu_1 \ty_1) (\tbe_1 \tla_1) (\oy_1 u_1)\mone \obe_1\mone
        & =
        \obe_2 (\oy_2 u_2) (\tbe_2 \tla_2)\mone (\tu_2 \ty_2)\mone
        \\
        (\oy_1 u_1)\mone \obe_1\mone (\tu_2 \ty_2) (\tbe_2 \tla_2)
        & =
        (\tbe_1 \tla_1)\mone (\tu_1 \ty_1)\mone \obe_2 (\oy_2 u_2)
        \\
        u_1\mone \, \oy_1\mone \, \tu_2 \, y' \, \beta' \, \ty_2 \, \tbe_2 \, \tla_2
        & =
        \tla_1\mone \, \tbe_1\mone \, \ty_1\mone \, y'' \, \beta'' \, \tu_1\mone \, \oy_2 \, u_2
        \\
        (u_1\mone \tu_2)
        \big((\oy_1\mone \lhd \tu_2) y' \ty_2\big) 
        (\beta' \tbe_2)
        \tla_2
        & =
        \tla_1\mone
        (\tbe_1\mone \beta'') 
        \big(\ty_1\mone y'' (\tu_1\mone \rhd \oy_2)\big)
        (\tu_1\mone u_2) ,
    \end{aligned}
    \label{ExtendedRibEq_Wedgeluing_Tetra}
\end{align}
with $\tu_2\mone \, \obe_1 \, \tu_2 = (y' \beta')\mone$ and $\tu_1\mone \, \obe_2 \tu_1 = y'' \beta''$.
The wedge gluing is therefore obtained through the symplectic reduction 
\be
    (\cB_{c_1,e} \times \cB_{c_2,e'})/\!/ G_2^*
    \label{PhaseSpace_WedgeGluing_Tetra}
\ee
with momentum map $\lambda^{c_1} \lambda^{c_2} \in G_2$ and fused wedge variable $(\ty_{ww'})^{v_2}_{c_t'} = - \oy^{v_2}_{c_1} \lhd \tu^{v_1'} +\ty^{v_1'}_{c_{t'}} - \big( (\tu^{v_1'})\mone \, \obe^{c_1}_{v_2} \, \tu^{v_1'} \big)\big|_{G_2^*} \in G_2^*$.

\paragraph{Wedge gluing Part 2.}
The gluing of wedges 
is done by identifying the edge data dual to each wedge. It is mostly relevant for building the holonomy around an edge.
\medskip \\
The geometric identification consists in setting\footnote{In this example the edges are identified with the same orientation.}
\be 
    e = e' 
    \,\,\,\leftrightarrow\,\,\,
    v_1 = v_1' \,,\,\, v_2 = v_2'.
\ee
We enforce the geometric constraint as a condition on $G_2$ elements. 
Let us represent the decorations of the edges $e$ and $e'$ at the nodes $c_t$ and $c_2$.
\be
    \begin{aligned}
        \cB_{l , e} :
        \quad &
        \cG^* = G_2 \ltimes G_1^* \ni 
        \tell^{c_t} \equiv \tell_1 = \tbe_1 \, \tla_1 \equiv \tbe^{c_t}_{v_1} \, \tla^{c_t} , \\
        \cB_{l' , e'} :
        \quad &
        \cG^* = G_1^* \rtimes G_2 \ni 
        \ell^{c_2} \equiv \ell_2 = \beta_2 \, \lambda_2 \equiv \beta^{c_2}_{v_1'} \,\lambda^{c_2} ,
    \end{aligned}
    \quad \,
    \begin{aligned}
        &
        \cG = G_2^* \rtimes G_1 \ni 
        \th^{v_1} \equiv \th_1 = \tu_1 \, \ty_1 \equiv \tu^{v_1} \, \ty^{v_1}_{c_t} , \\
        &
        \cG = G_2^* \rtimes G_1 \ni 
        \th^{v_1'} \equiv \th_2 = \toy_2 \, \tu_2 \equiv \toy^{v_1'}_{c_2} \, \tu^{v_1'} .
    \end{aligned}
\ee
So that the variables we wish to identify are in the same representation, we set $c_2=c_t$. The proper geometric condition is encoded in the identification
\be
    \tla^{c_t} = \lambda^{c_2}
    \,\,\Leftrightarrow\,\,
    \tla_1 = \lambda_2 .    
\ee
Using the above constraint and the Heisenberg ribbon equation, we obtain the extended ribbon equation
\begin{align}
    &
    \begin{cases}
        &
        \ell_1 h_1 = \th_1 \tell_1 = \th_1 (\tbe_1 \tla_1) 
        \quad \Leftrightarrow \quad 
        \tla_1 = \tbe_1\mone \th_1\mone \ell_1 h_1 \\
        &
        \th_2 \tell_2 = \ell_2 h_2 = (\lambda_2 \obe_2) h_2 
        \quad \Leftrightarrow \quad 
        \lambda_2 = \th_2 \tell_2 h_2\mone \obe_2\mone \\
    \end{cases}
    \quad \stackrel{\tla_1 = \lambda_2}{\Rightarrow} \quad
    \tbe_1\mone \th_1\mone \ell_1 h_1 = \th_2 \tell_2 h_2\mone \obe_2
    \nonumber \\
    &
    \nonumber \\
    & \qquad 
    \begin{aligned}
        \quad \Leftrightarrow \qquad
        \tbe_1\mone (\tu_1 \ty_1)\mone (\lambda_1 \obe_1) (\oy_1 u_1)
        & =
        (\tu_2 \ty_2) (\tbe_2 \tla_2) (\oy_2 u_2)\mone \obe_2
        \\
        (\lambda_1 \obe_1) (\oy_1 u_1) \obe_2 (\oy_2 u_2)
        & =
        (\tu_1 \ty_1) \tbe_1 (\tu_2 \ty_2) (\tbe_2 \tla_2)
        \\
        \lambda_1 \, \obe_1 \, \oy_1 \, \beta' \, y' \, u_1 \, \oy_2 \, u_2
        & =
        \tu_1 \, \ty_1 \, \tu_2 \, \ty' \, \tbe' \, \ty_2 \, \tbe_2 \, \tla_2
        \\
        \lambda_1 
        (\obe_1 \beta')
        \big(\oy_1 y' (u_1 \rhd \oy_2)\big)
        (u_1 u_2)
        & =
        (\tu_1 \tu_2)
        \big((\ty_1 \lhd \tu_2) \ty' \ty_2\big) 
        (\tbe' \tbe_2) 
        \tla_2 ,
    \end{aligned}
    \label{ExtendedRibEq_Wedgeluing_Loop}
\end{align}
with $u_1 \, \obe_2 \, u_1\mone = \beta' y'$ and $\tu_2\mone \, \tbe_1 \, \tu_2 = \ty' \tbe'$.
The wedge gluing is therefore obtained through the symplectic reduction
\be
    (\cB_{c_1,e} \times \cB_{c_2,e'}) /\!/ G_2^*
\ee
with momentum map $\tla_1 \lambda_2\mone \in G_2$ and fused wedge variable $(\oy_{ww'})^{v_1}_{c_1} = \oy^{v_1}_{c_1} + u^{v_2}\rhd \oy^{v_1'}_{c_2} + \big(u^{v_2} \, \obe^{c_2}_{v_2'} \, (u^{v_2})\mone)|_{G_2^*} \in G_2^*$.

\begin{figure}[H]
    \begin{center}
    \begin{tikzpicture}
    \node[anchor=south west,inner sep=0] (image) at (0,0)
    {\includegraphics[width=0.6\textwidth]{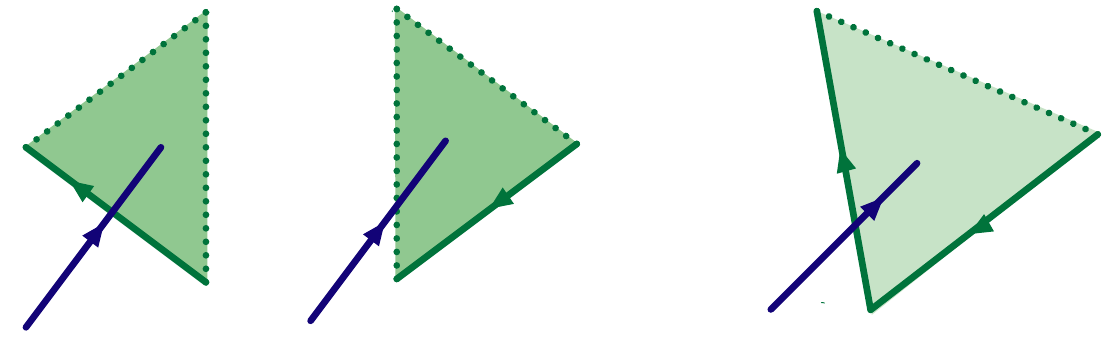}};
    \begin{scope}[x={(image.south east)},y={(image.north west)}]
    \node[] at (0,0) {$c_2$};
    \node[] at (0.27,0) {$c_1$};
    \node[] at (0.12,0.55) {$c_{t'}$};
    \node[] at (0.42,0.55) {$c_t$};
    \node[] at (0.65,0) {$c_1=c_2$};
    \node[] at (0.85, 0.55) {$c_t = c_{t'}$};
    \node[] at (0, 0.50) {$v_2'$};
    \node[] at (0.2,0.1) {$v_1'$};
    \node[] at (0.38,0.1) {$v_2$};
    \node[] at (0.55,0.55) {$v_1$};
    \node[] at (0.57,0.45) (glue1){};
    \node[] at (0.7,0.45) (glue2){};
    \tikzstyle{arrow} = [thick,->,>=stealth]
    \draw[arrow] (glue1) edge (glue2);
    \node[] at (0.88, 0.07) {$v_1' = v_2$};
    \node[] at (1, 0.54) {$v_1$};
    \node[] at (0.71,0.9) {$v_2'$};
    \end{scope}
    \end{tikzpicture}
    \caption{We fuse two faces by identifying the links $l$ and ${l'}$.  To merge the faces, we also require them to have the same root. We choose to root $t_e$ and $t_{e'}$ respectively at $v_2$ and $v_1'$. We  then identify $v_2 = v_1'$.
    }
    \label{face-fused}
    \end{center}
\end{figure}
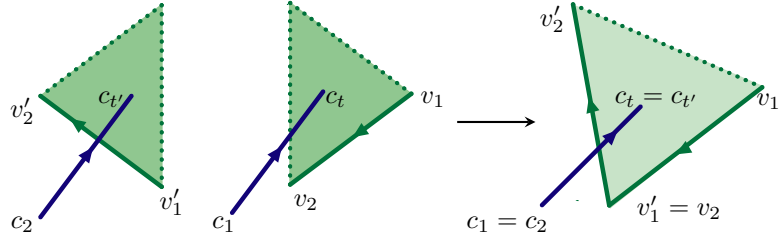

\paragraph{Face gluing.}
The gluing of two faces, as illustrated in Fig. \ref{face-fused}, is done by identifying the link data dual to each face.
\medskip\\
The geometric identification consists in setting 
\be
    l = l' 
    \,\,\, \leftrightarrow \,\,\,
    c_1 = c_2 \,,\,\, c_t = c_{t'} .
\ee
We enforce the geometric constraint as a condition on $G_1$ elements. 
Let us represent the decorations of the links $l$ and $l'$ at the vertices $v_2$ and $v_1'$.
\be
    \begin{aligned}
        \cB_{l , e} :
        \quad &
        \cG^* = G_2 \ltimes G_1^*\ni 
        \tell^{c_t} \equiv \tell_1 = \tla_1 \, \tobe_1 \equiv \tla^{c_t} \, \tobe^{c_t}_{v_2} , \\
        \cB_{,' , e'} :
        \quad &
        \cG^* = G_1^*\rtimes G_2\ni 
        \tell^{c_{t'}} \equiv \tell_2 = \tbe_2 \, \tla_2 \equiv \tbe^{c_{t'}}_{v_1'} \, \tla^{c_{t'}} ,
    \end{aligned} 
    \quad \,
    \begin{aligned}
        &
        \cG = G_1 \ltimes G_2^* \ni 
        h^{v_2} \equiv h_1 = u_1 \, y_1 \equiv u^{v_2} \, y^{v_2}_{c_t} \, \\
        &
        \cG = G_1 \ltimes G_2^* \ni 
        \th^{v_1'} \equiv \th_2 = \tu_2 \, \ty_2 \equiv \tu^{v_1'} \, \ty^{v_1'}_{c_{t'}} .
    \end{aligned}
\ee
In order to express the link decorations in the same reference frame, their indices have to match, hence we further demand that
\be
    v_2 = v_1' .
\ee
The proper geometric condition is encoded in the identification
\be
    u^{v_2} = \tu^{v_1'}
    \,\,\Leftrightarrow\,\,
    u_1 = \tu_2 .
    \label{Face_GluingConstraint}
\ee
Using the above constraint and the Heisenberg ribbon equation, we obtain the extended ribbon equation
\begin{align}
    &
    \begin{cases}
        &
        \th_1 \tell_1 = \ell_1 h_1 = \ell_1 (u_1 y_1) 
        \quad \Leftrightarrow \quad 
        u_1 = \ell_1\mone \th_1 \tell_1 y_1\mone \\
        &
        \ell_2 h_2 = \th_2 \tell_2 = (\tu_2 \ty_2) \tell_2 
        \quad \Leftrightarrow \quad 
        \tu_2 = \ell_2 h_2 \tell_2\mone \ty_2\mone \\
    \end{cases}
    \quad \stackrel{u_1 = \tu_2}{\Rightarrow} \quad
    \ell_1\mone \th_1 \tell_1 y_1\mone = \ell_2 h_2 \tell_2\mone \ty_2\mone
    \nonumber \\
    &
    \nonumber \\
    & \qquad
    \begin{aligned}
        \quad \Leftrightarrow \qquad
        \oy_1\mone (\lambda_1 \obe_1)\mone (\tu_1 \ty_1) (\tbe_1 \tla_1) 
        & =
        (\lambda_2 \obe_2) (\oy_2 u_2) (\tbe_2 \tla_2)\mone \ty_2\mone
        \\
        (\tu_1 \ty_1) (\tbe_1 \tla_1) \ty_2 (\tbe_2 \tla_2)
        & =
        (\lambda_1 \obe_1) \oy_1 (\lambda_2 \obe_2) (\oy_2 u_2) 
        \\
        \tu_1 \, \ty_1 \, \tbe_1 \, \ty' \, \tbe' \, \tla_1 \, \tbe_2 \, \tla_2
        & =
        \lambda_1 \, \obe_1 \, \lambda_2 \, \beta' \, y' \, \obe_2 \, \oy_2 \, u_2 
        \\
        \tu_1 
        (\ty_1 \ty') 
        \big(\tbe_1 \tbe' (\tla_1 \rhd \tbe_2)\big) 
        (\tla_1 \tla_2)
        & =
        (\lambda_1 \lambda_2)
        \big((\obe_1 \lhd \lambda_2) \beta' \obe_2\big) 
        (y' \oy_2)
        u_2 ,
        \end{aligned}
    \label{ExtendedRibEq_FaceGluing}
\end{align}
with $\tla_1 \, \ty_2 \tla_1\mone = \ty' \tbe'$ and $\lambda_2\mone \, \oy_1 \, \lambda_2 = \beta' y'$.
The face fusion is therefore obtained through the symplectic reduction 
\be
    (\cB_{c_1,e} \times \cB_{c_2,e'})/\!/ G_1^*
    \label{PhaseSpace_FaceeGluing}
\ee
with momentum map $u^{v_2} (\tu^{v_1'})\mone \in G_1$ and fused face variable $(\tbe_{tt'})^{c_t}_{v_1} = \tbe^{c_t}_{v_1} + \lambda^{c_t} \rhd \tbe^{c_{t'}}_{v_1'} + \big( \tla^{c_1} \, \ty^{v_1'}_{c_{t'}} \, (\tla^{c_1})\mone \big) \big|_{G_1^*} \in G_1^*$.

\begin{figure}[H]
    \begin{center}
    \begin{tikzpicture}
    \node[anchor=south west,inner sep=0] (image) at (0,0)
    {\includegraphics[width=0.6\textwidth]{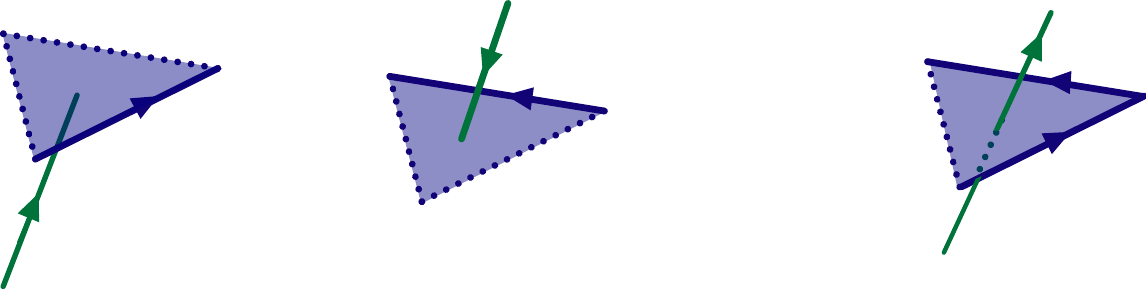}};
    \begin{scope}[x={(image.south east)},y={(image.north west)}]
    \node at (0.03,0) {$v_1$};
    \node at (0,0.5) {$c_1$};
    \node at (0.2, 0.7) {$c_t$};
    \node at (0.08,0.62) {$v_2$};
    \node at (0.3, 0.7) {$c_{t'}$};
    \node at (0.55,0.7) {$c_{2}$};
    \node at (0.41,1) {$v_1'$};
    \node at (0.38,0.6) {$v_2'$};
    \node at (0.8,0.04) {$v_1$};
    \node at (0.9,1) {$v_1'$};
    \node at (0.8, 0.4) {$c_1$};
    \node at (0.78, 0.73) {$c_{t'}$};
    \node at (1.05,0.6) {$c_2=c_t$};
    \node[] at (0.57,0.55) (glue1){};
    \node[] at (0.7,0.55) (glue2){};
    \tikzstyle{arrow} = [thick,->,>=stealth]
    \draw[arrow] (glue1) edge (glue2);
    \end{scope}
    \end{tikzpicture}
    \caption{We fuse two edges by identifying the wedges $w_l$ and $w_{l'}$. 
    Matching the wedge decorations imposes that they have the same root. Picking $c_t$ and $c_2$ as respectively root of $w_l$ and $w_{l'}$, we therefore must have $c_2 = c_t$. To fuse the edges, we must have that the target of the first edge, $e$ is the source of the second edge $e'$. Hence we must have $v_2=v_2'$  
    }
    \label{edge-fused}
    \end{center}
\end{figure}
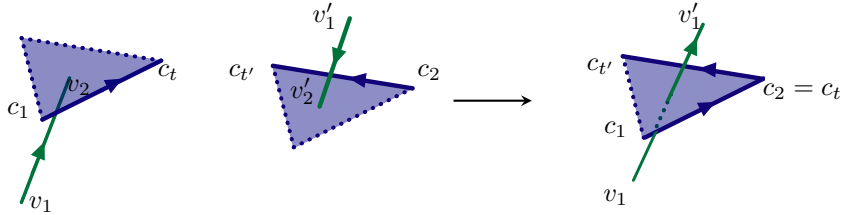

\paragraph{Edge gluing.}
The gluing of two edges, as illustrated in Fig. \ref{edge-fused}, is done by identifying the wedge data, dual to each edge.  
\medskip\\
The geometric identification consists in setting 
\be 
    w_l = w_{l'} .
\ee
We enforce this geometric constraint as a condition on $G_2^*$ element. Let us represent the decorations of the wedges $w_l$ and $w_{l'}$ at the vertices $v_2$ and $v_2'$. We root them respectively at the nodes $c_t$ and $c_2$.
\be
    \begin{aligned}
        \cB_{l , e} :
        \quad &
        \cG^* = G_2 \ltimes G_1^* \ni 
        \tell^{c_t} \equiv \tell_1 = \tla_1 \, \tobe_1 \equiv \tla^{c_t} \, \tobe^{c_t}_{v_2} , \\
        \cB_{l' , e'} :
        \quad & 
        \cG^*= G_2 \ltimes G_1^* \ni 
        \ell^{c_2} \equiv \ell_2 = \lambda_2 \, \obe_2 \equiv \lambda^{c_2} \, \obe^{c_2}_{v_2'} , \end{aligned}
        \quad \,
    \begin{aligned}
        &
        \cG = G_1 \ltimes G_2^* \ni 
        h^{v_2} \equiv h_1 = u_1 \, y_1 \equiv u^{v_2} \, y^{v_2}_{c_t} , \\
        &
        \cG = G_2^* \rtimes G_1 \ni 
        h^{v_2'} \equiv h_2 = \oy_2 \, u_2 \equiv \oy^{v_2'}_{c_2} \, u^{v_2'} .
    \end{aligned}
\ee
In order to express the wedge decorations in the same reference frame, their indices have to match, hence we further demand that
\be
    c_t = c_2 \,\,,\quad v_2 = v_2' .
\ee
The condition $v_2 = v_2'$ provides the gluing of the edges.
The proper geometric condition is encoded in the identification 
\be
    y^{v_2}_{c_t} = \oy^{v_2'}_{c_2}
    \,\,\Leftrightarrow\,\,
    y_1 = \oy_2 .    
    \label{Edge_GluingConstraint}
\ee
Using the above constraint and the Heisenberg ribbon equations, we obtain the extended ribbon equation
\begin{align}
    &
    \begin{cases}
        &
        \th_1 \tell_1 = \ell_1 h_1 = \ell_1 (u_1 y_1) 
        \quad \Leftrightarrow \quad 
        y_1 = u_1\mone \ell_1\mone \th_1 \tell_1 \\
        &
        \th_2 \tell_2 = \ell_2 h_2 = \ell_2 (\oy_2 u_2)
        \quad \Leftrightarrow \quad 
        \oy_2 = \ell_2\mone \th_2 \tell_2 u_2\mone \\
    \end{cases}
    \quad \stackrel{y_1 = \oy_2}{\Rightarrow} \quad
    u_1\mone \ell_1\mone \th_1 \tell_1 = \ell_2\mone \th_2 \tell_2 u_2\mone
    \nonumber \\
    &
    \nonumber \\
    & \qquad 
    \begin{aligned}
        \quad \Leftrightarrow \qquad
        u_1\mone (\lambda_1 \obe_1)\mone (\tu_1 \ty_1) (\tbe_1 \tla_1)
        & =
        (\beta_2 \lambda_2)\mone (\toy_2 \tu_2) (\tla_2 \tobe_2) u_2\mone
        \\
        (\tu_1 \ty_1) (\tbe_1 \tla_1) u_2 (\tla_2 \tobe_2)\mone
        & =
        (\lambda_1 \obe_1) u_1 (\beta_2 \lambda_2)\mone (\toy_2 \tu_2)
        \\
        \tu_1 \, \ty_1 \, \tbe_1 (\tla_1 \rhd u_2) (\tla_1 \lhd u_2) \tobe_2\mone \, \tla_2\mone
        & =
        \lambda_1 \, \obe_1 (u_1 \rhd \lambda_2\mone) (u_1 \lhd \lambda_2\mone) \beta_2\mone \, \toy_2 \, \tu_2
        \\
        \tu_1 \, \ty_1 (\tla_1 \rhd u_2) \ty' \, \tbe' ((\tla_1 \lhd u_2) \rhd \tobe_2\mone) (\tla_1 \lhd u_2) \tla_2\mone
        & =
        \lambda_1 (u_1 \rhd \lambda_2\mone) (\obe_1 \lhd (u_1 \rhd \lambda_2\mone)) 
        \beta' \, y' (u_1 \lhd \lambda_2\mone) \toy_2 \, \tu_2 
        \nonumber
    \end{aligned} 
    \nonumber \\
    & \nonumber \\
    &
    \begin{aligned}
        & \qquad \qquad
        \big( \lambda_1 (u_1 \rhd \lambda_2\mone) \big)
        \big( (\obe_1 \lhd (u_1 \rhd \lambda_2\mone) \beta' \big)
        \big( y' ((u_1 \lhd \lambda_2\mone) \rhd \toy_2) \big)
        \big( (u_1 \lhd \lambda_2\mone) \tu_2 \big)
        \\
        & \qquad \qquad \qquad \qquad \qquad =
        \big( \tu_1 (\tla_1 \rhd u_2) \big)
        \big( (\toy_1 \lhd (\tla_1 \rhd u_2)) \ty' \big)
        \big( \tobe' ((\tla_1 \lhd u_2) \rhd \tobe_2\mone) \big)
        \big( (\tla_1 \lhd u_2) \tla_2\mone \big)
        ,
    \end{aligned}
    \label{ExtendedRibEq_EdgeGluing}
\end{align}
with $(u_1 \lhd \lambda_2\mone) \, \beta_2 \, (u_1 \lhd \lambda_2\mone)\mone = \beta' y'$ and $(\tla_1 \rhd u_2) \, \tbe_1 \, (\tla_1 \rhd u_2)\mone = \tbe' \ty'$.
The wedge fusion is therefore obtained through the symplectic reduction 
\be
    (\cB_{c_1,e} \times \cB_{c_2,e'})/\!/ G_2
    \label{PhaseSpace_WedgeGluing_2}
\ee
with momentum map $y^{v_2}_{c_t} (\oy^{v_2'}_{c_2})\mone = y^{v_2}_{c_t} - \oy^{v_2'}_{c_2} \in G_2^*$ and fused edge variable $\lambda^{c_1}_{ee'} = \lambda^{c_1} (u^{v_2} \rhd (\lambda^{c_2})\mone) \in G_2$.

\subsubsection{Triangulation phase space }
\label{sec:triangulation-main}

By consecutively applying the different gluings, which are associative and commutative \cite{alekseev1998}, we can construct the phase space for a given triangulation/cellular complex. To have a simpler final formula, it is useful to introduce the following phase space
\be
    \cB_{l,e} \cong (\cB_{c_1,e} \times \cB_{c_2e'}) /\!/ \cG 
    \,,\quad 
    \textrm{ with } c_{t'}= c_{t} \,,\,\, t_e=-t_{e'} \textrm{ and } e=-e',  
\ee
which decorates a \textit{full} link, instead of a half link. This phase space is obtained by gluing both the \textit{wedge and the link data} at once, hence by identifying both edge and face data. The resulting phase space is then isomorphic to the atomic phase space (see Proposition \ref{prop:glue-ribbon}). Instead of being associated to half link, it is associated to a full link and its dual wedge, as well as edge and face.   

To build  the 3d triangulation phase space from a set of atomic phase spaces, we need
\begin{itemize} 
    \item three edges closing into a triangle $t$. We consider three atomic phase spaces and fuse them along a common link, recovering a phase space with one single fused face, with one dual link and   three edges together with their dual wedges. This fusion is implemented by a symplectic reduction along two copies of $G_1^*$. Finally we     impose a constraint to close the edges. This is implemented by constraining the product of the three $G_2$ elements 
    and hence inducing a symplectic reduction with respect to $G_2^*$, which we will denote $G_2^{*t}$. Hence the phase space of the triangle $\cP_t$  is obtained by the symplectic reduction along the group $H^t\equiv \big((G_1^*\times G_1^*)\times G^{*t}_2 \big)$
    \be
        \cP_t = 
        \cB_{l,e}^{(3)} /\!/ H^t
        \,, \quad 
        H^t\equiv \big((G_1^*\times G_1^*)\times G^{*t}_2 \big).
    \ee
    We provide the explicit construction in Sec. \ref{sec:ex}. 
    \item four triangles closing into a tetrahedron $\tau$. Once we have the triangle phase space, we glue four of them by identifying the edges pairwise to form the tetrahedron. This is implementing six wedge gluings, by constraining momentum maps in $G_2$. Finally we  impose the closure constraint to close the faces. This is implemented by constraining the product/sum of four $G_1^*$ elements, which is implementing a symplectic reduction along the group  we denote $G_1^{\tau}$. 
    
    Hence the phase space of the tetrahedron $\cP^\tau$  is obtained by the symplectic reduction along the group
    $H^\tau= \big((G_2^*)^{\times 6})\times G^{\tau}_1 \big)$
    \be
        \cP_{\tau} = 
        \cP_t{}^{(4)} /\!/ H^\tau
        \,,\quad 
        H^\tau= \big((G_2^*)^{\times 6})\times G^{\tau}_1 \big). 
    \ee
     We provide the explicit construction in Sec. \ref{sec:ex}. 
\end{itemize}
Since we have both decoration on the triangulation \textit{and} its dual, we also have to impose   geometric constraints that arise on the dual, as well as some possible  compatibility between the two. We have 
\begin{itemize} 
\item  \textit{(1-)flatness constraint.} When gluing the wedges together, we need to remember that we have a trivial $t$-map, which implies that the holonomy that decorates a closed link around the wedge $w$ has to be flat. 
First, we consider $n$ atomic phase spaces and fuse them along a common edge, to recover the phase space with one single fused wedge, its dual edge, and $n$ links around it. This fusion is implemented by a symplectic reduction along $n$ copies of $G_2^*$. Finally we impose a condition to constraint the closed holonomy. This is implemented by constraining the product of $n$  elements in $G_1$ and hence inducing a symplectic reduction with respect to $G_1^*$, which we will note $G_1^{*\, \partial w}$. Therefore, the phase space of a closed  holonomy around an edge $\cP_{\partial w}$ is obtained by the symplectic reduction along the group $H^{\partial w} = \big( (G_2^*)^{\times n-1} \times G_1^{* \, \partial w} \big)$
\be
    \cP_{\partial w} = \cB^{(n)}_{l,e} /\!/ H^{\partial w}
    \,,\quad
    H^{\partial w} = \big( (G_2^*)^{\times n-1} \times G_1^{* \, \partial w} \big) .
\ee

\item \textit{Edge simplicity}. The edge simplicity constraint discussed in \cite{Dittrich:2008ar, Dittrich:2010ey, Asante:2019lki} is weaker than the flatness constraint. It mixes information between the triangulation and its dual. It encodes that the transport of the edge data around a closed link is trivial. If we already assumed the (1-)flatness constraint, then this is automatically implemented. In fact, not demanding the flatness condition first, but instead the edge simplicity condition would amount to allow for the possibility of a non trivial flatness constrain along a closed link, which would be equivalent to have a non-trivial $t$-map. This would then be outside the realm of skeletal crossed modules which we focused on exclusively.   
\end{itemize}

The following definition constitutes the main result of the paper.

\begin{definition}
\label{Def:main_result}
    We define the phase space $\cP$ of a triangulation $\cT$, with dual $\cT^*$ to be 
    \begin{align}
        \cP = \big(\bigtimes_{l\in\cT^*} \cB_{l,e} \big) /\!/ \big( \bigtimes_{\partial w\in\cT^*} H^{\partial w}  \bigtimes_{t\in \cT} H^t \bigtimes_{\tau\in \cT} H^{\tau} \big).
    \end{align}
\end{definition}
We emphasize that the triangulation is not yet fully geometrically consistent. Indeed, we could expect that the triangle decoration in $G_1^*$ to be related to the normal which should be constructed from the edge information (in $G_2$). This would amount to impose the simplicity constraints \cite{Plebanski:1977zz}, which can be a subtle issue \cite{peter}.

We could demand that the decoration of the faces/wedges of the dual polyhedron close. This would be equivalent to demand 2-flatness in the 2-gauge theory picture. Typically this is the condition that encodes dynamics. This is similar to the 2d case, where the 1-flatness condition, around the closed links, encodes dynamics.

\medskip

Upon quantization, we expect this phase space to describe possibly quantum curved 3d discrete geometries. The states, spin 2-networks, will be given in terms of the 2-representations of the crossed module $\cG$. These provide a natural generalization of the notion of spin networks introduced by Penrose \cite{Penrose71angularmomentum:}.

\subsection{Examples} \label{sec:ex}

Given the general construction of the phase space of a 
triangulation, in this section we show some concrete examples. 
We construct the phase space for a triangle and a tetrahedron. These explicit cases allow us to identify the shape of the momentum map constraints inducing the closure and generating the symplectic reduction. 
We also recover the phase space leading to the $G$-networks when focusing on the Poincar\'e 2-group 
equipped with a trivial Poisson structure. 
We recall that while we exclusively discuss triangulations, our construction holds for any cellular decomposition and its dual with minor adjustments.

\subsubsection{Holonomy around an edge.}
For the \textit{(1-)flatness constraint}, it is useful to derive the phase space of a closed loop of links (holonomies) around an edge. For the sake of simplicity, we first consider the shortest non trivial loop, made of two links. \\
To determine this phase space we consider then two atomic phase spaces,  Heisenberg doubles $\cB_{l_i,e_i}$, associated to the full links $l_1 = (c_1 c_1')$, $l_2 = (c_2 c_2')$ and to the edges $e_1 = (v_1 v_1')$, $e_2 = (v_2 v_2')$. According to the \textit{Wedge Gluing Part 2} of Sec. \ref{Sec_GluingRules}, we impose the identification $e = e'$, representing the edges at the common nodes $c_1' = c_2$.
The proper geometric identification is encoded in the gluing condition
\be
    \tla_1 = \lambda_2 .    
    \label{Wedge_GluingConstraint_Loop}
\ee
In order to enforce the closure of the fused holonomy, we identify $c_2' = c_1$, by imposing the momentum map
\be
    u_1 u_2 = 1 .
\ee
The symplectic reduction for a closed holonomy made of two links around an edge is thus
\be
    \cP_{\partial w} = (\cB_{l_1,e_1} \times \cB_{l_2,e_2}) /\!/ (G_2^* \times G_1^*) ,
\ee
with momentum maps $(\tla_1 \lambda_2\mone, u_1 u_2)$ and fused wedge variable (decoration of a closed dual face) 
\be
    G_2^* \ni 
    \oy = \oy_1 + u_1 \rhd \oy_2 + (u_1 \, \obe_2 \, u_1\mone)|_{G_2^*} .
\ee
The generalization to a loop of $n$ links is straightforward: impose $n-1$ momentum maps to identify the $n$ edges
\be
    \lambda_1 = \lambda_2 = \dots = \lambda_n ,
    \label{MomentumMap_Loop_Edges}
\ee
and then impose the closure of the holonomy
\be
    u_1 u_2 \cdots u_n = 1 .
    \label{MomentumMap_Loop_Hol}
\ee
The symplectic reduction for the phase space of a closed holonomy around an edge is 
\be
    \cP_{\partial w} = \cB_{l,e}^{\times n} /\!/ \big((G_2^*)^{\times n-1} \times G_1^*\big) ,
\ee
with momentum maps \eqref{MomentumMap_Loop_Edges} and \eqref{MomentumMap_Loop_Hol}, and fused wedge variable rooted at the node $c_1 = c_n'$
\be
    \oy = \oy_1 + u_1 \rhd \oy_2 + \dots + (u_1 \cdots u_{n-1}) \rhd \oy_n +
    (u_1 \, \obe_2 \, u_1\mone)|_{G_2^*} + \big((u_1 \cdots u_{n_1}) \, \obe_n \, (u_1 \cdots u_{n_1})\mone\big)\big|_{G_2^*} .
\ee

\begin{figure}[H]
    \begin{center}
    \begin{tikzpicture}
     \node[anchor=south west,inner sep=0] (image) at (0,0)
    {\includegraphics[width=0.6\textwidth]{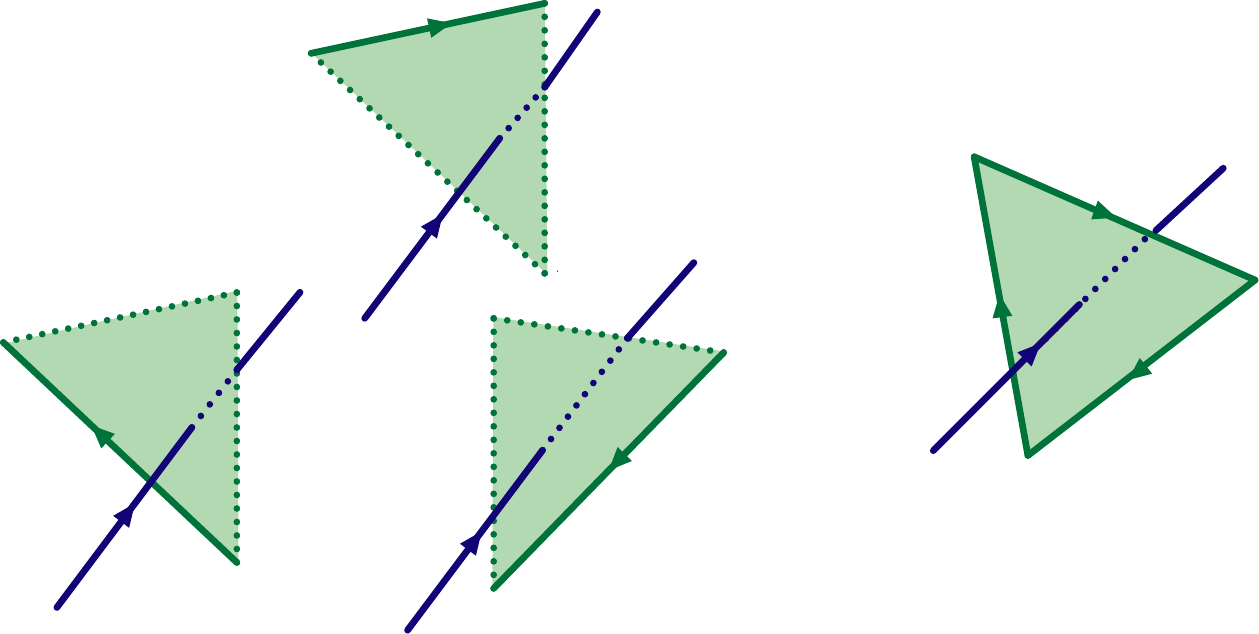}};
    \begin{scope}[x={(image.south east)},y={(image.north west)}]    
    \node at (0.05,0) {$c_1$};
    \node at (0.33,0.5) {$c_2$};
    \node at (0.23,0.57) {$c_1'$};
    \node at (0.5,1) {$c_2'$};
    \node at (0.31,-0.05) {$c_3$};
    \node at (0.55, 0.62) {$c_3'$};
    \tikzstyle{arrow}  = [thick,->,>=stealth]
    \node at (0.6,0.55) (glue1) {};
    \node at (0.75,0.55) (glue2) {};
    \draw[arrow] (glue1) edge (glue2);
    \node at (0,0.52) {$v_1'$};
    \node at (0.2,0.07) {$v_1$};
    \node at (0.6, 0.45) {$v_3$};
    \node at (0.41,0.05) {$v_3'$};
    \node at (0.23, 0.95) {$v_2$};
    \node at (0.4,1.05) {$v_2'$};
    \node at (0.7, 0.2) {$c_1 = c_2 =c_3$};
    \node at (0.9, 0.3) {$v_1 = v_3'$};
    \node at (1.05, 0.5) {$v_3 = v_2'$};
    \node at (0.8,0.8) {$v_1'=v_2$};
    \end{scope} 
    \end{tikzpicture}
    \caption{Identifying three links and fusing three faces to create the phase space associated to a triangle.}
    \label{fig:my_label}
    \end{center}
\end{figure}
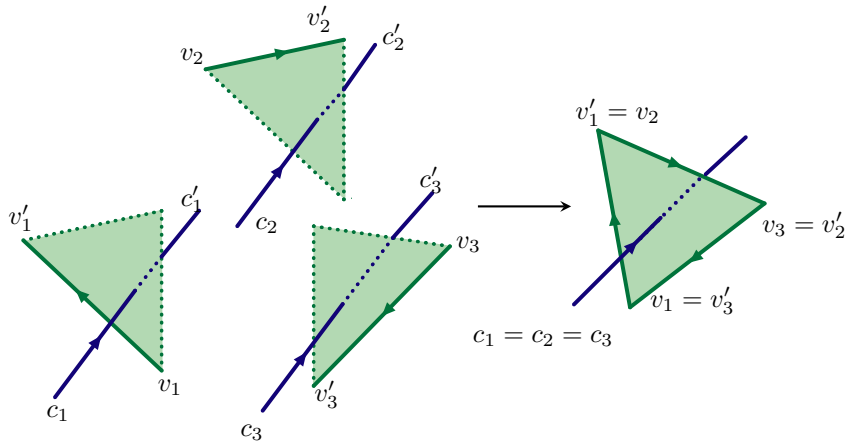

\subsubsection{Triangle.} 
To determine the phase space of a triangle, we take three Heisenberg doubles $\cB_{l_i,e_i}$, with $i = 1,2,3$: three atomic phase spaces are associated to the full links\footnote{The result would be similar if only using half links.} $l_i=(c_i, c_i')$ and edges $e_i=(v_i v_i')$ as well as their respective dual structures.

The three faces of the atomic pieces are fused together into a single surface ---
the triangle face. We will fuse the face decoration in pairs since the fusion process is associative.  We will then require that the three edge decorations close.

First, we glue two phase spaces by identifying the links $l_1=(c_1 c_1') = (c_2 c_2')=l_2$ in the representation sitting at the common vertices $v_1' = v_2$ (we refer to the paragraph \ref{Sec_GluingRules} for a more detailed explanation). The geometric identification induces the momentum constraint
\be
    u_1 = \tu_2 .
\ee
Then we add the third phase space, by imposing the condition $l_1 = l_2 = l_3$  on the links represented at the common vertices $v_2' = v_3$. This time the identification induces the momentum constraint
\be
    u_2 = \tu_3 .
\ee
Putting everything together, we get the full ribbon constraint 
\be
    \tu_1
    (\ty_1 \ty' \ty'')
    \big( \tbe_1 (\lambda_1 \rhd \tbe_2) \tbe' ((\lambda_1 \lambda_2) \rhd \tbe_3) \tbe'' \big)
    (\tla_1 \tla_2 \tla_3)
    =
    (\lambda_1 \lambda_2 \lambda_3)
    \big( (\obe_1 \lhd (\lambda_2\lambda_3)) ((\obe_2 \beta') \lhd \lambda_3) \obe_3 \beta'' \big)
    (y'' \oy_3)
    u_3 ,
    \label{Triangle_ExtrendedRibbonEq}
\ee
with $\tla_1 \, \ty_2 \, \tla_1\mone = \ty' \tbe'$, $\lambda_2\mone \, \oy_1 \, \lambda_2 = y' \beta'$, $(\tla_1 \tla_2) \, \ty_3 \, (\tla_1 \tla_2)\mone = \ty'' \tbe''$ and $\lambda_3\mone (y' \oy_2) \lambda_3 = y'' \beta''$.

Last, in order to close the triangle we need to identify the vertices $v_3' = v_1$. This identification allows us to define  the closure constraint (we recall that the $t$-map is trivial) 
\be
  G_2\ni  \lambda_1 \lambda_2 \lambda_3 = 1.
\ee
The symplectic reduction for a triangle is thus
\be
    \cP_t = (\cB_{l_1,e_1} \times \cB_{l_2,e_2} \times \cB_{l_3,e_3}) /\!/ \big((G_1^* \times G_1^*) \times G_2^*\big) ,
    \label{PhaseSpace_Triangle}
\ee
with momentum maps $(u_1 \tu_2\mone , u_2 \tu_3\mone , \lambda_1 \lambda_2 \lambda_3)$ and triangle fused face variable %
\be
    \obe_t = 
    \obe_1 \lhd (\lambda_2 \lambda_3) + \obe_2 \lhd \lambda_3 + \obe_3 +
    \big( \lambda_3\mone \oy_2 \lambda_3 \big)\big|_{G_1^*} +
    \big( (\lambda_2\lambda_3)\mone \, \oy_1 \, (\lambda_2\lambda_3) \big)\big|_{G_1^*} .
    \label{Triangle_TotalFace}
\ee

\begin{figure}[H]
    \begin{center}
    \begin{tikzpicture}
    \node[anchor=south west,inner sep=0] (image) at (0,0)
    {\includegraphics[width=0.6\textwidth]{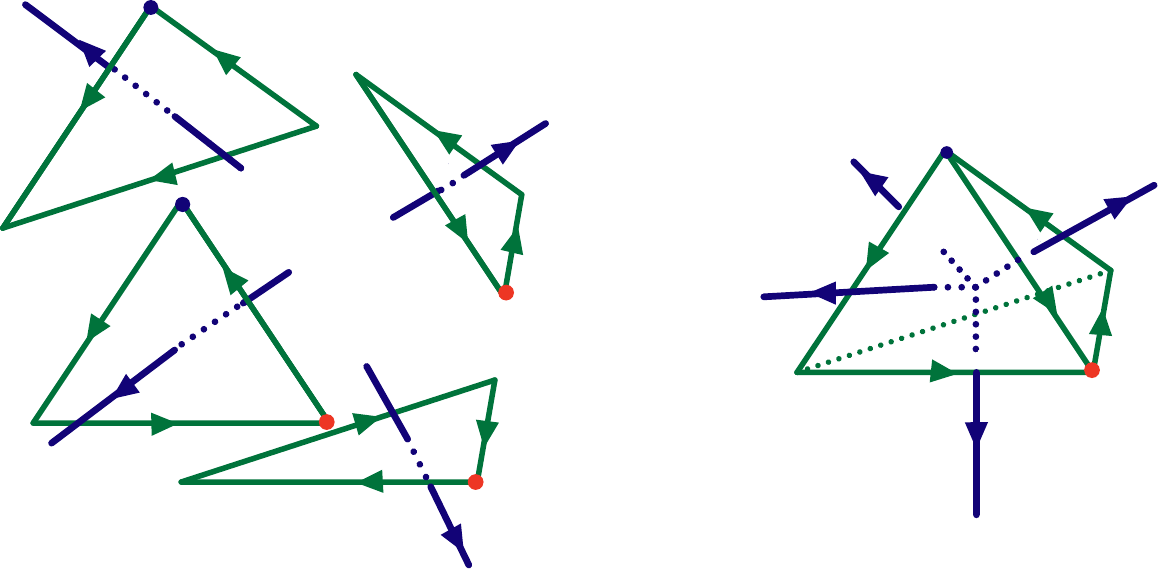}};
    \begin{scope}[x={(image.south east)},y={(image.north west)}]    
    \node at (0.45,0.12) {$v_{1;2}$};
    \node at (0.3,0.3) {$v_{1;1}$};
    \node at (0.49, 0.5) {$v_{3;1}$};
    \node at (0.16, 1) {$v_{4;1}$};
    \node at (0.11,0.63) {$v_{1;2}$};
    \node at (1,0.3) {$v_{1;1}=v_{2;1}=v_{3;1}$};
    \node at (0.85, 0.8) {$v_{4;1} = v_{1;2}$};
    \node at (0, 1) {$c_{4}'$};
    \node at (0.05,0.15) {$c_{1}'$};
    \node at (0.44,0) {$c_{2}'$};
    \node at (0.52,0.8) {$c_{3}'$};
    \node at (0.25,0.55) {$c_{1}$};
    \node at (0.18,0.5) {$e_1$};
    \node at (0.23, 0.67) {$c_{4}$}; 
    \node at (0.34, 0.37) {$c_{2}$};
    \node at (0.35,0.57) {$c_{3}$};
    
    \tikzstyle{arrow}  = [thick,->,>=stealth]
    \node at (0.55,0.55) (glue1) {};
    \node at (0.7,0.55) (glue2) {};
    \draw[arrow] (glue1) edge (glue2);
    
    \node at (0.85,0.55) {$c$};
    \end{scope}
    \end{tikzpicture}
    \caption{The phase space of a tetrahedron is constructed by gluing four triangles. To properly glue triangles, we must root the face variables at a common point. Three triangles are glued by identifying vertices $v_1^i$. The fourth triangle must have its root transported along an edge ($e_1$) before being identified. The links are connected at a common point $c \equiv c_1=c_2=c_3=c_4$.}
    \label{Fig:Tetra}
    \end{center}
\end{figure}
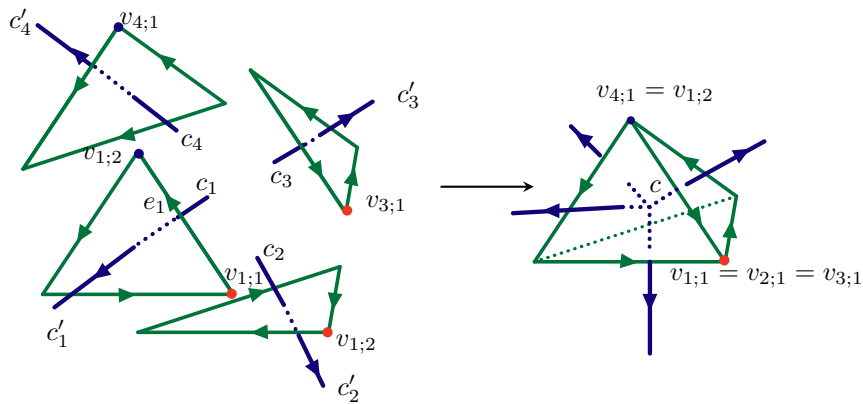

\subsubsection{Tetrahedron.}
In order to derive the phase space of a tetrahedron, we first consider twelve copies of the Heisenberg double, $\cB_{l_{a;i},e_{a;i}}$ for $a=1,2,3,4$ and $i=1,2,3$, associated to full links $l_{a;i} = (c_{a;i}c_{a;i}')$ and edges $e_{a;i} = (v_{a;i}v_{a;i}')$. Later on, the nodes $c_{a;1},c_{a;2},c_{a;3}$ will be identified as the source of the link dual to triangle $a$ and noted as $c_a$; then, the sources of the four links $c_1,c_2,c_3,c_4$ will be further identified as the center of the tetrahedron and noted as $c$. 

\smallskip 

To construct a tetrahedron, we have to impose first the momentum maps necessary to construct the  four triangle phase spaces:
\be
    u_{a;1} = \tu_{a;2} 
    \,\,,\quad
    u_{a;2} = \tu_{a;3}
    \,\,\, \text{ and } \,\,\,
    \lambda_{a;1} \lambda_{a;2} \lambda_{a;3} = 1
    \,\,,\,\, \text{ for } \,\,\,
    a = 1,2,3,4 .
    \label{MomentumMaps_Tetra_Links}
\ee
Then, we have to fuse the triangle phase spaces by identifying the three edges of each triangle with one of the edges of the other three triangles: use the \textit{Wedge Gluing Part 1} of Sec. \ref{Sec_GluingRules} to glue the triangles. According to Fig. \ref{Fig:Tetra}, we impose the the six momentum maps
\be
    \lambda_{1;1} = \lambda_{3;3}\mone
    \,\,,\quad
    \lambda_{1;2} = \lambda_{4;3}\mone
    \,\,,\quad
    \lambda_{1;3} = \lambda_{2;1}\mone
    \,\,,\quad
    \lambda_{2;2} = \lambda_{4;2}\mone
    \,\,,\quad
    \lambda_{2;3} = \lambda_{3;1}\mone
    \,\,,\quad
    \lambda_{3;2} = \lambda_{4;1}\mone .
    \label{MomentumMaps_Tetra_Edges}
\ee
Since the symplectic reduction prescription is commutative, the order in which we implement the above eighteen momentum maps does not matter. 
Computing the associated extended ribbon equation provides us the 
 expression of the fused decoration.
For instance, we can compute the extended ribbon equation for the four triangles (as in \eqref{Triangle_ExtrendedRibbonEq}) to derive the expression of the four fused face decorations
\be
    \begin{aligned}
        \obe_{t_1} & = 
        \beta_{1;1} + \obe_{1;2} \lhd \lambda_{1;3} + \obe_{1;3} +
        \big( \lambda_{1;3}\mone \oy_{1;2} \lambda_{1;3} \big)\big|_{G_1^*} +
        \big( (\lambda_{1;2}\lambda_{1;3})\mone \, \oy_{1;1} \, (\lambda_{1;2}\lambda_{1;3}) \big)\big|_{G_1^*} , \\
        \obe_{t_2} & = 
        \beta_{2;1} + \obe_{2;2} \lhd \lambda_{2;3} + \obe_{2;3} +
        \big( \lambda_{2;3}\mone \oy_{2;2} \lambda_{2;3} \big)\big|_{G_1^*} +
        \big( (\lambda_{2;2}\lambda_{2;3})\mone \, \oy_{2;1} \, (\lambda_{2;2}\lambda_{2;3}) \big)\big|_{G_1^*} , \\
        \obe_{t_3} & = 
        \beta_{3;1} + \obe_{3;2} \lhd \lambda_{3;3} + \obe_{3;3} +
        \big( \lambda_{3;3}\mone \oy_{3;2} \lambda_{3;3} \big)\big|_{G_1^*} +
        \big( (\lambda_{3;2}\lambda_{3;3})\mone \, \oy_{3;1} \, (\lambda_{3;2}\lambda_{3;3}) \big)\big|_{G_1^*} , \\
        \obe_{t_4} & = 
        \beta_{4;1} + \obe_{4;2} \lhd \lambda_{4;3} + \obe_{4;3} +
        \big( \lambda_{4;3}\mone \oy_{4;2} \lambda_{4;3} \big)\big|_{G_1^*} +
        \big( (\lambda_{4;2}\lambda_{4;3})\mone \, \oy_{4;1} \, (\lambda_{4;2}\lambda_{4;3}) \big)\big|_{G_1^*} .
    \end{aligned}
    \label{Tetrahedron_Faces}
\ee
From the ribbon equation we also know that the above face decorations are represented at the source of the four links, common node $c \equiv c_1 = c_2 = c_3 = c_4$ and rooted respectively at the vertices $v_{1;1}=v_{1;3}'$, $v_{2;1}=v_{2;3}'$, $v_{3;1}=v_{3;3}'$, $v_{4;1}=v_{4;3}'$. 

The extended ribbon equation is also useful to change the reference frame: we 
act by some edge or link decoration on the ribbon equation. 
For instance, starting by the twelve atomic phase spaces and fusing them pairwise by the momentum maps \eqref{MomentumMaps_Tetra_Edges}, we compute the six extended ribbon equations associated to the six edges of a tetrahedron; in principle, they are given in different reference frames. Then, acting with some edge decorations on such ribbon equations, is then possible to derive the expression of the six fused wedge decorations dual to the six edges, all rooted at the node $c$, and represented at the common vertex $v \equiv v_{1;1} = v_{2;1} = v_{3;1}$ (red vertex in Fig. \eqref{Fig:Tetra}):
\be
    \begin{aligned}
        \oy_{12} & =
        - \oy_{2;1} + (\lambda_{1;3} \, \oy_{1;3} \, \lambda_{1;3}\mone)|_{G_2^*}
        \,,\\
        \oy_{13} & =
        \oy_{1;1} - (\lambda_{3;3} \, \oy_{3;3} \, \lambda_{3;3}\mone)|_{G_2^*} 
        \,, \\
        \oy_{14} & =
        \oy_{1;2} - (\lambda_{4;3} \, \oy_{4;3} \, \lambda_{4;3}\mone)|_{G_2^*} \,,
    \end{aligned}
    \quad
    \begin{aligned}
        \oy_{23} & =
        -\oy_{3;1} + (\lambda_{2;3} \, \oy_{2;3} \, \lambda_{2;3}\mone)|_{G_2^*} 
        \,,\\
        \oy_{24} & =
        \oy_{2;2} - (\lambda_{4;2} \, \oy_{4;2} \, \lambda_{4;2}\mone)|_{G_2^*} 
        \,,\\
        \oy_{34} & =
        \oy_{3;2} - (\lambda_{4;1} \, \oy_{4;1} \, \lambda_{4;1}\mone)|_{G_2^*} \,.
    \end{aligned}
    \label{Tetra_FusedWedges}
\ee
where we denoted $\oy_{ab}$ the fused wedge decoration shared by triangles $a$ and $b$.
\smallskip \\
The construction of the tetrahedron phase space is completed with the last momentum map, known as 1-Gauss constraint, that enforces the closure of the four faces. According to Fig. \ref{Fig:Tetra}, three of the four triangles share the same reference frame, while the fourth is at a different vertex (vertex $v_{4;1}$ instead of $v$). To impose the 1-Gauss constraint, we thus act from the left on the ribbon equation of the fourth triangle by the edge decoration that connect the vertices $v_{4;1}$ and $v$, $\lambda_{1;1}$. Once we have the expression of the four face decorations all at the same reference frame, we can correctly impose the last momentum map
\be
    \obe_{t_1} + \obe_{t_2} + \obe_{t_3} + \lambda_{1;1} \rhd \obe_{t_4} + (\lambda_{1;1} \, \oy_{t_4} \, \lambda_{1;1}\mone)|_{G_1^*} = 0 ,
    \label{1-Gauss}
\ee
where we used the expressions in \eqref{Tetrahedron_Faces} and
\be
    \oy_{t_4} = \oy_{4:3} + (\lambda_{4;3}\mone \, \oy_{4;2} \, \lambda_{4;3})|_{G_2^*} + \big( (\lambda_{4;2}\lambda_{4;3})\mone \, \oy_{4;1} \, (\lambda_{4;2}\lambda_{4;3}) \big)\big|_{G_2^*} .
\ee
The symplectic reduction for a tetrahedron phase space is thus
\be
    \cP_{\tau} = 
    \cB_{l,e}^{\times 12} /\!/ \big(((G_1^* \times G_1^*) \times G_2^*)^{\times 4} \times ((G_2^*)^{\times 6} \times G_1)\big) ,
    \label{PhaseSpace_Tetrahedron}
\ee
with momentum maps \eqref{MomentumMaps_Tetra_Links}, \eqref{MomentumMaps_Tetra_Edges} and \eqref{1-Gauss}. 
We can alternatively express the constraint \eqref{1-Gauss} as a combination of four face decorations plus the contribution of the six fused wedge decorations \eqref{Tetra_FusedWedges}. This expression will be particularly useful in the next subsection. As we did for the fused wedges, we denote $\lambda_{ab}$ the edge shared by triangles $a$ and $b$; in particular we have 
\be
    \lambda_{12} \equiv \lambda_{2;1}
    \,\,,\quad
    \lambda_{13} \equiv \lambda_{1;1}
    \,\,,\quad
    \lambda_{14} \equiv \lambda_{1;2}
    \,\,,\quad
    \lambda_{23} \equiv \lambda_{3;1}
    \,\,,\quad
    \lambda_{24} \equiv \lambda_{2;2}
    \,\,,\quad
    \lambda_{34} \equiv \lambda_{3;2} .
\ee
The constraint \eqref{1-Gauss} then becomes
\begin{align}
    b_1 + b_2 + b_3 + \lambda_{13} \rhd b_4 
    & + 
    (\lambda_{12} \, \oy_{12} \, \lambda_{12}\mone)|_{G_1^*} +
    (\lambda_{13} \, \oy_{13} \, \lambda_{13}\mone)|_{G_1^*} + 
    \lambda_{13} \rhd (\lambda_{14} \, \oy_{14} \, \lambda_{14}\mone)|_{G_1^*} 
    \nonumber \\
    & +
    (\lambda_{23} \, \oy_{23} \, \lambda_{23}\mone)|_{G_1^*} +
    \lambda_{12} \rhd (\lambda_{24} \, \oy_{24} \, \lambda_{24}\mone)|_{G_1^*} + 
    \lambda_{23} \rhd (\lambda_{34} \, \oy_{34} \, \lambda_{34}\mone)|_{G_1^*} = 0 ,
    \label{1-Gauss_FusedWedges}
\end{align}
with the face decorations
\be
    \begin{aligned}
        b_1 & =
        \beta_{1;1} + \lambda_{13} \rhd \beta_{1;2} + \obe_{1;3} + 
        \lambda_{12} \rhd (\lambda_{21}\mone \, \oy_{1;3} \, \lambda_{21})|_{G_1^*} +
        \big( \lambda_{13} \, (\lambda_{14} \, \oy_{1;2} \, \lambda_{14}\mone)|_{G_2^*} \, \lambda_{13}\mone \big)\big|_{G_1^*}
        , \\
        b_2 & = 
        - \beta_{2;1} + \lambda_{12} \rhd \beta_{2;2} + \obe_{2;3} +
        \lambda_{34}\mone \rhd (\lambda_{32}\mone \, \oy_{2;3} \, \lambda_{32})|_{G_1^*} +
        \big( \lambda_{12} \, (\lambda_{24} \, \oy_{2;2} \, \lambda_{24}\mone)|_{G_2^*} \, \lambda_{12}\mone \big)\big|_{G_1^*}
        , \\
        b_3 & =
        - \beta_{3;1} + \lambda_{32} \rhd \beta_{3;2} - \obe_{3;3}
        - \lambda_{13} \rhd (\lambda_{13}\mone \, \oy_{3;3} \, \lambda_{13})|_{G_1^*} +
        \big( \lambda_{23} \, (\lambda_{34} \, \oy_{3;2} \, \lambda_{34}\mone)|_{G_2^*} \, \lambda_{23}\mone \big)\big|_{G_1^*}
        , \\
        b_4 & =
        - \beta_{4;1} - \lambda_{34}\mone \rhd \beta_{4;2} - \obe_{4;3}
        - \lambda_{13}\mone \rhd \big(
        (\lambda_{13} \, \oy_{4;3} \, \lambda_{13}\mone)|_{G_2^*} + \lambda_{12}\mone \rhd (\lambda_{14}\mone \, \oy_{4;3} \, \lambda_{14})|_{G_1^*} 
        \\
        & \quad \, +
        (\lambda_{12} \, \oy_{4;2} \, \lambda_{12}\mone)|_{G_2^*} + \lambda_{23} \rhd (\lambda_{24}\mone \, \oy_{4;2} \, \lambda_{24})|_{G_1^*} +
        (\lambda_{23} \, \oy_{4;1} \, \lambda_{23}\mone)|_{G_2^*} + \lambda_{13} \rhd (\lambda_{34}\mone \, \oy_{4;1} \, \lambda_{34})|_{G_1^*}
        \big) .
    \end{aligned}
    \label{Tetra_FusedFaces_b}
\ee
Note that, we derived the expression \eqref{1-Gauss_FusedWedges}, such that each fused wedge decoration $\oy_{ab}$ is conjugated by the respective dual edge decoration $\lambda_{ab}$.

\subsubsection{4-simplex boundary.}
As  we managed to construct the tetrahedron phase space, the fundamental building block of a 3$d$ space, we can use it to build any triangulation phase space. Let us give here an example of a closed piece of triangulation: we build the phase space of the boundary of a 4-simplex which was the main example considered in \cite{Asante:2019lki}. 

The boundary of a 4-simplex is given in terms of five tetrahedra, so we consider five tetrahedron phase spaces and fuse them by identifying each of the faces of each tetrahedron with a face of one of the other four tetrahedra.
The starting point to derive the 4-simplex boundary phase space are then sixty copies of the Heisenberg double $\cB_{l_{A;a;i},e_{A;a;i}}$, for $A=1,2,3,4,5$, $a=1,2,3,4$, $i=1,2,3$. For simplicity and to not be redundant, it is simpler to start by five copies of tetrahedron phase space $\cP_{\tau_A}$, as given in \eqref{PhaseSpace_Tetrahedron}.

We represent the 4-simplex boundary construction in Fig. \ref{Fig:4-simplex}. We denote by $c_A$ the center of the tetrahedron $\tau_A$. The five tetrahedra share five vertices: $v$ (in black) is shared by tetrahedra 1,2,3,5, $v'$ (in white) is shared by tetrahedra 1,3,4,5, $v''$ (in gray) is shared by tetrahedra 1,2,4,5, $v'''$ (in yellow) is shared by tetrahedra 2,3,4,5 and $v''''$ (in green) is shared by tetrahedra 1,2,3,4.
Let $t_{A;a}$ be the triangle $a$ of tetrahedron $A$ and $\obe_{t_{A;a}}$ of  \eqref{Tetrahedron_Faces} (or $b_{A;a}$ of  \eqref{Tetra_FusedFaces_b}) is the respective decoration represented at the center of tetrahedron $A$. $\oy_{A;ab}$ is instead  the variable that decorates the fused wedge shared by triangles $a$ and $b$ of tetrahedron $A$, rooted at the center of it.
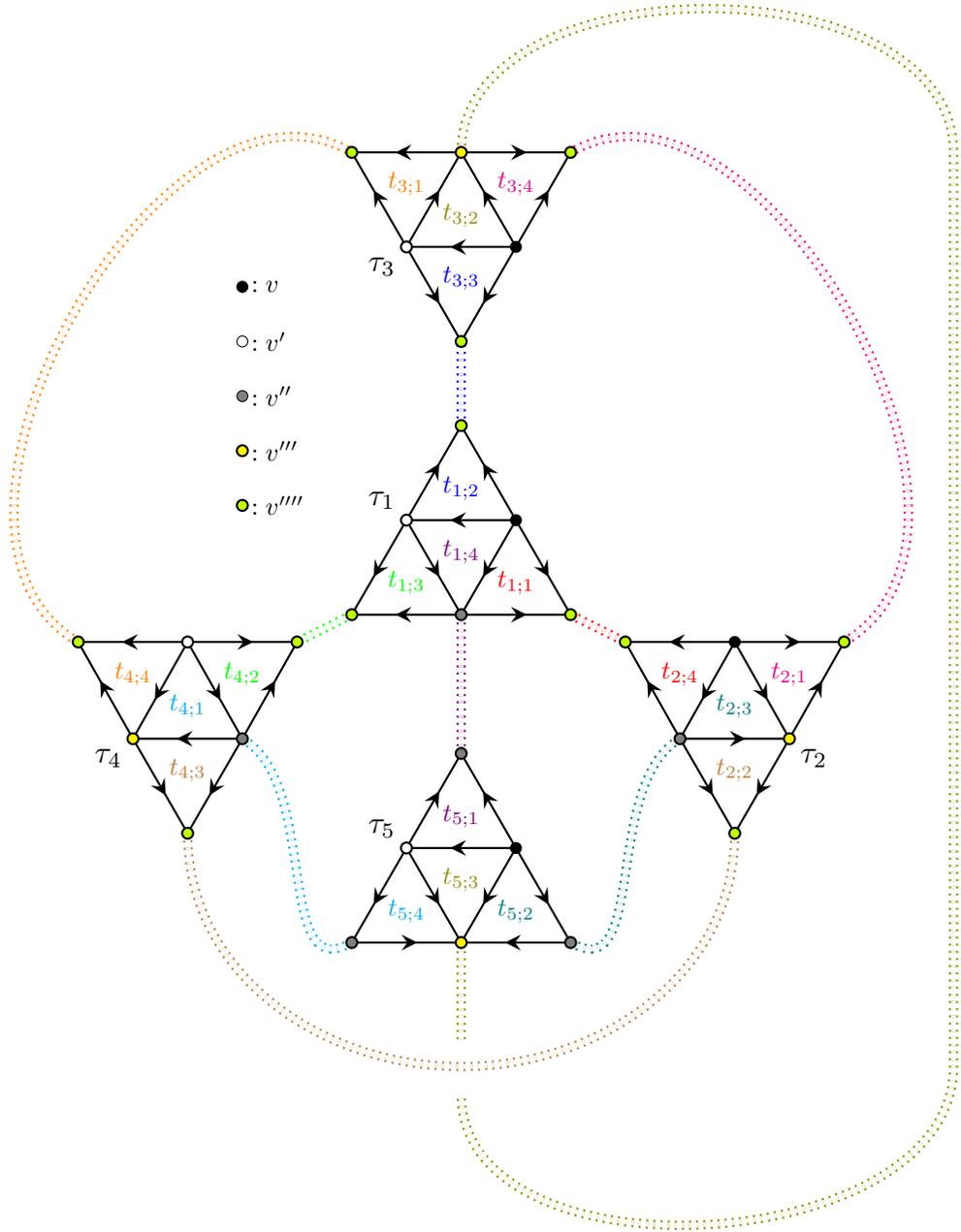
\begin{figure}
    \centering
        \begin{tikzpicture}[scale=0.75]
        \coordinate (v11) at (0,-5);
        \coordinate (v12) at (2,-5);
        \coordinate (v13) at (1,-6.73);
        \coordinate (v142) at (1,-3.27);
        \coordinate (v143) at (-1,-6.73);
        \coordinate (v144) at (3,-6.73);
        
        \coordinate (c11) at ($1/3*(v11) + 1/3*(v12) + 1/3*(v13)$);
        \coordinate (c12) at ($1/3*(v11) + 1/3*(v12) + 1/3*(v142)$);
        \coordinate (c13) at ($1/3*(v11) + 1/3*(v13) + 1/3*(v143)$);
        \coordinate (c14) at ($1/3*(v12) + 1/3*(v13) + 1/3*(v144)$);
        
        \node[above left , scale=1.2] at (v11) {$\tau_1$};
        
        \draw[- , thick] (v11) -- (v12) -- (v13) -- cycle;
        \draw[- , thick] (v12) -- (v144) -- (v13);
        \draw[- , thick] (v13) -- (v143) -- (v11);
        \draw[- , thick] (v11) -- (v142) -- (v12);
        
        \node[violet] at (c11) {$t_{1;4}$};
        \node[blue] at (c12) {$t_{1;2}$};
        \node[green] at (c13) {$t_{1;3}$};
        \node[red] at (c14) {$t_{1;1}$};
        
        \draw[decoration={markings,mark=at position 0.6 with {\arrow[scale=1.5,thick,>=stealth]{>}}},postaction={decorate}] (v12) -- (v11);
        \draw[decoration={markings,mark=at position 0.6 with {\arrow[scale=1.5,thick,>=stealth]{>}}},postaction={decorate}] (v12) -- (v13);
        \draw[decoration={markings,mark=at position 0.6 with {\arrow[scale=1.5,thick,>=stealth]{>}}},postaction={decorate}] (v12) -- (v142);
        \draw[decoration={markings,mark=at position 0.6 with {\arrow[scale=1.5,thick,>=stealth]{>}}},postaction={decorate}] (v12) -- (v144);
        \draw[decoration={markings,mark=at position 0.6 with {\arrow[scale=1.5,thick,>=stealth]{>}}},postaction={decorate}] (v13) -- (v144);
        \draw[decoration={markings,mark=at position 0.6 with {\arrow[scale=1.5,thick,>=stealth]{>}}},postaction={decorate}] (v11) -- (v142);
        \draw[decoration={markings,mark=at position 0.6 with {\arrow[scale=1.5,thick,>=stealth]{>}}},postaction={decorate}] (v13) -- (v143);
        \draw[decoration={markings,mark=at position 0.6 with {\arrow[scale=1.5,thick,>=stealth]{>}}},postaction={decorate}] (v11) -- (v143);
        \draw[decoration={markings,mark=at position 0.6 with {\arrow[scale=1.5,thick,>=stealth]{>}}},postaction={decorate}] (v11) -- (v13);

        \coordinate (v21) at (5,-9);
        \coordinate (v22) at (7,-9);
        \coordinate (v23) at (6,-7.23);
        \coordinate (v242) at (8,-7.23);
        \coordinate (v243) at (6,-10.73);
        \coordinate (v244) at (4,-7.23);
        
        \coordinate (c21) at ($1/3*(v21) + 1/3*(v22) + 1/3*(v23)$);
        \coordinate (c22) at ($1/3*(v22) + 1/3*(v23) + 1/3*(v242)$);
        \coordinate (c23) at ($1/3*(v21) + 1/3*(v22) + 1/3*(v243)$);
        \coordinate (c24) at ($1/3*(v21) + 1/3*(v23) + 1/3*(v244)$);
        
        \node[below right , scale=1.2] at (v22) {$\tau_2$};
        
        \draw[- , thick] (v21) -- (v22) -- (v23) -- cycle;
        \draw[- , thick] (v22) -- (v242) -- (v23);
        \draw[- , thick] (v22) -- (v243) -- (v21);
        \draw[- , thick] (v21) -- (v244) -- (v23);
        
        \node[teal] at (c21) {$t_{2;3}$};
        \node[magenta] at (c22) {$t_{2;1}$};
        \node[brown] at (c23) {$t_{2;2}$};
        \node[red] at (c24) {$t_{2;4}$};
        
        \draw[decoration={markings,mark=at position 0.6 with {\arrow[scale=1.5,thick,>=stealth]{>}}},postaction={decorate}] (v23) -- (v244);
        \draw[decoration={markings,mark=at position 0.6 with {\arrow[scale=1.5,thick,>=stealth]{>}}},postaction={decorate}] (v23) -- (v242);
        \draw[decoration={markings,mark=at position 0.6 with {\arrow[scale=1.5,thick,>=stealth]{>}}},postaction={decorate}] (v23) -- (v21);
        \draw[decoration={markings,mark=at position 0.6 with {\arrow[scale=1.5,thick,>=stealth]{>}}},postaction={decorate}] (v23) -- (v22);
        \draw[decoration={markings,mark=at position 0.6 with {\arrow[scale=1.5,thick,>=stealth]{>}}},postaction={decorate}] (v21) -- (v22);
        \draw[decoration={markings,mark=at position 0.6 with {\arrow[scale=1.5,thick,>=stealth]{>}}},postaction={decorate}] (v21) -- (v244);
        \draw[decoration={markings,mark=at position 0.6 with {\arrow[scale=1.5,thick,>=stealth]{>}}},postaction={decorate}] (v22) -- (v242);
        \draw[decoration={markings,mark=at position 0.6 with {\arrow[scale=1.5,thick,>=stealth]{>}}},postaction={decorate}] (v21) -- (v243);
        \draw[decoration={markings,mark=at position 0.6 with {\arrow[scale=1.5,thick,>=stealth]{>}}},postaction={decorate}] (v22) -- (v243);

        \coordinate (v31) at (0,0);
        \coordinate (v32) at (2,0);
        \coordinate (v33) at (1,1.73);
        \coordinate (v342) at (3,1.73);
        \coordinate (v343) at (1,-1.73);
        \coordinate (v344) at (-1,1.73);
        
        \coordinate (c31) at ($1/3*(v31) + 1/3*(v32) + 1/3*(v33)$);
        \coordinate (c32) at ($1/3*(v32) + 1/3*(v33) + 1/3*(v342)$);
        \coordinate (c33) at ($1/3*(v31) + 1/3*(v32) + 1/3*(v343)$);
        \coordinate (c34) at ($1/3*(v31) + 1/3*(v33) + 1/3*(v344)$);
        
        \node[below left , scale=1.2] at (v31) {$\tau_3$};
        
        \draw[- , thick] (v31) -- (v32) -- (v33) -- cycle;
        \draw[- , thick] (v32) -- (v342) -- (v33);
        \draw[- , thick] (v32) -- (v343) -- (v31);
        \draw[- , thick] (v31) -- (v344) -- (v33);
        
        \node[olive] at (c31) {$t_{3;2}$};
        \node[magenta] at (c32) {$t_{3;4}$};
        \node[blue] at (c33) {$t_{3;3}$};
        \node[orange] at (c34) {$t_{3;1}$};
        
        \draw[decoration={markings,mark=at position 0.6 with {\arrow[scale=1.5,thick,>=stealth]{>}}},postaction={decorate}] (v32) -- (v342);
        \draw[decoration={markings,mark=at position 0.6 with {\arrow[scale=1.5,thick,>=stealth]{>}}},postaction={decorate}] (v32) -- (v343);
        \draw[decoration={markings,mark=at position 0.6 with {\arrow[scale=1.5,thick,>=stealth]{>}}},postaction={decorate}] (v32) -- (v31);
        \draw[decoration={markings,mark=at position 0.6 with {\arrow[scale=1.5,thick,>=stealth]{>}}},postaction={decorate}] (v32) -- (v33);
        \draw[decoration={markings,mark=at position 0.6 with {\arrow[scale=1.5,thick,>=stealth]{>}}},postaction={decorate}] (v33) -- (v342);
        \draw[decoration={markings,mark=at position 0.6 with {\arrow[scale=1.5,thick,>=stealth]{>}}},postaction={decorate}] (v33) -- (v344);
        \draw[decoration={markings,mark=at position 0.6 with {\arrow[scale=1.5,thick,>=stealth]{>}}},postaction={decorate}] (v31) -- (v344);
        \draw[decoration={markings,mark=at position 0.6 with {\arrow[scale=1.5,thick,>=stealth]{>}}},postaction={decorate}] (v31) -- (v343);
        \draw[decoration={markings,mark=at position 0.6 with {\arrow[scale=1.5,thick,>=stealth]{>}}},postaction={decorate}] (v31) -- (v33);

        \coordinate (v41) at (-5,-9);
        \coordinate (v42) at (-3,-9);
        \coordinate (v43) at (-4,-7.23);
        \coordinate (v442) at (-2,-7.23);
        \coordinate (v443) at (-4,-10.73);
        \coordinate (v444) at (-6,-7.23);
        
        \coordinate (c41) at ($1/3*(v41) + 1/3*(v42) + 1/3*(v43)$);
        \coordinate (c42) at ($1/3*(v42) + 1/3*(v43) + 1/3*(v442)$);
        \coordinate (c43) at ($1/3*(v41) + 1/3*(v42) + 1/3*(v443)$);
        \coordinate (c44) at ($1/3*(v41) + 1/3*(v43) + 1/3*(v444)$);
        
        \node[below left , scale=1.2] at (v41) {$\tau_4$};
        
        \draw[- , thick] (v41) -- (v42) -- (v43) -- cycle;
        \draw[- , thick] (v42) -- (v442) -- (v43);
        \draw[- , thick] (v42) -- (v443) -- (v41);
        \draw[- , thick] (v41) -- (v444) -- (v43);
        
        \node[cyan] at (c41) {$t_{4;1}$};
        \node[green] at (c42) {$t_{4;2}$};
        \node[brown] at (c43) {$t_{4;3}$};
        \node[orange] at (c44) {$t_{4;4}$};
        
        \draw[decoration={markings,mark=at position 0.6 with {\arrow[scale=1.5,thick,>=stealth]{>}}},postaction={decorate}] (v43) -- (v442);
        \draw[decoration={markings,mark=at position 0.6 with {\arrow[scale=1.5,thick,>=stealth]{>}}},postaction={decorate}] (v42) -- (v442);
        \draw[decoration={markings,mark=at position 0.6 with {\arrow[scale=1.5,thick,>=stealth]{>}}},postaction={decorate}] (v43) -- (v42);
        \draw[decoration={markings,mark=at position 0.6 with {\arrow[scale=1.5,thick,>=stealth]{>}}},postaction={decorate}] (v43) -- (v41);
        \draw[decoration={markings,mark=at position 0.6 with {\arrow[scale=1.5,thick,>=stealth]{>}}},postaction={decorate}] (v42) -- (v41);
        \draw[decoration={markings,mark=at position 0.6 with {\arrow[scale=1.5,thick,>=stealth]{>}}},postaction={decorate}] (v43) -- (v444);
        \draw[decoration={markings,mark=at position 0.6 with {\arrow[scale=1.5,thick,>=stealth]{>}}},postaction={decorate}] (v41) -- (v444);
        \draw[decoration={markings,mark=at position 0.6 with {\arrow[scale=1.5,thick,>=stealth]{>}}},postaction={decorate}] (v41) -- (v443);
        \draw[decoration={markings,mark=at position 0.6 with {\arrow[scale=1.5,thick,>=stealth]{>}}},postaction={decorate}] (v42) -- (v443);

        \coordinate (v51) at (0,-11);
        \coordinate (v52) at (2,-11);
        \coordinate (v53) at (1,-12.73);
        \coordinate (v542) at (1,-9.27);
        \coordinate (v543) at (-1,-12.73);
        \coordinate (v544) at (3,-12.73);
        
        \coordinate (c51) at ($1/3*(v51) + 1/3*(v52) + 1/3*(v53)$);
        \coordinate (c52) at ($1/3*(v51) + 1/3*(v52) + 1/3*(v542)$);
        \coordinate (c53) at ($1/3*(v51) + 1/3*(v53) + 1/3*(v543)$);
        \coordinate (c54) at ($1/3*(v52) + 1/3*(v53) + 1/3*(v544)$);
        
        \node[above left , scale=1.2] at (v51) {$\tau_5$};
        
        \draw[- , thick] (v51) -- (v52) -- (v53) -- cycle;
        \draw[- , thick] (v52) -- (v544) -- (v53);
        \draw[- , thick] (v53) -- (v543) -- (v51);
        \draw[- , thick] (v51) -- (v542) -- (v52);
        
        \node[olive] at (c51) {$t_{5;3}$};
        \node[violet] at (c52) {$t_{5;1}$};
        \node[cyan] at (c53) {$t_{5;4}$};
        \node[teal] at (c54) {$t_{5;2}$};
        
        \draw[decoration={markings,mark=at position 0.6 with {\arrow[scale=1.5,thick,>=stealth]{>}}},postaction={decorate}] (v52) -- (v542);
        \draw[decoration={markings,mark=at position 0.6 with {\arrow[scale=1.5,thick,>=stealth]{>}}},postaction={decorate}] (v52) -- (v51);
        \draw[decoration={markings,mark=at position 0.6 with {\arrow[scale=1.5,thick,>=stealth]{>}}},postaction={decorate}] (v52) -- (v53);
        \draw[decoration={markings,mark=at position 0.6 with {\arrow[scale=1.5,thick,>=stealth]{>}}},postaction={decorate}] (v52) -- (v544);
        \draw[decoration={markings,mark=at position 0.6 with {\arrow[scale=1.5,thick,>=stealth]{>}}},postaction={decorate}] (v51) -- (v542);
        \draw[decoration={markings,mark=at position 0.6 with {\arrow[scale=1.5,thick,>=stealth]{>}}},postaction={decorate}] (v51) -- (v53);
        \draw[decoration={markings,mark=at position 0.6 with {\arrow[scale=1.5,thick,>=stealth]{>}}},postaction={decorate}] (v544) -- (v53);
        \draw[decoration={markings,mark=at position 0.6 with {\arrow[scale=1.5,thick,>=stealth]{>}}},postaction={decorate}] (v543) -- (v53);
        \draw[decoration={markings,mark=at position 0.6 with {\arrow[scale=1.5,thick,>=stealth]{>}}},postaction={decorate}] (v51) -- (v543);

        \draw[- , double distance=2pt , thick , dotted , blue] (v142) -- (v343);
        \draw[- , double distance=2pt , thick , dotted , green] (v143) -- (v442);
        \draw[- , double distance=2pt , thick , dotted , red] (v144) -- (v244);
        \draw[- , double distance=2pt , thick , dotted , magenta] (v242) to [out=30 , in=30] (v342);
        \draw[- , double distance=2pt , thick , dotted , orange] (v444) to [out=150 , in=150] (v344);
        \draw[- , double distance=2pt , thick , dotted , brown] (v443) to [out=-90 , in=180] (1,-15) to [out=0 , in=-90] (v243);
        \draw[- , double distance=2pt , thick , dotted , violet] (v542) -- (v13);
        \draw[- , double distance=2pt , thick , dotted , teal] (v544) to [out=-30 , in=-150] (v21);
        \draw[- , double distance=2pt , thick , dotted , cyan] (v543) to [out=-150 , in=-30] (v42);
        
        \draw[- , double distance=2pt , thick , dotted , olive] (v33) to [out=90 , in=90] ($ (v33) + (9,0) $) -- (10,-15) to [out=-90 , in=-90] (1,-15.5);
        \draw[- , double distance=2pt , thick , dotted , olive] (1,-14.5) -- (v53);

        \draw[fill] (v12) circle [radius=0.1];
        \draw[thick , fill=white] (v11) circle [radius=0.1];
        \draw[thick , fill=gray] (v13) circle [radius=0.1];
        \draw[thick , fill=lime] (v142) circle [radius=0.1];
        \draw[thick , fill=lime] (v143) circle [radius=0.1];
        \draw[thick , fill=lime] (v144) circle [radius=0.1];
        
        \draw[fill] (v23) circle [radius=0.1];
        \draw[thick , fill=gray] (v21) circle [radius=0.1];
        \draw[thick , fill=yellow] (v22) circle [radius=0.1];
        \draw[thick , fill=lime] (v242) circle [radius=0.1];
        \draw[thick , fill=lime] (v243) circle [radius=0.1];
        \draw[thick , fill=lime] (v244) circle [radius=0.1];
        
        \draw[fill] (v32) circle [radius=0.1];
        \draw[thick , fill=white] (v31) circle [radius=0.1];
        \draw[thick , fill=yellow] (v33) circle [radius=0.1];
        \draw[thick , fill=lime] (v342) circle [radius=0.1];
        \draw[thick , fill=lime] (v343) circle [radius=0.1];
        \draw[thick , fill=lime] (v344) circle [radius=0.1];
        
        \draw[thick , fill=white] (v43) circle [radius=0.1];
        \draw[thick , fill=yellow] (v41) circle [radius=0.1];
        \draw[thick , fill=gray] (v42) circle [radius=0.1];
        \draw[thick , fill=lime] (v442) circle [radius=0.1];
        \draw[thick , fill=lime] (v443) circle [radius=0.1];
        \draw[thick , fill=lime] (v444) circle [radius=0.1];
        
        \draw[fill] (v52) circle [radius=0.1];
        \draw[thick , fill=white] (v51) circle [radius=0.1];
        \draw[thick , fill=yellow] (v53) circle [radius=0.1];
        \draw[thick , fill=gray] (v542) circle [radius=0.1];
        \draw[thick , fill=gray] (v543) circle [radius=0.1];
        \draw[thick , fill=gray] (v544) circle [radius=0.1];

        \draw[fill] ($(v343) + (-4,1)$) circle [radius=0.1];
        \node[right] at ($(v343) + (-4,1)$) {$:v$};
        
        \draw[fill=white] ($(v343) + (-4,0)$) circle [radius=0.1];
        \node[right] at ($(v343) + (-4,0)$) {$:v'$};
        
        \draw[fill=gray] ($(v343) + (-4,-1)$) circle [radius=0.1];
        \node[right] at ($(v343) + (-4,-1)$) {$:v''$};
        
        \draw[thick , fill=yellow] ($(v343) + (-4,-2)$) circle [radius=0.1];
        \node[right] at ($(v343) + (-4,-2)$) {$:v'''$};
        
        \draw[thick , fill=lime] ($(v343) + (-4,-3)$) circle [radius=0.1];
        \node[right] at ($(v343) + (-4,-3)$) {$:v''''$};
    \end{tikzpicture}
    \caption{4-simplex boundary construction: five tetrahedra share five vertices. Each of the four faces of each tetrahedron is identified with one of the faces of the other four tetrahedra. We use the same color and a double dotted line for the identified faces.}
    \label{Fig:4-simplex}
\end{figure} 

In order to fuse the tetrahedra phase spaces, we use the \textit{Link Gluing} of Sec. \ref{Sec_GluingRules}; according to the Fig. \ref{Fig:4-simplex}, we denote $t_{AB}$ the face shared by tetrahedra $A$ and $B$:
\be
    \begin{aligned}
        &
        t_{12} \equiv t_{1;1} = - t_{2;4} 
        \,\,,\quad
        t_{13} \equiv t_{1;2} = - t_{3;3}
        \,\,,\quad
        t_{14} \equiv t_{1;3} = - t_{4;2}
        \,\,,\quad
        t_{15} \equiv t_{1;4} = - t_{5;1}
        \,\,,\quad
        t_{23} \equiv t_{2;1} = - t_{3;4}
        \,\,,\\
        &
        t_{24} \equiv t_{2;2} = - t_{4;3}
        \,\,,\quad
        t_{25} \equiv t_{2;3} = - t_{5;2}
        \,\,,\quad
        t_{34} \equiv t_{3;1} = - t_{4;4}
        \,\,,\quad
        t_{35} \equiv t_{3;2} = - t_{5;3}
        \,\,,\quad
        t_{45} \equiv t_{4;1} = - t_{5;4} .
    \end{aligned}
\ee
We enforce these geometric conditions on the variables $\obe_{t_{A;a}}$ through the ten momentum maps
\be
    \begin{aligned}
        &
        \obe_{12} \equiv \obe_{t_{1;1}} = - (u_{12} \, \obe_{t_{2;4}} \, u_{12}\mone)|_{G_1^*}
        \,\,,\quad
        \obe_{31} \equiv \obe_{t_{3;3}} = - (u_{13}\mone \, \obe_{t_{1;2}} \, u_{13})|_{G_1^*}
        \,\,,\quad
        \obe_{14} \equiv \obe_{t_{1;3}} = - (u_{14} \, \obe_{t_{4;2}} \, u_{14}\mone)|_{G_1^*}
        \,\,,\\
        &
        \obe_{51} \equiv \obe_{t_{5;1}} = - (u_{15}\mone \, \obe_{t_{1;4}} \, u_{15})|_{G_1^*}
        \,\,,\quad
        \obe_{23} \equiv \obe_{t_{2;1}} = - (u_{23} \, \obe_{t_{3;4}} \, u_{23}\mone)|_{G_1^*}
        \,\,,\quad
        \obe_{42} \equiv \obe_{t_{4;3}} = - (u_{24}\mone \, \obe_{t_{2;2}} \, u_{24})|_{G_1^*}
        \,\,,\\
        &
        \obe_{25} \equiv \obe_{t_{2;3}} = - (u_{25} \, \obe_{t_{5;2}} \, u_{25}\mone)|_{G_1^*}
        \,\,,\quad
        \obe_{34} \equiv \obe_{t_{3;1}} = - (u_{34} \, \obe_{t_{4;4}} \, u_{34}\mone)|_{G_1^*}
        \,\,,\quad
        \obe_{53} \equiv \obe_{t_{5;3}} = - (u_{35}\mone \, \obe_{t_{3;2}} \, u_{35})|_{G_1^*}
        \,\,,\\
        &
        \obe_{45} \equiv \obe_{t_{4;1}} = - (u_{45} \, \obe_{t_{5;4}} \, u_{45}\mone)|_{G_1^*} ,
    \end{aligned}
    \label{MomentumMaps_Sphere_Faces}
\ee
where we denoted $\obe_{AB}$ the decoration of the fused face shared by tetrahedra $A$ and $B$, rooted at the center of tetrahedron $A$. Using a similar notation, for later purpose, it is useful to introduce the variable $b_{AB}$:
\be
    \begin{aligned}
        &
        b_{12} \equiv b_{1;1} = - (u_{12} \, b_{2;4} \, u_{12}\mone)|_{G_1^*}
        \,\,,\quad
        b_{31} \equiv b_{3;3} = - (u_{13}\mone \, b_{1;2} \, u_{13})|_{G_1^*}
        \,\,,\quad
        b_{14} \equiv b_{1;3} = - (u_{14} \, b_{4;2} \, u_{14}\mone)|_{G_1^*}
        \,\,,\\
        &
        b_{51} \equiv b_{5;1} = - (u_{15}\mone \, b_{1;4} \, u_{15})|_{G_1^*}
        \,\,,\quad
        b_{23} \equiv b_{2;1} = - (u_{23} \, b_{3;4} \, u_{23}\mone)|_{G_1^*}
        \,\,,\quad
        b_{42} \equiv b_{4;3} = - (u_{24}\mone \, b_{2;2} \, u_{24})|_{G_1^*}
        \,\,,\\
        &
        b_{25} \equiv b_{2;3} = - (u_{25} \, b_{5;2} \, u_{25}\mone)|_{G_1^*}
        \,\,,\quad
        b_{34} \equiv b_{3;1} = - (u_{34} \, b_{4;4} \, u_{34}\mone)|_{G_1^*}
        \,\,,\quad
        b_{53} \equiv b_{5;3} = - (u_{35}\mone \, b_{3;2} \, u_{35})|_{G_1^*}
        \,\,,\\
        &
        b_{45} \equiv b_{4;1} = - (u_{45} \, b_{5;4} \, u_{45}\mone)|_{G_1^*} .
    \end{aligned}
\ee
According to the \textit{Link Gluing} in Sec. \ref{Sec_GluingRules}, for each o these momentum maps we have a fused link decoration
\be
    \begin{aligned}
        &
        u_{12} = u_{1;1} u_{2;4}\mone
        \,\,,\quad
        u_{13} = u_{1;2} u_{3;3}\mone
        \,\,,\quad
        u_{14} = u_{1;3} u_{4;2}\mone
        \,\,,\quad
        u_{15} = u_{1;4} u_{5;1}\mone
        \,\,,\quad
        u_{23} = u_{2;1} u_{3;4}\mone
        \,\,,\\
        &
        u_{24} = u_{2;2} u_{4;3}\mone
        \,\,,\quad
        u_{25} = u_{2;3} u_{5;2}\mone
        \,\,,\quad
        u_{34} = u_{3;1} u_{4;4}\mone
        \,\,,\quad
        u_{35} = u_{3;2} u_{5;3}\mone
        \,\,,\quad
        u_{45} = u_{4;1} u_{5;4}\mone
        \,\,,
    \end{aligned}
    \label{FusedLinks_Sphere}
\ee
where $u_{AB}$ decorates the fused link with source at the center of tetrahedron $A$ and target at the center of tetrahedron $B$. Last, to complete the gluing of the five tetrahedra, we identify the boundary of the faces, too. 
According to Fig. \ref{Fig:4-simplex}, denote $\lambda_{A;ab}$ the edge shared by faces $a$ and $b$ of tetrahedron $A$, represented at the center of it. Let $e_{ABC}$ be the edge shared by tetrahedra $A$,$B$,$C$ and $\lambda_{ABC}^D$ be the respective edge decoration represented at the node $D$ defined through the following momentum maps
\be
    \begin{aligned}
        &
        \lambda_{123}^1 \equiv 
        \lambda_{1;12} = u_{12} \rhd \lambda_{2;14} = u_{13} \rhd \lambda_{3;34} 
        \,\,,\quad
        \lambda_{124}^1 \equiv 
        \lambda_{1;13} = u_{12} \rhd \lambda_{2;24} = u_{14} \rhd \lambda_{4;23}
        \,\,,\\
        &
        \lambda_{125}^1 \equiv 
        \lambda_{1;14} = u_{12} \rhd \lambda_{2;34} = u_{15} \rhd \lambda_{5;12} 
        \,\,,\quad
        \lambda_{134}^1 \equiv 
        \lambda_{1;23} = u_{13} \rhd \lambda_{3;13} = u_{14} \rhd \lambda_{4;24}
        \,\,,\\
        &
        \lambda_{135}^1 \equiv 
        \lambda_{1;24} = u_{13} \rhd \lambda_{3;23} = u_{15} \rhd \lambda_{5;13} 
        \,\,,\quad
        \lambda_{145}^1 \equiv 
        \lambda_{1;34} = u_{14} \rhd \lambda_{4;12} = u_{15} \rhd \lambda_{5;14}
        \,\,,\\
        &
        \lambda_{234}^2 \equiv 
        \lambda_{2;12} = u_{23} \rhd \lambda_{3;14} = u_{24} \rhd \lambda_{4;34} 
        \,\,,\quad
        \lambda_{235}^2 \equiv 
        \lambda_{2;13} = u_{23} \rhd \lambda_{3;24} = u_{25} \rhd \lambda_{5;23}
        \,\,,\\
        &
        \lambda_{245}^2 \equiv 
        \lambda_{2;23} = u_{24} \rhd \lambda_{4;13} = u_{25} \rhd \lambda_{5;24} 
        \,\,,\quad
        \lambda_{345}^3 \equiv 
        \lambda_{3;12} = u_{34} \rhd \lambda_{4;14} = u_{35} \rhd \lambda_{5;34}
        \,\,.
    \end{aligned}
    \label{MomentumMaps_Sphere_Edges}
\ee
From the \textit{Wedge Gluing Part 1} of Sec. \ref{Sec_GluingRules}, we have a fused wedge decoration for each momentum map
\be
    \begin{aligned}
        \oy_{123}^1 & =
        \oy_{1;12} + u_{12} \rhd \oy_{2;14} + u_{13} \rhd \oy_{3;34} +
        (u_{12} \, \obe_{23} \, u_{12}\mone)|_{G_2^*} +
        (u_{13} \, \obe_{31} \, u_{13}\mone)|_{G_2^*}
        \,,\\
        \oy_{124}^1 & =
        \oy_{1;13} + u_{12} \rhd \oy_{2;24} + u_{14} \rhd \oy_{4;23} +
        (u_{14} \, \obe_{42} \, u_{14}\mone)|_{G_2^*}
        \,,\\
        \oy_{125}^1 & =
        \oy_{1;14} + u_{12} \rhd \oy_{2;34} + u_{15} \rhd \oy_{5;12} +
        (u_{12} \, \obe_{25} \, u_{12}\mone)|_{G_2^*} +
        (u_{15} \, \obe_{51} \, u_{15}\mone)|_{G_2^*} 
        \,,\\
        \oy_{134}^1 & =
        \oy_{1;23} + u_{13} \rhd \oy_{3;14} + u_{14} \rhd \oy_{4;24} +
        (u_{13} \, \obe_{34} \, u_{13}\mone)|_{G_2^*} +
        (u_{13} \, \obe_{31} \, u_{13}\mone)|_{G_2^*}
        \,,\\
        \oy_{135}^3 & =
        \oy_{1;24} + u_{13} \rhd \oy_{3;23} + u_{15} \rhd \oy_{5;13} +
        (u_{35} \, \obe_{51} \, u_{35}\mone)|_{G_2^*} +
        (u_{35} \, \obe_{53} \, u_{35}\mone)|_{G_2^*}
        \,,\\
        \oy_{145}^1 & =
        \oy_{1;34} + u_{14} \rhd \oy_{4;12} + u_{15} \rhd \oy_{5;14} +
        (u_{14} \, \obe_{45} \, u_{14}\mone)|_{G_2^*} +
        (u_{15} \, \obe_{51} \, u_{15}\mone)|_{G_2^*}
        \,,\\
        \oy_{235}^2 & =
        \oy_{2;13} + u_{23} \rhd \oy_{3;24} + u_{25} \rhd \oy_{5;23} +
        (u_{25} \, \obe_{53} \, u_{25}\mone)|_{G_2^*}
        \,,\\
        \oy_{245}^2 & =
        \oy_{2;23} + u_{24} \rhd \oy_{4;13} + u_{25} \rhd \oy_{5;24} +
        (u_{24} \, \obe_{45} \, u_{24}\mone)|_{G_2^*} +
        (u_{24} \, \obe_{42} \, u_{24}\mone)|_{G_2^*}
        \,,\\
        \oy_{234}^2 & =
        \oy_{2;12} + u_{23} \rhd \oy_{3;14} + u_{24} \rhd \oy_{4;34} +
        (u_{23} \, \obe_{34} \, u_{23}\mone)|_{G_2^*} +
        (u_{24} \, \obe_{42} \, u_{24}\mone)|_{G_2^*}
        \,,\\
        \oy_{345}^3 & =
        \oy_{3;12} + u_{34} \rhd \oy_{4;14} + u_{35} \rhd \oy_{5;34} +
        (u_{34} \, \obe_{45} \, u_{34}\mone)|_{G_2^*} +
        (u_{35} \, \obe_{53} \, u_{35}\mone)|_{G_2^*} ,
    \end{aligned}
    \label{ClosedFusedWedges}
\ee
where the variable $\oy_{ABC}^D$ decorates the fused wedge dual to the edge $e_{ABC}$ (closed surface in the dual complex shared by tetrahedra $A$,$B$, $C$), it is rooted at the center of tetrahedron $D$. 

\medskip

The symplectic reduction for a 4-simplex boundary  phase space is thus
\begin{align}
    \cP_{\tau^4} 
    & =
    \cB_{l,e}^{\times 60} /\!/ \big(\big(((G_1^* \times G_1^*) \times G_2^*)^{\times 4} \times ((G_2^*)^{\times 6} \times G_1)\big)^{\times 5} \times ((G_1)^{\times 10} \times (G_2^*)^{\times 10}) \big)
    \nonumber \\
    & =
    \cP_{\tau}^{\times 5} /\!/ \big((G_1)^{\times 10} \times (G_2^*)^{\times 10} \big) ,
\end{align}
with momentum maps \eqref{MomentumMaps_Sphere_Faces} and \eqref{MomentumMaps_Sphere_Edges}.
Using \eqref{1-Gauss_FusedWedges}, the 1-Gauss constraints of the five tetrahedra, each of them represented at the center of the respective tetrahedron, are
\begin{align}
    \label{1-Gauss_Tetra1}
    \tau_1 : \quad &
    b_{12} - (u_{13} \, b_{31} \, u_{13}\mone)|_{G_1^*} + \lambda_{135}^1 \rhd b_{14} - (u_{15} \, b_{51} \, u_{15}\mone)|_{G_1^*}
    \nonumber \\
    & \quad +
    (\lambda_{123}^1 \, \oy_{1;12} \, (\lambda_{123}^1)\mone)|_{G_1^*} +
    \lambda_{125}^1 \rhd (\lambda_{124}^1 \, \oy_{1;13} \, (\lambda_{124}^1)\mone)|_{G_1^*} +
    (\lambda_{125}^1 \, \oy_{1;14} \, (\lambda_{125}^1)\mone)|_{G_1^*} 
    \nonumber \\
    & \quad +
    \lambda_{135}^1 \rhd (\lambda_{134}^1 \, \oy_{1;23} \, (\lambda_{234}^1)\mone)|_{G_1^*} +
    (\lambda_{135}^1 \, \oy_{1;24} \, (\lambda_{135}^1)\mone)|_{G_1^*} +
    \lambda_{135}^1 \rhd (\lambda_{145}^1 \, \oy_{1;34} \, (\lambda_{145}^1)\mone)|_{G_1^*} = 0
    , \\
    \tau_2 : \quad &
    b_{23} - \lambda_{125}^2 \rhd (u_{24} \, b_{42} \, u_{24}\mone)|_{G_1^*} + b_{25} - (u_{12}\mone \, b_{12} \, u_{12})|_{G_1^*} 
    \nonumber \\
    & \quad +
    \lambda_{235}^2 \rhd (\lambda_{234}^2 \, \oy_{2;12} \, (\lambda_{234}^2)\mone)|_{G_1^*} +
    (\lambda_{235}^2 \, \oy_{2;13} \, (\lambda_{235}^2)\mone)|_{G_1^*} +
    (\lambda_{123}^2 \, \oy_{2;14} \, (\lambda_{123}^2)\mone)|_{G_1^*} 
    \nonumber \\
    & \quad +
    \lambda_{125}^2 \rhd (\lambda_{245}^2 \, \oy_{2;23} \, (\lambda_{245}^2)\mone)|_{G_1^*} +
    \lambda_{125}^2 \rhd (\lambda_{124}^2 \, \oy_{2;24} \, (\lambda_{124}^2)\mone)|_{G_1^*} +
    (\lambda_{125}^2 \, \oy_{2;34} \, (\lambda_{125}^2)\mone)|_{G_1^*} = 0
    , \\
    \tau_3 : \quad &
    \lambda_{135}^3 \rhd b_{34} - (u_{35} \, b_{53} \, u_{35}\mone)|_{G_1^*} + b_{31} - (u_{23}\mone \, b_{23} \, u_{23})|_{G_1^*} 
    \nonumber \\
    & \quad +
    \lambda_{135}^3 \rhd (\lambda_{345}^3 \, \oy_{3;12} \, (\lambda_{345}^3)\mone)|_{G_1^*} +
    \lambda_{135}^3 \rhd (\lambda_{134}^3 \, \oy_{3;13} \, (\lambda_{134}^3)\mone)|_{G_1^*} +
    \lambda_{235}^3 \rhd (\lambda_{234}^3 \, \oy_{3;14} \, (\lambda_{234}^3)\mone)|_{G_1^*} 
    \nonumber \\
    & \quad +
    (\lambda_{135}^3 \, \oy_{3;23} \, (\lambda_{135}^3)\mone)|_{G_1^*} +
    (\lambda_{235}^3 \, \oy_{3;24} \, (\lambda_{235}^3)\mone)|_{G_1^*} +
    (\lambda_{123}^3 \, \oy_{3;34} \, (\lambda_{123}^3)\mone)|_{G_1^*} = 0
    , \\
    \tau_4 : \quad &
    \lambda_{135}^4 \rhd \big(b_{45} - (u_{14}\mone \, b_{14} \, u_{14})|_{G_1^*} + \lambda_{145}^4 \rhd b_{42} - (u_{34}\mone \, b_{34} \, u_{34})|_{G_1^*}\big) 
    \nonumber \\
    & \quad +
    \lambda_{135}^4 \rhd (\lambda_{145}^4 \, \oy_{4;12} \, (\lambda_{145}^4)\mone)|_{G_1^*} +
    \lambda_{125}^4 \rhd (\lambda_{245}^4 \, \oy_{4;13} \, (\lambda_{245}^4)\mone)|_{G_1^*} +
    \lambda_{135}^4 \rhd (\lambda_{345}^4 \, \oy_{4;14} \, (\lambda_{345}^4)\mone)|_{G_1^*} 
    \nonumber \\
    & \quad +
    \lambda_{125}^4 \rhd (\lambda_{124}^4 \, \oy_{4;23} \, (\lambda_{124}^4)\mone)|_{G_1^*} +
    \lambda_{135}^4 \rhd (\lambda_{134}^4 \, \oy_{4;24} \, (\lambda_{134}^4)\mone)|_{G_1^*} +
    \lambda_{235}^4 \rhd (\lambda_{234}^4 \, \oy_{4;34} \, (\lambda_{234}^4)\mone)|_{G_1^*} = 0
    , \\
    \tau_5 : \quad &
    b_{51} - (u_{25}\mone \, b_{25} \, u_{25})|_{G_1^*} + b_{53} - \lambda_{135}^5 \rhd (u_{45}\mone \, b_{45} \, u_{45})|_{G_1^*}
    \nonumber \\
    & \quad +
    (\lambda_{125}^5 \, \oy_{5;12} \, (\lambda_{125}^5)\mone)|_{G_1^*} +
    (\lambda_{135}^5 \, \oy_{5;13} \, (\lambda_{135}^5)\mone)|_{G_1^*} +
    \lambda_{135}^5 \rhd (\lambda_{145}^5 \, \oy_{5;14} \, (\lambda_{145}^5)\mone)|_{G_1^*} 
    \nonumber \\
    & \quad +
    (\lambda_{235}^5 \, \oy_{5;23} \, (\lambda_{235}^5)\mone)|_{G_1^*} +
    \lambda_{125}^5 \rhd (\lambda_{245}^5 \, \oy_{5;24} \, (\lambda_{245}^5)\mone)|_{G_1^*} +
    \lambda_{135}^5 \rhd (\lambda_{345}^5 \, \oy_{5;34} \, (\lambda_{345}^5)\mone)|_{G_1^*} = 0 .
    \label{1-Gauss_Tetra5}
\end{align}

\subsubsection{Recovering the  un-deformed case.} In this section, we consider the explicit example of the Poincar\'e 2-group with a trivial Poisson structure. This is to recover the phase space that was introduced in \cite{Asante:2019lki}  and recovered again in \cite{peter}. The most difficult constraints to recover are the 1-Gauss constraints since they are non-trivial.

Let $\cG = \ISO(4)\cong \SO(4)\ltimes \R^4\cong \R^4\rtimes\SO(4)$ and $\cG^*=\ISO(4)^*\cong \SO^*(4)\times \R^4\cong \R^4\times\SO^*(4)$. Let $\cT$ be a triangulation with edges and faces decorated by elements of $\bbR^4$ and $\SO^*(4)$ respectively. The dual 2-complex $\cT^*$ has links and wedges decorated by holonomies of $\SO(4)$ and $\bbR^{4*} \cong \bbR^4$ respectively. 
Specifically, we have 
\begin{align}
    G_1 = \SO(4) \ni u 
    \,,\quad 
    G^*_2 = \bbR^4 \ni y, \oy
    \,,\quad 
    G_2 = \bbR^4 \ni \lambda, \tla
    \,,\quad 
    G_1^* = \SO^*(4)\cong \bbR^6 \ni \beta, \tbe.
\end{align}
Note that there is no action of $\bbR^4$ on $\SO(4)$ nor on $\SO^*(4)$ and the conjugations \eqref{conjug} take a simpler shape
\be
    \begin{aligned}
        u \beta u\mone
        = \beta' y'
        & \,\,\, \Leftrightarrow \,\,
        \begin{cases}
            & \beta' =  (u \beta u\mone)|_{SO^*(4)}  \\
            & y' = (u \beta u\mone)|_{\bbR^4}=1 
        \end{cases}
        \\
        \lambda y \lambda\mone 
        = y'' \beta'' 
        & \,\,\, \Leftrightarrow \,\,
        \begin{cases}
            & \beta' = (\lambda y \lambda\mone)|_{SO^*(4)} \equiv [\lambda , y]  \\
            & y' = (\lambda y \lambda\mone)|_{\bbR^4} = y ,
        \end{cases}
    \end{aligned}
\ee
where for the second conjugation we used a convenient representation\footnote{If $X_\mu$ is a Lie algebra generator of $\bbR^4$, we use a representation such that $X_\mu X_\nu=0$. This implements that $e^X= 1+ X$. If $Y_\nu$ are the generators of another $\bbR^4$, then $e^Y e^X e^{-Y}= 1+X+ [Y,X]$. } of $\bbR^4$. 
\medskip 
First, to construct the triangle phase space we impose the momentum maps $u_1=\tu_2$ and $u_2=\tu_3$, plus the closure constraint that here takes the form $\lambda_1 + \lambda_2 + \lambda_3 = 0$. The fused face decoration \eqref{Triangle_TotalFace} is thus
\be
    \beta_t = \beta_1 + \beta_2 + \beta_3 + \beta_4 
    - [\lambda_2 \,,\, \oy_1] - [\lambda_3 \,,\, \oy_1] - [\lambda_3 \,,\, \oy_2] .
\ee
Similarly, to construct the tetrahedron phase space we first consider four triangle phase spaces and fuse them imposing the six momentum maps \eqref{MomentumMaps_Tetra_Edges}, 
\be
    \lambda_{1;1} = -\lambda_{3;3}
    \,\,,\quad
    \lambda_{1;2} = -\lambda_{4;3}
    \,\,,\quad
    \lambda_{1;3} = -\lambda_{2;1}
    \,\,,\quad
    \lambda_{2;2} = -\lambda_{4;2}
    \,\,,\quad
    \lambda_{2;3} = -\lambda_{3;1}
    \,\,,\quad
    \lambda_{3;2} = -\lambda_{4;1} .
\ee
From these momentum maps we then derive the six fused wedge decorations \eqref{Tetra_FusedWedges} represented at the center of the tetrahedron:
\be
    \begin{aligned}
        &
        \oy_{12} = - \oy_{2;1} + \oy_{1;3} 
        \,\,,\quad
        \oy_{13} = \oy_{1;1} - \oy_{3;3} 
        \,\,,\quad
        \oy_{14} = \oy_{1;2} - \oy_{4;3} 
        \,\,,\\
        &
        \oy_{23} = - \oy_{3;1} + \oy_{2;3} 
        \,\,,\quad
        \oy_{24} = \oy_{2;2} - \oy_{4;2} 
        \,\,,\quad
        \oy_{34} = \oy_{3;2} - \oy_{4;1} .
    \end{aligned}
\ee
The last momentum map that we impose for the construction of the tetrahedron is the 1-Gauss constraint \eqref{1-Gauss_FusedWedges}, that now reduces to
\be
    b_1 + b_2 + b_3 + b_4 + 
    [\lambda_{12} \,,\, \oy_{12}] +
    [\lambda_{13} \,,\, \oy_{13}] +
    [\lambda_{14} \,,\, \oy_{14}] +
    [\lambda_{23} \,,\, \oy_{23}] +
    [\lambda_{24} \,,\, \oy_{24}] +
    [\lambda_{34} \,,\, \oy_{34}] = 0 ,
\ee
where 
\be
    \begin{aligned}
        b_1 & =
        \beta_{1;1} + \beta_{1;2} + \beta_{1;3} -
        [\lambda_{21} \,,\, \oy_{1;3}] + [\lambda_{13} \,,\, \oy_{1;2}]
        , \\
        b_2 & = 
        - \beta_{2;1} + \beta_{2;2} + \beta_{2;3} -
        [\lambda_{32} \,,\, \oy_{2;3}] +
        [\lambda_{12} \,,\, \oy_{2;2}]
        , \\
        b_3 & =
        - \beta_{3;1} + \beta_{3;2} - \beta_{3;3} -
        [\lambda_{13} \,,\, \oy_{3;3}] +
        [\lambda_{23} \,,\, \oy_{3;2}]
        , \\
        b_4 & =
        - \beta_{4;1} - \beta_{4;2} - \beta_{4;3} -
        [(\lambda_{13} - \lambda_{14}) \,,\, \oy_{4;3}] -
        [(\lambda_{12} - \lambda_{24}) \,,\, \oy_{4;2}] -
        [(\lambda_{23} - \lambda_{34}) \,,\, \oy_{4;1}] .
    \end{aligned}
\ee
Once we derived the tetrahedron phase space, we can use it to build the 4-simplex boundary phase space.

In order to do it, we first consider five tetrahedra phase spaces and fuse them through the momentum maps \eqref{MomentumMaps_Sphere_Faces} and \eqref{MomentumMaps_Sphere_Edges}, so that we derive the fused link decorations \eqref{FusedLinks_Sphere} and the fused wedge decorations \eqref{ClosedFusedWedges}. Then we consider the following change of variables for the face decorations
\be
    \begin{aligned}
        b_{12}' = b_{12}
        & -
        u_{12} \, [\lambda_{123}^2 \,,\, \oy_{2;14}] \, u_{12}\mone -
        u_{12} \, [\lambda_{124}^2 \,,\, \oy_{2;24}] \, u_{12}\mone -
        u_{12} \, [\lambda_{125}^2 \,,\, \oy_{2;34}] \, u_{12}\mone
        \, , \\
        b_{31}' = b_{31} 
        & +
        [\lambda_{134}^3 \,,\, \oy_{3;13}] +
        [\lambda_{123}^3 \,,\, \oy_{3;34}] -
        u_{13}\mone \, [\lambda_{135}^1 \,,\, \oy_{1;24}] \, u_{13}
        \, , \\
        b_{14}' = b_{14}
        & -
        u_{14} \, [\lambda_{145}^4 \,,\, \oy_{4;12}] \, u_{14}\mone -
        u_{14} \, [\lambda_{124}^4 \,,\, \oy_{4;23}] \, u_{14}\mone -
        u_{14} \, [\lambda_{134}^4 \,,\, \oy_{4;24}] \, u_{14}\mone
        \, , \\
        b_{51}' = b_{51}
        & +
        [\lambda_{125}^5 \,,\, \oy_{5;12}] +
        [\lambda_{145}^5 \,,\, \oy_{5;14}]
        \, , \\
        b_{23}' = b_{23}
        & -
        u_{23} \, [\lambda_{234}^3 \,,\, \oy_{3;14}] \, u_{23}\mone -
        u_{23} \, [\lambda_{235}^3 \,,\, \oy_{3;24}] \, u_{23}\mone
        \, , \\
        b_{42}' = b_{42}
        & +
        [\lambda_{245}^4 \,,\, \oy_{4;13}] +
        [\lambda_{234}^4 \,,\, \oy_{4;34}]
        \, , \\
        b_{25}' = b_{25}
        & -
        u_{25} \,
        [\lambda_{235}^5 \,,\, \oy_{5;23}] \, u_{25}\mone -
        u_{25} \, [\lambda_{245}^5 \,,\, \oy_{5;24} \, u_{25}\mone
        \, , \\
        b_{34}' = b_{34}
        & -
        u_{34} \, [\lambda_{345}^4 \,,\, \oy_{4;14}] \, u_{34}\mone
        \, , \\
        b_{53}' = b_{53}
        & +
        [\lambda_{135}^5 \,,\, \oy_{5;13}] +
        [\lambda_{345}^5 \,,\, \oy_{5;34}] .
    \end{aligned}
\ee
We use these   new variables in the five 1-Gauss constraints \eqref{1-Gauss_Tetra1}-\eqref{1-Gauss_Tetra5}, in order to express the closure constraints in terms of the closed faces (wedges) variables \eqref{ClosedFusedWedges}.
\begin{align}
    \tau_1 : \quad
    b_{12}' - (u_{13} \, b_{31}' \, u_{13}\mone) + b_{14}' - (u_{15} \, b_{51}' \, u_{15}\mone) 
    & +
    [\lambda_{123}^1 \,,\, \oy_{123}^1] +
    [\lambda_{124}^1 \,,\, \oy_{124}^1] +
    [\lambda_{125}^1 \,,\, \oy_{125}^1] 
    \nonumber \\
    & +
    [\lambda_{134}^1 \,,\, \oy_{134}^1] +
    [\lambda_{145}^1 \,,\, \oy_{145}^1]
    = 0 
    , \\
    \tau_2 : \quad
    b_{23}' - (u_{24} \, b_{42}' \, u_{24}\mone) + b_{25}' - (u_{12}\mone \, b_{12}' \, u_{12}) 
    & +
    [\lambda_{235}^2 \,,\, \oy_{235}^2] +
    [\lambda_{245}^2 \,,\, \oy_{245}^2] +
    [\lambda_{234}^2 \,,\, \oy_{234}^2]
    = 0
    , \\
    \tau_3 : \quad
     b_{34}' - (u_{35} \, b_{53}' \, u_{35}\mone) + b_{31}' - (u_{23}\mone \, b_{23}' \, u_{23}) 
     & +
    [\lambda_{135}^3 \,,\, \oy_{135}^3] +
    [\lambda_{345}^3 \,,\, \oy_{345}^3]
    = 0
    , \\
    \tau_4 : \quad
    b_{45} - (u_{14}\mone \, b_{14}' \, u_{14}) + b_{42}' - (u_{34}\mone \, b_{34}' \, u_{34})
    & = 0
    , \\
    \tau_5 : \quad
    b_{51}' - (u_{25}\mone \, b_{25}' \, u_{25}) + b_{53}' - (u_{45}\mone \, b_{45} \, u_{45})
    & = 0 .
\end{align}
Comparing this result with that of \cite{peter}, we derive the dictionary between the variables here and in \cite{peter} in Table \ref{dic}. 
\begin{table}[h!]
    \centering
    \begin{tabular}{|c|c|c|}
    \hline
          & Our notation & Notation in \cite{peter} \\
         \hline\hline
        Face    (triangle) & $b_{ij}'$  & $b_{(ij)^*}$
         \\
         \hline
         Link & $u_{ij}$ & $h_{ij}$
         \\
         \hline
         Edge $e$ & $\lambda^i_{e^*}$ & $\ell_i^e$\\
         \hline
         Wedge/dual face  & $Y^{i}_{e^*}$ & $V_{i}^{e^*}$\\\hline
    \end{tabular}
    \caption{Dictionary between the variables used here and the ones in \cite{peter}. }
    \label{dic}
\end{table}

Hence our prescription allows to recover exactly  phase space of \cite{Asante:2019lki, peter} in the case of the Euclidean Poisson Lie 2-group with a trivial Poisson structure.

\subsubsection{Other choices of groups $G_i$.} There are a number of other interesting cases that we can highlight. 
\begin{itemize}
\item \textbf{Curved case.} We have described the flat case and we can input curvature in the triangulation, ie dealing with non-abelian groups on the edges,  in different manners. 
\begin{itemize}
    \item \textit{Use the $\kappa$-Poincar\'e deformation}. This would amount to consider the case we described in more details in Sec. \ref{Sec_Semidualization}. We would take 
    $$
    G_1=\SO, \, G_2=\AN, \, G_1^*=\SO^*, \, G_2^*=\AN^*.
    $$
    This type of deformation would work for both sign of the cosmological constant in the Lorentzian case, and for a negative cosmological constant in the Euclidian case. 
    \item \textit{Use a self dual case}. We could start the classical double based on $T^*Spin(4)\cong T^*(\SU(2)\times\SU(2))\cong (\SU(2)\times\SU(2)) \ltimes \bbR^6$ and semi-dualize it, to recover   $(\SU(2)\ltimes\R^3) \bowtie (\SU(2)\ltimes\R^3),$
    which can be seen as the double of the Drinfeld double of $\SU(2)$ (as a Poisson Lie group). 
\end{itemize} 
    \item \textbf{Dual case.}  Instead of using $\cG$ to decorate $\cT$, we use $\cG^*$ and $\cG$ decorates $\cT^*$. This would amount to dualize the cases we considered and can be seen as the analogue of the BF vacuum picture \cite{Dittrich:2016typ}. 
\end{itemize}

\section{Conclusion}

Building a phase space associated to a triangulation is relevant both for topological studies and quantum gravity models. Upon quantization, the phase space  leads to the relevant notion of quantum state of discrete geometry. Spin networks are graphs decorated by (quantum) group representations and encode the naturally quantum states of (curved) discrete geometries in a 2+1 dimensional space. While spin networks can be generalized to the 3+1 dimensional case, they are likely not the best tool to discuss discrete geometries in this case. We expect instead  spin 2-networks based on representations of (quantum) 2-groups to be relevant. This means that we should have some phase space based on 2-groups associated with a triangulation. This is the construction we proposed in Def.\ \ref{Def:main_result}. 

We did not deal with the most general type of 2-groups. Instead we dealt with skeletal strict 2-groups, ie strict 2-groups such that the $t$-map is trivial, such as the Poincar\'e 2-group. We focused on this class of 2-groups as they are similar to groups in many aspects. Because of this similarity, we could recycle many structures such as Heisenberg double and symplectic reduction to achieve our construction. 

The skeletal 2-groups we dealt with will lead upon quantization to quantum 2-groups in general. Indeed, they are in general equipped with a  non-trivial Poisson structure. As an example we showed how the $\kappa$-Poincar\'e deformation arises naturally, when dealing with the Poincar\'e 2-group. But in fact, any deformation (ie non-trivial Poisson structure) of the Poincar\'e group could be used here and exported to the Poincar\'e 2-group setting.   

We gave the explicit construction for the main building blocks of a 3d triangulation, namely the triangle and the tetrahedron. We also recovered the phase space introduced in \cite{Asante:2019lki} based on the un-deformed Poincar\'e 2-group and discussed how they can be generalized to different cases. We highlighted the $\kappa$-Poincar\'e case which can be viewed as defining  polyhedra with curved edges, but still flat face decorations. 

\smallskip

The construction of the triangulation phase space and the construction of  examples of Poisson 2-groups relevant for describing the 3d triangulation phase space  constitute the main results of the paper.

\smallskip

They open many new directions and beg for many questions to be answered.

\paragraph{Generalization, away from skeletal 2-groups.} 
Our construction is based on the fact that skeletal Poisson 2-groups share many similarities with Poisson Lie groups. The first question to explore is  how  can we extend our construction to a non-trivial $t$-map? This would allow to have a non-abelian decorations on the faces. In this case it is unlikely that we can use the group structures and some type of semi-dualization as we did. In a sense, we expect non skeletal 2-groups to be more 2-groups than groups. 

In the  construction we proposed, the edge variables in $G_2$ are always Poisson commutative since dually the wedge decoration in $G^*_2$ is abelian. Switching on curvature on the wedge side ---  making $G^*_2$ non abelian and dealing with a non-trivial $t$-map in the crossed module $\G=G_1\ltimes G_2^*$ --- would generate some Poisson non-commutativity for the edge variables. With the proper choice of group $G^*_2$, such  as $\SU(2)$, this could imply a discrete spectrum for the length operator. This interesting aspect is currently under investigation.   

Another interesting generalization is to consider the case of weak 2-groups \cite{baez:2003}. Finally, to have a full description of the 3d triangulation, it seems that a 3-group structure might exhaust all the possible triangulation decorations, namely, 1d, 2d, and 3d. We expect that the volume decoration would be dually equivalent to considering particle like defects. A different argument, also pointing to the use of 3-groups for matter is \cite{Radenkovic:2019qme}.

\paragraph{Derivation from a classical action.} The phase space defined in terms of the un-deformed Poincar\'e 2-group, proposed in \cite{Asante:2019lki},  can be derived from the \textit{BFCG} action \cite{Girelli:2007tt} defined as a 2-gauge theory for the Poincar\'e 2-group, or from the \textit{BF} action seen as a gauge theory for the Poincar\'e group \cite{peter}. Indeed, such \textit{BFCG} and \textit{BF} actions can be related by a change of polarization in the translation sector only \cite{Mikovic:2011si}. This is the continuum version of the semi-dualization relating the Lie bi-algebras $\d$ and $\b$. It would be interesting to see whether we could derive the other cases of skeletal 2-groups. In particular, we expect that the $\kappa$-Poincar\'e 2-group should arise when dealing with a $\SO(4,1)$ \textit{BF} theory, where we would do a change of polarization in the $\AN$ sector. This would allow to define a \textit{BFCG} type of action of a slightly more complicated shape than the usual one, due to the presence of the mutual actions.  

Another natural example to consider is $\SO(3,1)$  \textit{BF} theory which can be constrained to recover 4d gravity (with no cosmological constant). We conjecture that such theory can be discretized in terms of the 3d $\kappa$-Poincar\'e 2-group. This would provide potentially an interesting new route to define a spinfoam model for 4d gravity, when implementing the simplicity constraints.

\paragraph{Simplicity constraints. }
The polyhedron we defined is not really geometrical since the face decorations are independent of the edge decorations. We would like to have --- on faces --- some data encoding the normal and the area of the face, as done in the Minkowski theorem, while maintaining  the edge information. 
Imposing such geometricity constraint is equivalent to the simplicity constraint which are essential to project the topological sector to the gravity sector. This type of constraint usually breaks the 2-gauge invariance. It would be interesting to see whether we can still define a 2-group where the face decoration, the normal, is defined in terms of the edges. This would potentially  provide an interesting structure to discuss discrete gravity.            
{
\paragraph{Revisiting twisted geometries. }
In the usual loop quantum gravity picture, one uses $T^*\SU(2)$ to decorate the link and triangle data. This phase space can be decomposed in a different manner to have discrete geometry variables, such as the dihedral angle, the normal and the area of the triangle \cite{Freidel:2010aq, Freidel:2013bfa}. This discrete geometry description goes  under the name of \textit{twisted geometries}. We have increased the phase  space dimension to include variables describing the edges (and the dual faces) which also allows to discuss curved geometries. It would be interesting to explore how this can generate a more elaborate version of the twisted geometries, possibly also incorporating curvature.   
}


\paragraph{Quantization: spin 2-networks, 2-representations of deformed Poincar\'e 2-group, state sum. }
Studying the (2-)representation theory of a 2-group is a difficult task in general. The case of skeletal 2-groups was done in \cite{Baez:2008hz}. It would be interesting to see how such structure is deformed when dealing with the deformed Poincar\'e 2-groups. Strict quantum 2-groups have been defined in \cite{Majid:2012gy}, but their representation theory is still an open question to the best of our knowledge. Using similar arguments as in \cite{Asante:2019lki}, one could  discuss the spectrum of the length operator  and assess whether it becomes discrete or not according to the choice of Poisson 2-group.   

Starting from $G$-networks, it was shown how we could recover the KBF amplitude/state sum model \cite{Korepanov:2002tp, Korepanov:2002tq, Korepanov:2002tr, Baratin:2014era}, defined in terms of the Poincar\'e 2-group representations. Interestingly, there was no need to use some type of Plancherel/Peter-Weyl formula to recover it. It would be interesting to see whether we can apply this type calculations to determine the deformed version of the KBF model.

\paragraph{Semi-dualization, relation between spin networks and spin 2-networks?} As we alluded several times now, skeletal 2-groups such as the Poincar\'e 2-group behave in some respect as groups. 
Since the Lie algebra case $\d$ and the skeletal Lie 2-algebra case, $\b$, can be related by a semi-dualization map, we conjecture that the BF amplitude based on spin networks defined for $G_1\bowtie G_2$ could be related to a BFCG type amplitude based on spin 2-networks defined for the skeletal2-group $G_1\ltimes G_2^*$. We leave this for later studies.

\section*{Acknowledgments.}
We would like to thank A. Riello for discussions and suggestions. 

\bibliographystyle{hplain}

\bibliography{biblio}

\end{document}